%% file: TIT_final.tex
\documentclass[journal,10pt,letterpaper]{IEEEtran}

\include{preamble_global}

\usepackage{fancyhdr}
\usepackage{lastpage}
\usepackage{empheq}
\usepackage{cases}
\usepackage{graphicx,subfigure}
\usepackage{cite}
\usepackage[dvipsnames,table]{xcolor}
\usepackage{enumitem}
\usepackage[hidelinks]{hyperref}
\definecolor{cadmiumgreen}{rgb}{0.0, 0.42, 0.24}
\definecolor{darkspringgreen}{rgb}{0.09, 0.45, 0.27}
\definecolor{pimentgreen}{rgb}{0.0, 0.65, 0.31}

\usepackage{pgfplots,pgf}
\pgfplotsset{minor grid style={dotted,gray!25}}
\pgfplotsset{major grid style={dashed,gray!25}}

\usepackage[ruled,vlined,linesnumbered]{algorithm2e}%linesnumbered
\usepackage{array,multirow}
\graphicspath{{fig/}}

\allowdisplaybreaks

\usepackage{tikz}
\usetikzlibrary{calc,fadings,decorations.pathreplacing,spy}

\pgfplotsset{compat=1.12}

\newcommand\pgfmathsinandcos[3]{%
	\pgfmathsetmacro#1{sin(#3)}%
	\pgfmathsetmacro#2{cos(#3)}%
}
\newcommand\LongitudePlane[3][current plane]{%
	\pgfmathsinandcos\sinEl\cosEl{#2} % elevation
	\pgfmathsinandcos\sint\cost{#3} % azimuth
	\tikzset{#1/.style={cm={\cost,\sint*\sinEl,0,\cosEl,(0,0)}}}
}
\newcommand\LatitudePlane[3][current plane]{%
	\pgfmathsinandcos\sinEl\cosEl{#2} % elevation
	\pgfmathsinandcos\sint\cost{#3} % latitude
	\pgfmathsetmacro\yshift{\cosEl*\sint}
	\tikzset{#1/.style={cm={\cost,0,0,\cost*\sinEl,(0,\yshift)}}} %
}
\newcommand\DrawLongitudeCircle[2][1]{
	\LongitudePlane{\angEl}{#2}
	\tikzset{current plane/.prefix style={scale=#1}}
	\pgfmathsetmacro\angVis{atan(sin(#2)*cos(\angEl)/sin(\angEl))} %
	\draw[current plane] (\angVis:1) arc (\angVis:\angVis+180:1);
	\draw[current plane,dashed] (\angVis-180:1) arc (\angVis-180:\angVis:1);
}
\newcommand\DrawLatitudeCircle[2][1]{
	\LatitudePlane{\angEl}{#2}
	\tikzset{current plane/.prefix style={scale=#1}}
	\pgfmathsetmacro\sinVis{sin(#2)/cos(#2)*sin(\angEl)/cos(\angEl)}
	\pgfmathsetmacro\angVis{asin(min(1,max(\sinVis,-1)))}
	\draw[current plane] (\angVis:1) arc (\angVis:-\angVis-180:1);
	\draw[current plane,dashed] (180-\angVis:1) arc (180-\angVis:\angVis:1);
}

\tikzset{%
	>=latex, % option for nice arrows
	inner sep=0pt,%
	outer sep=2pt,%
	mark coordinate/.style={inner sep=0pt,outer sep=0pt,minimum size=3pt,
		fill=black,circle}%
}

\renewcommand{\rv}[1]{{\mathsf{#1}}}
\newcommand{\rvVec}[1]{{\bm{\mathsf{#1}}}}
\newcommand{\rvMat}[1]{{\bm{\mathsf{#1}}}}

\newcommand{\T}{{\scriptscriptstyle\mathsf{T}}}
\renewcommand{\H}{{\scriptscriptstyle\mathsf{H}}}

\newcommand{\of}[1]{^{(#1)}}
\newcommand{\ofH}[1]{^{(#1)\H}}

\newcommand{\Span}[1][]{\ifthenelse{\isempty{#1}}{{\rm Span}}{{\rm Span}\left(#1\right)}}
\newcommand*\dif{\mathop{}\mathrm{d}}
\newcommand{\ind}[1]{{\mathbbm{1}{\{#1\}}}}

\renewcommand{\SNR}{\mathrm{SNR}}

\renewcommand{\defeq}{:=}
\renewcommand{\eqdef}{=:}
\renewcommand{\Id}{\mat{I}}

\newcommand{\Xcal}{\mathcal{X}}
\newcommand{\LLR}{\rv{L}(\Xm\to{\Xm'})}
\newcommand{\meanLLR}{\E\big[\rv{L}(\Xm\to{\Xm'})\big]}
\newcommand{\varLLR}{\Var\big[\rv{L}(\Xm\to{\Xm'})\big]}

\title{Joint Constellation Design for Noncoherent MIMO Multiple-Access Channels} 

\author{	
	Khac-Hoang Ngo,~\IEEEmembership{Member,~IEEE, }%
	Sheng Yang,~\IEEEmembership{Member,~IEEE, }\\%
	Maxime Guillaud,~\IEEEmembership{Senior Member,~IEEE, }%
	Alexis Decurninge,~\IEEEmembership{Member,~IEEE} %
	\thanks{Khac-Hoang Ngo is with Department of Electrical Engineering, Chalmers University of Technology, 41296 Gothenburg, Sweden %He was with Mathematical and Algorithmic Sciences Laboratory, Paris Research Center, Huawei Technologies, 92100 Boulogne-Billancourt, France, and was also with Laboratory of Signals and Systems, CentraleSup\'elec, Paris-Saclay University, 91190 Gif-sur-Yvette, France.
		~(e-mail: \texttt{ngok@chalmers.se}).} 
	\thanks{Sheng Yang is with Laboratory of Signals and Systems, CentraleSup\'elec, Paris-Saclay University, 91190 Gif-sur-Yvette, France (e-mail: \texttt{sheng.yang@centralesupelec.fr}).} 
	\thanks{Maxime Guillaud and Alexis Decurninge are with Advanced Wireless Technology Laboratory, Paris Research Center, Huawei Technologies, 92100 Boulogne-Billancourt, France~(e-mail: \texttt{\{maxime.guillaud, alexis.decurninge\}@huawei.com}).}%
	\thanks{{This article was presented in part at the 2020 IEEE Information Theory Workshop (ITW), Riva del Garda, Italy, April 2021~\cite{Ngo2021} and the 25th International ITG Workshop on Smart Antennas~(WSA), French Riviera, France, November 2021~\cite{HoangWSA2021Riemannian}. The results for the two-user case appeared in part in~\cite{Hoang2020const_2user_MAC}.}}
	\vspace{-.5cm}
}

\mathtoolsset{showonlyrefs}

\newcommand{\sheng}[1]{{#1}}

\begin{document}
	
	\maketitle
	\date{\today}
	\begin{abstract}
		%	\vspace{-.2cm}
		We consider the joint constellation design problem for the noncoherent multiple-input multiple-output multiple-access channel~(MAC). By analyzing the noncoherent maximum-likelihood detection error, we propose novel design criteria so as to minimize the error probability. {As a baseline approach, we adapt several existing design criteria for the point-to-point channel to the MAC. Furthermore, we propose new design criteria. Our first proposed design metric is the dominating term in nonasymptotic lower and upper bounds on the pairwise error probability exponent. We give a geometric interpretation of the bound using Riemannian distance in the manifold of Hermitian positive definite matrices.} %Another proposed criterion} is the minimum expected pairwise log-likelihood ratio 	over the joint constellation{, which follows from the Kullback-Leibler divergence.} 
	From an analysis of this metric at high signal-to-noise ratio, we
	obtain further simplified metrics. For any given set of constellation sizes, the proposed metrics can be optimized over the set of {constellation symbols.
		Motivated by the simplified metric, we propose a simple
		constellation construction consisting in \textit{partitioning} a single-user constellation. We also provide a generalization of our previously proposed construction based on \textit{precoding} individual constellations of lower dimensions.} For a fixed joint constellation, the design metrics can be further optimized over the per-user transmit power, especially when the users transmit at different rates. Considering unitary space-time modulation, we investigate the option of building each individual constellation as a set of truncated unitary matrices scaled by the respective transmit power. Numerical results show that our proposed metrics are meaningful, and can be used as objectives to generate constellations through numerical optimization that perform better, for the same transmission rate and power constraint, than a common pilot-based scheme and the constellations optimized with existing metrics.
\end{abstract}

%\vspace{-.3cm}
\begin{IEEEkeywords}
	% 	\vspace{-.2cm}
	Multiple-input multiple-output~(MIMO), noncoherent communications, multiple-access channel {(MAC),}
	unitary space-time modulation {(USTM),} ML  detector. 
\end{IEEEkeywords}

%-----------------------------------------------
%\vspace{-.6cm}
\section{Introduction} \label{sec:introduction}
%\vspace{-.3cm}
In \sheng{multiple-input multiple-output}~(MIMO) communications, it is
usually assumed that the channel state information~(CSI) is known or
estimated (typically by sending pilots and/or using feedback),
and then used for \sheng{precoding} at the transmitter and/or detection at
the receiver. This is \sheng{known as} the {\em coherent} approach. On
the other hand, in the {\em noncoherent} approach, the transmission and
reception are designed \sheng{without using
	\emph{a priori} knowledge of the
	CSI}~\sheng{\cite{Marzetta1999capacity,Hochwald2000unitaryspacetime,ZhengTse2002Grassman,Yang2013CapacityLargeMIMO,Moser,Hoang2021DoF}}. 
This paper studies the latter approach for the MIMO block-fading
multiple-access channel (MAC), i.e., the channel is assumed to remain unchanged during each coherence block of length $T$ and varies between blocks.

%Channel state information~(CSI) is critical to wireless communication systems %, especially the multiple-antenna systems, 
%over fading channels since it enables to adapt the transmission and reception to the current fading states. % and mitigate the destructive effect of channel variation. 
%%In particular, CSI is used by the receiver for coherent detection of the transmitted signals, and/or by the transmitter to perform spatial multiplexing. 
%It is widely considered in the literature, and used in practice, to estimate the instantaneous CSI by periodically sending pilot symbols known to the receiver. The (normally imperfect) channel estimate is then used for coherent detection in the subsequent interval during which the channel is assumed to be stable~\cite{Hassibi2003howmuchtraining}. %In this way, the channel estimation is separated with the channel coding, transmission, and reception mechanisms, and the latter is operated under the assumption that the CSI is, possibly imperfectly, known. 
%This {\em coherent approach} simplifies the system design by allowing a functional split between channel estimation and data transmission-reception mechanisms, but always dedicates a fraction of communication resources to the former. Furthermore, it was shown that this approach is at a constant performance gap below the full channel capacity~\cite{ZhengTse2002Grassman}.
%%\vspace{-.1cm}
In the single-user case \sheng{with isotropic Rayleigh fading}, a noncoherent approach, so-called unitary space-time modulation~{(USTM)}\cite{Hochwald2000unitaryspacetime}, is to transmit $T\times M$ isotropically distributed and truncated unitary signal matrices, where $M$ is the number of transmit antennas. The subspaces of these matrices belong to the Grassmann manifold $G(\CC^T,M)$, defined as the space of $M$-dimensional subspaces in $\CC^T$~\cite{boothby1986introduction}. Information is carried by the position of the transmitted signal matrix subspace in the manifold. The intuition behind this approach is that the signal subspace is not affected by the random fading coefficients. This approach was shown to be within a vanishing gap from the high-SNR capacity if $T\ge N + \min\{M,N\}$~\cite{Hochwald2000unitaryspacetime,ZhengTse2002Grassman}, and within a constant gap if $2M\le T \le M+N$~\cite{Yang2013CapacityLargeMIMO}, where $N$ is the number of receive antennas. Motivated by this, there has been extensive research on the design of noncoherent constellations as a set of points on the Grassmann manifold. Many of these so-called Grassmannian constellations have been proposed, with a common design criterion of maximizing the minimum pairwise chordal distance between the symbols~\cite{Gohary2009GrassmanianConstellations,Kammoun2007noncoherentCodes,Hoang_cubesplit_journal,Dhillon2008constructing}.

In the multi-user case, {a simple and effective design criterion for noncoherent joint constellation remains unclear.} A straightforward extension of the single-user
coherent approach is 
%the time division multiple access~(TDMA)
%strategy, i.e., only one user is active at a time.}
%\sheng{Another} straightforward extension is 
to divide the coherence block into two parts:
1)~a training part in which orthogonal pilot sequences are sent to estimate the
CSI for each user, and 2)~a data transmission part in which different
users communicate in a nonorthogonal fashion~\cite{Murugesan2007optimization}.
Although this approach achieves the optimal degree-of-freedom (DoF)
region in the two-user single-input multiple-output~(SIMO) \sheng{MAC}~\cite{Hoang2018DoF_MAC}, its optimality in
terms of achievable rate and detection error probability \sheng{remains unclear}. An amplitude-based encoding scheme was proposed in~\cite{Manolakos2016noncoh_Energybased_massiveSIMO}, but the accompanying energy detector relies on a large number of receive antennas so that the average received power across all antennas concentrates. Also with massive receive antenna array, some differential encoding schemes were investigated based on phase shift keying~(PSK)~\cite{Schenk2013noncoh_massiveMIMO,Baeza2018noncoh_SIMO_MAC_DPSK_BICM} or quadrature amplitude modulation~(QAM)~\cite{Kong2016differential}. A joint constellation can also be built on PSK constellations which are absolutely additively uniquely decomposable, i.e., each individual PSK symbol can be uniquely decoded from any linear combination of two PSK constellation points with positive weights~\cite{Li2018noncoherent_massive_STBC,Yu2019design_NOMA}. %Yu2019design_NOMA 
In this scheme, the signal unique decodability relies on the asymptotic orthogonality between the users' channels when the number of antennas is large. A similar uniquely decomposable property was also exploited for QAM-based multi-user space-time modulation~\cite{Chen2020design_NOMA_massiveMIMO}. 
{In~\cite{HoangAsilomar2018multipleAccess}, we proposed a precoding-based multiple-access scheme for the SIMO MAC.} % is proposed, but
%offers no performance guarantee.

%In this work, we consider the MIMO MAC and aim to derive simple and effective joint constellation construction criteria so as to minimize the joint maximum
%likelihood~(ML) symbol detection error.
%\subsubsection*{Contributions} 
In this work, we consider a $K$-user MIMO MAC with Rayleigh block fading with coherence time $T \ge 2$ where user $k$ is equipped with $M_k$ antennas and the receiver with $N$ antennas. We aim to derive simple and effective joint constellation construction criteria so as to minimize the joint maximum
likelihood~(ML) symbol detection error. 
%Let user $k$ transmit from a finite discrete constellation $\Xc_k$ with equally likely $T\times M_k$ matrix-valued symbols. These individual constellations $\Xc_k, k\in [K]$, can have different average symbol power. $\Xcal := \{ [\Xm_1 \ \Xm_2 \ \dots \ \Xm_K]:\; \Xm_k\in\Xcal_k \}$ is called the {\em joint} constellation. 
If the users could cooperate, the system could be seen as a $\big(\sum_{k=1}^{K} M_k\big) \times N$ MIMO point-to-point channel, for which {USTM} is optimal, or near-optimal, in the high-SNR regime. 
%input is uniformly distributed on the Grassmannian $G\big(\CC^T,\sum_{k=1}^{K}M_k\big)$~\cite{ZhengTse2002Grassman}. 
Inspired by this observation, the joint constellation can be treated as a Grassmannian constellation on $G\big(\CC^T,\sum_{k=1}^{K}M_k\big)$, %\footnote{As long as $\Xm_1,\dots,\Xm_K$ are linearly independent, $\Span[[\Xm_1 \dots \Xm_K]]$ is a $\sum_{k=1}^{K}M_k$-dimensional subspace of $\CC^T$, which is an element of $G\big(\CC^T,\sum_{k=1}^{K}M_k\big)$.} 
which leads to a design criterion mimicking the max-min chordal distance criterion.
%\begin{align} \label{eq:criterion_maxminChordal}
%\max_{\Xc} \min_{\Xm \ne \Xm' \in \Xc} \sqrt{\textstyle\sum_{k=1}^{K}M_k - \trace[\Xm^\H \Xm'{\Xm'}^\H\Xm]}.
%%\widehat{\Cc} \defeq \left\{\bigoplus_{k=1}^K {\rm span}(\rvVec{x}_k): \{\rvVec{x}_1,\dots,\rvVec{x}_K\} \in \prod_{k=1}^K \Cc_k\right\},
%\end{align} 
%This criterion was the motivation for the design of the precoders in~\cite{HoangAsilomar2018multipleAccess}.
Brehler and Varanasi derived the error probability of the ML detector for the MIMO MAC in~\cite{Brehler2001noncoherent} and analyzed the high-SNR asymptotic regime. With cooperating users, this analysis led to a design criterion similar to that for a single-user MIMO channel proposed in~\cite[Eq.~(8)]{McCloudIT2002signalDesignAndConvCode}. However, for noncooperating users (as we consider here), using the same criterion would be suboptimal. %no explicit constellation design criterion  was given.
%\begin{align} \label{eq:criterion_minsumDet}
%\max_{\Xc} \min_{\Xm \ne \Xm' \in \Xc} \sum_{\Xm\ne \Xm'\in \Xc} \frac{1}{\det[\Id - \Xm^\H \Xm'{\Xm'}^\H\Xm]^N}.
%\end{align}
%This criterion can be used for the non-cooperating users case by modifying the optimization space. It guarantees the full diversity order in the single-user case~\cite{McCloudIT2002signalDesignAndConvCode}, but depends on the number $N$ of receive antennas. 
The joint ML pairwise error exponent can be shown {to be related to} the Kullback-Leibler~(KL) divergence between the output distributions conditioned on either of the transmitted symbols~\cite{BorranTIT2003nonCoherentKLdistance}. Based on this analysis, a criterion consisting in maximizing the minimum KL divergence was proposed in~\cite{Chen2020design_NOMA_massiveMIMO}, but was used only to optimize the transmit powers and the sub-constellation assignment.

\subsubsection*{Contributions}
Following the approach of \cite{Brehler2001noncoherent}, we analyze the
worst-case pairwise error probability~{(PEP) of the ML detector} and introduce new constellation design metrics. {First, since the exact closed-form expression of the PEP is hard to optimize, we resort to maximizing a lower bound of the worst-case PEP exponent. Then, to reduce the complexity of the constellation optimization, we further simplify the metric, and propose simple constructions inspired by the simplified metric.}
Our contributions are summarized as follows.
\begin{itemize}
\item By analyzing the PEP exponent, we propose a constellation design metric for the MIMO MAC which is the dominating term in nonasymptotic lower and upper lower bounds on the worst-case PEP exponent. The lower bound is obtained via the Chernoff bound.

\item We give a geometric interpretation of the required property for a pair of joint constellation symbols to achieve a low PEP. Specifically, the PEP exponent between a pair of joint $\big(T\times \sum_{k=1}^{K}M_k\big)$-matrix symbols $\Xm$ and $\Xm'$ scales linearly with a Riemannian distance between $\Id_T + \Xm\Xm^\H$ and $\Id_T + \Xm'{\Xm'}^\H$. This metric is the length of the geodesic (shortest path) joining these matrices in the manifold of Hermitian positive definite matrices. Therefore, a pair of joint symbols $\Xm$ and $\Xm'$ attains a low PEP if the matrices $\Id_T + \Xm\Xm^\H$ and $\Id_T + \Xm'{\Xm'}^\H$ are well separated in this manifold.

%	\item {We propose another design metric given by} the minimum expected pairwise log-likelihood ratio~(PLLR) over the joint constellation. It coincides with the minimum KL divergence metric in the massive MIMO regime, i.e., when the number of receive antennas is large.
%	%Our metric accounts for the fact that the individual constellations can have different sizes and different average symbol power.

\item From the dominant term of {a relaxed version of the Chernoff-based metric} at high SNR, we
obtain further simplified metrics {to reduce the complexity of constellation optimization.} We also propose an alternating optimization consisting in iteratively optimizing one user at a time to simplify the optimization.%Specifically, for any given pair of		constellation sizes, we can optimize the proposed metric over the set of		signal matrices.

\item {Inspired by our simplified metric, we propose a simple
	construction that consists in {\em partitioning} a single-user constellation. We also generalize our previously proposed construction based on {\em precoding} individual constellations of lower dimension.}

\item For a fixed joint constellation, we investigate power optimization and establish analytically the optimal set of per-user powers optimizing the metrics in the two-user SIMO case. {We also provide insights for power optimization in the $K$-user case.} 

\item {As a baseline approach, we adapt the existing criteria for the MIMO point-to-point channel to the MIMO MAC, namely, the max-min chordal distance criteria~\cite{Hochwald2000unitaryspacetime}, a criterion base on a high-SNR asymptotic bound on the PEP proposed in~\cite{Brehler2001asymptotic}, and a criterion based on the KL divergence~\cite{BorranTIT2003nonCoherentKLdistance}. We evaluate these baselines in terms of symbol-error rate and the value of our proposed metrics.}

\item For any given set of constellation sizes, the proposed metrics can be optimized over the set of constellation symbols. Assuming {per-user USTM,} we {implement a numerical routine} to solve the metric optimization problem, generate joint constellations, and compare with a pilot-based constellation and constellations optimized with {baseline} metrics.
Numerical results show that our Chernoff-based metric leads to significantly better symbol-error-rate performance than the state-of-the-art metrics, while our simplified metric leads to similar performance but lower optimization complexity than the existing ones.
%Numerical results show that our proposed metrics are meaningful \sheng{and effective}, and the resulting constellations  outperform the aforementioned baselines. 
\end{itemize}

We remark that our metrics are general for the multi-user case and, therefore, apply naturally to the single-user case. In the single-user case, our metrics lead to similar performance as the state-of-the-art metrics, which well exploit the asymptotic optimality of USTM. On the other hand, the advantage of our metrics over existing ones is more pronounced in the multi-user case, where the unitary property cannot hold for the joint constellation symbols due to the independence between users. In this case, the existing metrics for joint constellation design, relying on heuristic arguments instead of examining carefully the error exponent, do not provide performance guarantee as our proposed ones. %In other words, we do not assume user cooperation.}

\subsubsection*{Paper Organization}
\sheng{The remainder of the paper is organized as follows. In Section~\ref{sec:model}, we present the system model and
formulate the problem. In Section~\ref{sec:criteria}, we analyze the detection error probability
and derive the design metrics, based on which we propose {a} simple constellation construction in Section~\ref{sec:constructions}. In Section~\ref{sec:power_opt}, we address the transmit power optimization. We present the numerical results in Section~\ref{sec:performance} and conclude the paper in Section~\ref{sec:conclusion}. {A discussion on the extension to correlated fading, a generalization of our precoding-based design~\cite{HoangAsilomar2018multipleAccess},} and the proofs can be found in the 
%extended version~\cite{Hoang2020arXiv_MAC}
appendices}.

%are presented in Sec.~\ref{sec:performance}. Finally, Sec.~\ref{sec:conclusion} concludes the paper.

\subsubsection*{Notation}
%For random quantities, we use upper case letters,
%e.g., $\rv{X}$, for scalars, upper case letters with bold and non-italic fonts, e.g., $\rvVec{V}$, for vectors, and upper
%case letter with bold and sans serif fonts, e.g., $\rvMat{M}$, for matrices.  Deterministic quantities are denoted in
%a rather conventional way with italic letters, e.g., a scalar $x$, a vector $\pmb{v}$, and a matrix $\pmb{M}$.
Random quantities are denoted with non-italic letters with sans-serif fonts, e.g., a scalar $\rv{x}$, a vector $\rvVec{v}$, and a matrix $\rvMat{M}$. 
Deterministic quantities are denoted %in a rather conventional way
with italic letters, e.g., a scalar $x$, a vector $\pmb{v}$, and a
matrix $\pmb{M}$. {The $n\times n$ identity matrix is denoted by $\Id_n$.}
The Euclidean norm is denoted by $\|\cdot\|$ and the Frobenius norm by $\|\cdot\|_{\rm F}$. The trace, transpose, conjugate, and conjugate transpose of $\pmb{M}$ are respectively $\trace[\Mm]$,
$\pmb{M}^\T, \Mm^*,$ and $\pmb{M}^\H$. The $i$-th eigenvalue of $\Mm$ in decreasing order is denoted by $\sigma_i(\Mm)$, unless otherwise specified. %For a matrix $\Mm = \Um \Sigmam \Vm^\H$, we denote $|\Mm| = \Um \tilde{\Sigmam} \Vm^\H$ where $\tilde{\Sigmam}$ is the component-wise absolute value of $\Sigmam$. 
{We write $A:= B$ or $B=: A$ to denote that $A$ is defined by $B$.}
We use $\prod$ to denote the conventional or Cartesian product, depending on the factors;
$[n] \defeq \{1,2,\dots,n\}$; %The logarithm $\ln(\cdot)$ is to base $e$. 
$\ind{A}$ is the indicator function, whose value is $1$ if $A$ is true and $0$ if $A$ is false.  
%An element of the Grassmann manifold $G(\CC^T,M)$ is represented by a $T \times M$ truncated unitary matrix.
Given two functions $f(x)$ and $g(x)$, we write: $f(x) = O(g(x))$ if there exists a constant $c>0$ and some $x_0$ such that $|f(x)| \le c |g(x)|, \forall x \ge x_0$; $f(x) = \Theta(g(x))$ if $f(x) = O(g(x))$ and $g(x) = O(f(x))$. 
%The Grassmann manifold $G(\mathbb{C}^T,M)$ is defined as the space of $M$-dimensional subspaces in $\mathbb{C}^T$. % with $\mathbb{K} = \CC$ or $\mathbb{K} = \RR$. 
%In particular, $G(\mathbb{C}^T,1)$ is the Grassmannian of lines.
{Finally, $\Cc\Nc(\muv,\Sigmam)$ denotes the distribution of a complex proper
Gaussian random vector with mean $\muv$ and covariance matrix~$\Sigmam$.}

%-----------------------------------------------
\newcommand{\Mtot}{{M_{\rm tot}}}
%\vspace{-.2cm}
\section{System Model and Problem Formulation} \label{sec:model}
%\vspace{-.2cm}
We consider a MIMO MAC consisting of a receiver equipped with $N$ antennas and $K$ users, user $k$ with $M_k$ antennas, $k \in [K]$. The channel is assumed to be flat and block fading with equal-length and synchronous (across the users) coherence intervals of length $T \ge 2$. That is, the channel matrix $\rvMat{H}_k \in \CC^{N\times M_k}$ of user~$k$ remains constant within each coherence block of $T$ channel uses and changes between blocks.
%\footnote{The block fading model approximates, in a tractable manner, a continuous fading process, such as Jakes'~\cite{Jakes:1994:MMC:561302}, by a piecewise-constant process. The inter-block independence is relevant for a system with, e.g., block interleaving, random time-division multiple access~(TDMA), or sporadic transmissions, in which the gap between successive transmissions is indefinite.} 
Furthermore, the {\em distribution} of $\rvMat{H}_k$ is assumed to be known, but its {\em realizations} are unknown to both the users and the receiver. 
We consider independent and identically distributed~(i.i.d.) Rayleigh fading,\footnote{{We discuss the extension to the spatially correlated fading case in Appendix~\ref{app:correlated_fading}.}} namely, the rows of $\rvMat{H}  := [\rvMat{H}_1 \ \rvMat{H}_2 \ \dots \ \rvMat{H}_K]$ are independent and follow $\Cc\Nc\big(\mathbf{0},\Id_\Mtot\big)$ where $\Mtot \defeq \sum_{k = 1}^{K}M_k$. {Motivated by \cite[Th.~1]{Marzetta1999capacity}, we assume that $ \Mtot \le T$.} % with $M \defeq \sum_{k = 1}^{K}M_k$.
%We consider Rayleigh fading with correlation between the antennas of the
%same \sheng{user. Namely,} $\rvMat{H}_k$ has independent rows following
%$\Cc\Nc(\mathbf{0},\Psim_k)$ with $\trace(\Psim_k) = M_k$. The antennas
%across different users are assumed to be independent since the users
%are not co-located. In this way, the rows of $\rvMat{H} \sheng{:=} [\rvMat{H}_1 \ \rvMat{H}_2]$ are independent and follow $\Cc\Nc(\mathbf{0},\Psim)$ with $\Psim \defeq \big[\begin{smallmatrix}
%\Psim_1 & \mathbf{0} \\ \mathbf{0} & \Psim_2
%\end{smallmatrix}\big].$
Within {a representative coherence block,} each user $k$
sends a signal matrix\sheng{~symbol} $\rvMat{X}_k \in \CC^{T\times M_k}$, and the
receiver \sheng{observes} % a $N\times T$ signal matrix 
%%\vspace{-.1cm}
\begin{equation}
\rvMat{Y} = \sum_{k=1}^{K}\rvMat{X}_k \rvMat{H}_k^\T + \rvMat{Z}, %\quad t = 1,2,\ldots
\label{eq:channel_model}
\end{equation}
where the additive noise $\rvMat{Z} \in \CC^{T\times N}$ has i.i.d.~$\Cc\Nc(0,1)$ entries independent of $\{\rvMat{H}_k\}$, {and we omitted the block index for notational simplicity.} 
%We consider the power constraint
%	\begin{equation} \label{eq:power_constraint}
%		{\E[\|\rvMat{X}_k\|_{\rm F}^2]} \le PT.
%	%\tfrac{1}{n}\sum_{t=1}^{n}\|\rvMat{X}_k[t]\|_{\rm F}^2 \le P T, \quad k \in [K]
%	\end{equation}
%where $n$ is the number of blocks spanned by a codeword.
%Thus, $P$ is also an upper bound of the  per-user SNR at each
%receive antenna.

%, and $P$ is the maximum SNR per user. 

We assume that the transmitted symbol $\rvMat{X}_k$ takes value from a {\em finite constellation} $\Xcal_k$ of fixed size $|\Xc_k| = 2^{R_k T}$ with equally likely symbols, where $R_k$ (bits/channel use) is the transmission rate. %, i.e., $\Xm_k \in \Xcal_k$ for any realization $\Xm_k$ of $\rvMat{X}_k$, $k\in\{1,2\}$. 
{Let $P_k \defeq \frac{1}{T|\Xc_k|} \sum_{\Xm_k \in \Xc_k} \|\Xm_k\|_{\rm F}^2$ be the average normalized symbol power of user $k$. We consider the power constraint $P_k \le P$, $\forall k \in [K]$. Thus, $P$ is an upper bound of the average per-user SNR.% To satisfy the power constraint {\eqref{eq:power_constraint}}, we impose that $P_k \le P$.
} 
%\begin{align} \label{eq:const_power_constraint}
%	\frac{1}{|\Xc_k|} \sum_{\Xm_k \in \Xc_k} \|\Xm_k\|_{\rm F}^2 \eqdef P_kT \le PT, \quad k\in [K].
%\end{align}
We assume without loss of generality~{(w.l.o.g.)} that $\max_k P_k = P$.
Let us %focus on one representative block and omit the block index $t$, and
rewrite~\eqref{eq:channel_model} as  
\begin{equation}
\rvMat{Y}  = [
\rvMat{X}_1 \ \rvMat{X}_2 \ \dots \ \rvMat{X}_K] [
\rvMat{H}_1 \ \rvMat{H}_2 \ \dots \ \rvMat{H}_K]^\T + \rvMat{Z} = \rvMat{X} \rvMat{H}^\T + \rvMat{Z}, \label{eq:model_concatenated}
\end{equation}%
where the concatenated signal matrix $\rvMat{X} := [\rvMat{X}_1 \ \rvMat{X}_2 \ \dots \ \rvMat{X}_K]$ takes value from 
\begin{align}
\Xcal := \big\{ [
\Xm_1\ \Xm_2 \dots \Xm_K]:\ \Xm_k\in\Xcal_k \big\} = \prod_{k=1}^{K} \Xc_k.
\end{align}
Our goal is to derive the
desirable properties of the set tuple $(\Xcal_1, \Xcal_2, \dots, \Xcal_K)$ for a given rate
tuple $(R_1,R_2, \dots, R_K)$ to achieve \sheng{low {\em symbol detection error
	probability}}. 

%%\vspace{-.25cm}
\begin{remark}
\sheng{In the trivial case where only one of the users has non-zero
	rate, the joint constellation design problem boils down to the single-user
	constellation design.}
%then it is reasonable to let $P_2 = 0$ to avoid interference. This recovers TDMA and the problem is essentially a single-user constellation design, which has been well studied. 
%In general, with $R_1 > 0$ and $R_2 > 0$, the jointly optimization of $(\Xcal_1, \Xcal_2)$ and $(P_1,P_2)$ becomes non-trivial.  
\end{remark}
%%\vspace{-.25cm}

%Let us consider the maximum likelihood~(ML) detector.
Given $\rvMat{X} = \Xm$, the received signal $\rvMat{Y}$ is a Gaussian matrix with $N$ independent columns having zero mean and the same covariance matrix $\Id_T+\Xm \Xm^\H$. Thus, 
the likelihood function $p_{\rvMat{Y} | \rvMat{X}}$ is given by
\begin{equation}
%\lefteqn{
	p_{\rvMat{Y} | \rvMat{X}}(\Ym | \Xm) = \frac{\exp(-\trace(\Ym^\H(\Id_T+\Xm\Xm^\H)^{-1}\Ym ))}{\pi^{NT}\det^N(\Id_T+\Xm\Xm^\H)}.
	%} 
%\\
%&= \frac{\exp\big(-\frac{1}{\sigma^{2}}\|\Ym\|_{\rm F}^2+\frac{1}{\sigma^{2}}\trace[\Ym^\H\Dm(\sigma^2\Id_K + \Dm^\H\Dm)^{-1}\Dm^\H\Ym] \big)}{\pi^{NT}\det^N(\sigma^2\Id_K+{\Dm}^\H\Dm)}. 
\label{eq:likelihood}
\end{equation}
%\subsubsection{Maximum-likelihood decoder}
%{For convenience, we denote $\Gm_\Xm \defeq $}
Therefore, given the received \sheng{symbol} $\rvMat{Y} = \Ym$, the joint-user ML symbol detector %\footnote{The joint ML detection can be solved approximately with reduced complexity based on, e.g., expectation propagation~\cite{Hoang2018multiuser_EP}.}
is
\begin{align} \label{eq:MLdecoder}
\Xi(\Ym) %&= \arg \max_{\rvVec{s}_k \in \Sc_k, k = 1,\dots,K} p(\Ym | \Dm) \\
= \arg \max_{\Xm \in \Xc} 
&\big(-\trace\big((\Id_T+\Xm\Xm^\H)^{-1}\Ym \Ym^\H \big) \notag \\
&\quad - N\ln\det(\Id_T+\Xm\Xm^\H)\big).
\end{align} 
We aim to design $\Xc$ so as to minimize the ML detection error $P_e(\Xcal) = \P[\Xi(\rvMat{Y}) \ne \rvMat{X}]$, i.e., 
\begin{equation}
\Xc^* = \arg{\min_{\Xc}} \; P_e(\Xcal), \label{eq:criterion_MLerror}
\end{equation}
{subject to $P_k \le P$, $\forall k$, and $\max_k P_k = P$.}
\sheng{Since} $p_{\rvMat{Y} | \rvMat{X}}(\Ym | \Xm)$ depends on $\Xm$ only
through $\Xm \Xm^\H$, \sheng{the following proposition is
straightforward.}

%\vspace{-.3cm}
\begin{proposition}[Identifiability condition] \label{prop:identifiability}
For the joint ML detection error probability $P_e(\Xc)$ to vanish at high SNR, the joint constellation $\Xc$ must satisfy $\Xm \Xm^\H \ne {\Xm'} {\Xm'}^\H$ %, i.e., $\Xm$ and ${{\Xm'}}$ must not be the rotation of each other, 
for any pair \sheng{of distinct symbols} $\Xm$ and ${\Xm'}$ in $\Xc$. 
\end{proposition}

%%\vspace{-.3cm}

In the next section, we analyze the error probability and derive more specific design criteria.

%\vspace{-.3cm}
%-----------------------------------------------
\section{Constellation Design Criteria} \label{sec:criteria}
%%\vspace{-.2cm}
With $\rvMat{X}$ uniformly distributed
in $\Xc$, $P_e(\Xcal)$ can be written as
\begin{equation} \label{eq:joint_SER}
P_e(\Xcal) = \frac{1}{|\Xcal|}\sum_{\Xm\in\Xcal}\mathbb{P}\left(
\Xi(\rvMat{Y}) \ne \Xm | \rvMat{X} = \Xm\right).
\end{equation}%
We denote the pairwise error event as $\{\Xm\to{\Xm'}\} \defeq %\{\Xi{\Xm}(\rvMat{Y}) = {\Xm'} | \rvMat{X} = \Xm\} \equiv 
\{ p_{\rvMat{Y} | \rvMat{X}}(\rvMat{Y} | \Xm) \le
p_{\rvMat{Y} | \rvMat{X}}(\rvMat{Y} | {\Xm'}) | \rvMat{X} = \Xm\}$.
{For every given $\rvMat{X} = \Xm$, the ML detection error event $\{ \Xi(\rvMat{Y}) \ne \Xm | \rvMat{X} = \Xm\}$ is the union of the pairwise error events denoted by $\bigcup_{{\Xm'} \in\Xcal\setminus \{\Xm\}} \{\Xm\to{\Xm'}\}$. Therefore, $\mathbb{P}\left(
\Xi(\rvMat{Y}) \ne \Xm | \rvMat{X} = \Xm\right) \ge \mathbb{P}(\Xm\to{\Xm'})$ for every $\Xm' \in \Xc \setminus \{\Xm\}$, which implies that  $\mathbb{P}\left(
\Xi(\rvMat{Y}) \ne \Xm | \rvMat{X} = \Xm\right) \ge  \max_{\Xm' \in \Xc \setminus \{\Xm\}}
\mathbb{P}(\Xm\to{\Xm'})$, and thus
\begin{align}
	\sum_{\Xm\in\Xcal} \mathbb{P}\left(
	\Xi(\rvMat{Y}) \ne \Xm | \rvMat{X} = \Xm\right) &\ge  \max_{\Xm \in \Xc} \max_{\Xm' \in \Xc \setminus \{\Xm\}}
	\mathbb{P}(\Xm\to{\Xm'}) \\
	&= \max_{\Xm \ne \Xm' \in \Xc}
	\mathbb{P}(\Xm\to{\Xm'}) \label{eq:tmp400}.
\end{align}
Furthermore, for every $\Xm \in \Xc$,
\begin{align}
	&\mathbb{P}\left(
	\Xi(\rvMat{Y}) \ne \Xm | \rvMat{X} = \Xm\right) \notag \\
	&\le \sum_{\Xm' \in \Xc \setminus \{\Xm\}} \mathbb{P}(\Xm\to{\Xm'}) \label{eq:tmp405}\\
	&\le (|\Xcal|-1) \max_{\Xm \ne \Xm' \in \Xc}
	\mathbb{P}(\Xm\to{\Xm'}) \label{eq:tmp406}
\end{align}
where \eqref{eq:tmp405} follows from the union bound and \eqref{eq:tmp406} holds by replacing $\mathbb{P}(\Xm\to{\Xm'})$ by its maximal value over $\Xm \ne \Xm' \in \Xc$.}
%\begin{align}
%\bigcup_{{\Xm'} \in\Xcal: {\Xm'}\ne\Xm} \{\Xm\to{\Xm'}\} = 
%\bigcup_{{\Xm'} \in\Xcal: {\Xm'}\ne\Xm} \left\{\ln p_{\rvMat{Y} | \rvMat{X}}(\rvMat{Y} | \Xm) \le
%\ln p_{\rvMat{Y} | \rvMat{X}}(\rvMat{Y} | {\Xm'})\right\}. \label{eq:error_union}
%\end{align}%
%Since the ML detection error event is the union of the pairwise error events, after some simple manipulations, 
{Introducing \eqref{eq:tmp400} and \eqref{eq:tmp406} into \eqref{eq:joint_SER},} 
we have the
following upper and lower bounds on $P_e(\Xcal)$
\begin{multline}
\frac{1}{|\Xcal|} \max_{\Xm \ne \Xm' \in \Xc}
\mathbb{P}(\Xm\to{\Xm'}) \le P_e(\Xcal) 
\\ \le (|\Xcal|-1) \max_{\Xm \ne \Xm' \in \Xc}
\mathbb{P}(\Xm\to{\Xm'}). \label{eq:Pe_unionbound}
\end{multline}%
We see that for a given constellation size $|\Xcal|$, the symbol
detection error $P_e(\Xc)$ vanishes if and only if the \emph{worst-case PEP}, $\displaystyle\max_{\Xm \ne \Xm' \in \Xc}
\mathbb{P}(\Xm\to{\Xm'})$, vanishes. Therefore, {our goal from now on
is to minimize the worst-case PEP.} 

Following~\cite[Proposition 1]{Brehler2001noncoherent}, the PEP $\mathbb{P}(\Xm\to{\Xm'})$ can be derived in closed form as given in Appendix~\ref{app:closedFormPEP}. This closed-form expression, however, {is not exploitable for optimization and} does not bring clear insights into the constellation design. A high-SNR asymptotic expression of the PEP was given in~\cite[Proposition 3]{Brehler2001noncoherent}, but is also hard to exploit. {Therefore, one needs to resort to further simplified design criteria.} % since it involves the number of different individual symbols in $\Xm$ and $\Xm'$, which is highly non-smooth. 

\subsection{{Baseline Approach and Criteria}}
{A baseline approach is to treat the joint constellation as a constellation of an $\Mtot \times N$ MIMO point-to-point channel and adapt the existing criteria for that channel. This approach results in three criteria given below.

\subsubsection{Adapting the Max-Min Chordal Distance Criterion} 

By treating~\eqref{eq:model_concatenated} as a point-to-point channel, one can consider {USTM} and regard $\frac{\rvMat{X}}{\|\rvMat{X}\|_{\rm F}}$ as a truncated unitary matrix. Then, according to \cite[Th.~6]{Hochwald2000unitaryspacetime}, a Chernoff upper bound on the PEP $\mathbb{P}(\Xm\to{\Xm'})$ is given by a function of the squared singular values $\bar{\lambda}^2_1, \bar{\lambda}^2_2,\dots, \bar{\lambda}^2_{M_{\rm tot}}$ of the matrix $\frac{\Xm^\H}{\|\Xm\|_{\rm F}} \frac{\Xm'}{\|\Xm'\|_{\rm F}}$ (see~\cite[Eq.~(18)]{Hochwald2000unitaryspacetime}). As argued in~\cite[Sec.~2]{Agrawal2001MIMOconstellations}, this upper bound is increasing with a term dominated by $\bar{\lambda}^2_1 + \bar{\lambda}^2_2 + \dots + \bar{\lambda}^2_{M_{\rm tot}} = \trace[\frac{\Xm\Xm^\H \Xm'{\Xm'}^\H}{\|\Xm\|_{\rm F}^2 \|\Xm'\|_{\rm F}^2}]$. This leads to the design criterion\footnote{{A variant of this criterion proposed in~\cite{Gohary2009GrassmanianConstellations} consists in minimizing $\bar{\lambda}_1 + \bar{\lambda}_2 + \dots + \bar{\lambda}_{M_{\rm tot}} = \trace[\frac{\Xm^\H \Xm'}{\|\Xm\|_{\rm F} \|\Xm'\|_{\rm F}}]$. However, numerical simulations suggest that the resulting performance is similar. Therefore, we focus on \eqref{eq:criterion_maxminChordal} in this paper.}}
\begin{equation} \label{eq:criterion_maxminChordal}
%\Xc^* = \arg\max_{\Xc} \min_{\Xm \ne \Xm' \in \Xc} \textstyle \trace[\Id - \frac{\Xm^\H \Xm'{\Xm'}^\H\Xm}{\|\Xm\|_{\rm F}^2 \|\Xm'\|_{\rm F}^2}] \leftrightarrow 
(\text{Min-$m_1$}) \colon \Xc^* = \arg\min_{\Xc} \underbrace{\max_{\Xm \ne \Xm' \in \Xc} \textstyle \trace[\frac{\Xm \Xm^\H \Xm'{\Xm'}^\H}{\|\Xm\|_{\rm F}^2 \|\Xm'\|_{\rm F}^2}]}_{=: m_1(\Xc)}.
%\widehat{\Cc} \defeq \left\{\bigoplus_{k=1}^K {\rm span}(\rvVec{x}_k): \{\rvVec{x}_1,\dots,\rvVec{x}_K\} \in \prod_{k=1}^K \Cc_k\right\},
\end{equation}
This criterion coincides with the max-min chordal distance criterion for Grassmannian packing considered in~\cite{Conway1996packing,Dhillon2008constructing}.
%	We denote $m_1(\Xc) \defeq \displaystyle \max_{\Xm \ne \Xm' \in \Xc}\textstyle \trace[\frac{\Xm^\H \Xm'{\Xm'}^\H\Xm}{\|\Xm\|_{\rm F}^2 \|\Xm'\|_{\rm F}^2}]$ for future reference.

\subsubsection{Adapting a High-SNR Asymptotic Bound on the PEP}

Another design metric for the point-to-point channel based on a high-SNR asymptotic bound on the PEP~\cite{Brehler2001asymptotic} and the union bound on the average error probability was proposed in~\cite[Eq.~(8)]{McCloudIT2002signalDesignAndConvCode}. Adopting this metric, we consider the following criterion
\begin{align} 
&(\text{Min-$m_2$}) \colon  \\ 
&\Xc^* = \arg\min_{\Xc} \ \underbrace{\ln \sum_{\Xm\ne \Xm'\in \Xc} \det^{-N}\Big(\Id_T - M_{\rm tot}^2\tfrac{ \Xm\Xm^\H\Xm'{\Xm'}^\H}{\|\Xm\|_{\rm F}^2 \|\Xm'\|_{\rm F}^2}\Big)}_{=: m_2(\Xc)}. \label{eq:criterion_minsumDet}
\end{align}
%We denote $m_2(\Xc) \defeq \ln \displaystyle\sum_{\Xm\ne \Xm'\in \Xc} \det^{-N}\Big(\Id - M_{\rm tot}^2 \tfrac{\Xm^\H \Xm'{\Xm'}^\H\Xm}{\|\Xm\|_{\rm F}^2 \|\Xm'\|_{\rm F}^2}\Big)$ for future reference.

\subsubsection{Adapting a Criterion Based on the KL Divergence}
We note that minimizing the worst-case PEP is equivalent to maximizing the worst-case PEP exponent:
\begin{equation} 
\Xc^* = \arg\max_{\Xc} \min_{\Xm \ne \Xm' \in \Xc} \left(-\frac{1}{N} \ln \mathbb{P}(\Xm\to{\Xm'}) \right). \label{eq:criterion_PEP_exponent}
\end{equation}
An analysis of the PEP exponent follows from a relation between the joint symbol detection problem and hypothesis testing. 
Given the received signal $\{ \yv_n\}_{n=1}^N$, let us consider
two hypotheses: ${H}_0:\ \{ \yv_n\}_{n=1}^N \sim
\Cc\Nc(\mathbf{0},\Id_T + \Xm \Xm^\H)$ and $H_1:\ \{
\yv_i\}_{i=1}^N \sim \Cc\Nc\big(\mathbf{0},\Id_T + {\Xm'}
{\Xm'}^\H\big)$ where $\{ \yv_n\}_{n=1}^N$ are realizations of $N$ columns of $\rvMat{Y}$. Then, the detection of the transmitted joint symbol between $\Xm$ and $\Xm'$ can be seen as a hypothesis test between $H_0$ and $H_1$. The PEP $\P (\Xm\to{\Xm'})$ can be seen as the type-1 error probability of the likelihood ratio test. Then, it follows from the Chernoff-Stein Lemma~\cite[Th.~11.8.3]{Cover2006elements} that as $N\to\infty$, the {\em lowest achievable} error exponent for $\P[\Xm \to \Xm']$, with the constraint that $\P[\Xm' \to \Xm]$ is smaller than a given threshold\footnote{{In Appendix~\ref{proof:KL}, we show that $\P({\Xm}\to {\Xm}') \to 0$ as $N\to \infty$ for any pair of distinct symbols $\Xm$ and $\Xm'$ of a joint constellation satisfying the identifiability condition in Proposition~\ref{prop:identifiability}. Swapping the symbols' roles, we obtain that $\P({\Xm}'\to {\Xm}) \to 0$ as $N\to \infty$. Therefore, for any $\epsilon \in (0,1/2)$, there exists $N_\epsilon > 0$ such that $\P[\Xm' \to \Xm] < \epsilon$ for $N > N_\epsilon$.}} $\epsilon \in (0,1/2)$, is given by $D\big(\Cc\Nc(\mathbf{0},\Id_T + \Xm \Xm^\H) \|  \Cc\Nc\big(\mathbf{0},\Id_T + {\Xm'} {\Xm'}^\H\big)\big)$ where $D(\cdot\|\cdot)$ denotes the KL divergence.
The convergence of the PEP exponent to the KL divergence was also exploited in~\cite{BorranTIT2003nonCoherentKLdistance,Chen2020design_NOMA_massiveMIMO,Li2021constellation}.
However, note that this error exponent is not achieved with the considered joint-ML detector~\eqref{eq:MLdecoder}, but with a detector which is highly biased in favor of $H_0$~\cite{BorranTIT2003nonCoherentKLdistance}. It serves as an upper bound on the PEP exponent of the joint-ML detector. In fact, the performance of the joint-ML detector is related to the KL divergence as shown in~\cite[Lemma~3]{BorranTIT2003nonCoherentKLdistance}. This result is stated in the following to be self-contained. 
\begin{proposition}[{Relation of the joint-ML detection error and the KL divergence~\cite[Lemma~3]{BorranTIT2003nonCoherentKLdistance}}] \label{prop:KL}
Let $\{ \yv_n\}_{n=1}^N \in \Yc$ be drawn~i.i.d. according to the probability density function~(pdf)~$p_0$ on $\Yc$. Let $p_1$ and $p_2$ be pdfs on $\Yc$ with $0 < D(p_0 \| p_2) < D(p_0 \| p_1) < \infty$. Consider two hypothesis tests, one between $\{ \yv_n\}_{n=1}^N \sim p_0$ and $\{ \yv_n\}_{n=1}^N \sim p_1$, and the other between $\{ \yv_n\}_{n=1}^N \sim p_0$ and $\{ \yv_n\}_{n=1}^N \sim p_2$. Let $L_i = \prod_{n = 1}^{N} \frac{p_0(\yv_n)}{p_i(\yv_n)}$ denote the likelihood ratios for the two tests so that the probability of mistaking $p_0$ for $p_i$ using the ML detector is given by $\P[p_0 \to p_i] = \P_{p_0}[L_i < 1]$, $i\in \{1,2\}$. Let $\Delta D \defeq D(p_0 \| p_1) - D(p_0 \| p_2) > 0$. It holds that 
\begin{align}
\P_{p_0}\left[L_1 < \exp\Big(\frac{N \Delta D}{2}\Big) L_2\right] \to 0, \quad  \text{as~} N\to\infty.
\end{align}
\end{proposition}

Proposition~\ref{prop:KL} says that, for $N$ large enough, the likelihood ratio of the first test is greater than the likelihood ratio of the second test with high probability. This implies that for large $N$, the first test\textemdash corresponding to the hypothesis with greater KL divergence from the true channel output distribution\textemdash has a lower error probability than the second test. In other words, a pair of joint symbols that leads to higher KL divergence are less likely to be misdetected for each other. Motivated by this, the KL divergence $D\big(\Cc\Nc(\mathbf{0},\Id_T + \Xm \Xm^\H) \|  \Cc\Nc\big(\mathbf{0},\Id_T + {\Xm'} {\Xm'}^\H\big)\big) = \meanLLR$ can be used as a design metric for the joint constellation design, as proposed for the point-to-point channel is~\cite[Eq.~(32)]{BorranTIT2003nonCoherentKLdistance}. 
Specifically, %letting $e_{\min}(\Xc) \defeq \frac{1}{N} \displaystyle\min_{\Xm \ne \Xm' \in \Xc}  \E\big[\rv{L}(\Xm\to{\Xm'})\big]$, 
we consider the following design criterion
\begin{equation} 
(\text{Max-$e_{\min}$}) \colon \Xc^* = \arg\max_{\Xc} \ \underbrace{\frac{1}{N} \displaystyle\min_{\Xm \ne \Xm' \in \Xc}  \E\big[\rv{L}(\Xm\to{\Xm'})\big]}_{=: e_{\min}(\Xc)} \label{eq:criterion_minMean}
\end{equation}
where it follows from \eqref{eq:LLR0} and $\mathbb{E}[\rvMat{Y} \rvMat{Y}^\H] = N \big(\Id_T+ \Xm\Xm^\H)$ that
\begin{align}
\meanLLR
%&= N\ln\frac{\det\left( \Id_T+ {\Xm'}^\H
%  {\Xm'} \right)}{\det\left( \Id_T+ \Xm \Xm^\H
%  \right)} - 
%  \trace\left( ( \Id_T+ \Xm \Xm^\H )^{-1} \mathbb{E}(\rvMat{Y}^\H
%  \rvMat{Y}) \right)
%  + \trace\left( ( \Id_T+ {\Xm'}{\Xm'}^\H )^{-1} \mathbb{E}
%  ( \rvMat{Y}^\H \rvMat{Y} ) \right) \\
%&= N\ln\frac{\det\left( \Id_T+ {\Xm'}{\Xm'}^\H \right)}{\det\left( \Id_T+ \Xm \Xm^\H \right)} - N + N \trace\left( ( \Id_T+ {\Xm'}{\Xm'}^\H )^{-1} ( \Id_T+ \Xm \Xm^\H) \right) \\
&= N\ln\frac{\det( \Id_T+  {\Xm'}{\Xm'}^\H
)}{\det\left( \Id_T+  \Xm \Xm^\H \right)} - N \\
&\quad + 
N \trace\big( ( \Id_T
+ {\Xm'}{\Xm'}^\H )^{-1} \big) \notag \\ &\quad+
N\trace\big( ( \Id_T+ {\Xm'}{\Xm'}^\H )^{-1} \Xm
\Xm^\H \big). \quad \label{eq:meanLLR}
%&= \underbrace{N\ln\frac{\det\left( \Id_T+ {\Xm'}{\Xm'}^\H
%\right)}{\det\left( \Id_T+ \Xm \Xm^\H \right)}}_{O(\ln P)} - N +
%\underbrace{N \trace\left( ( \Id
%+ {\Xm'}{\Xm'}^\H )^{-1} \right)}_{O(1)} +
%\underbrace{N\trace\left( ( \Id_T+ {\Xm'}{\Xm'}^\H )^{-1} \Xm^\H
%\Xm \right)}_{\Theta(P)}. \label{eq:meanLLR}
\end{align}
%An alternative approach is to treat the joint constellation as a constellation of an $\Mtot \times N$ MIMO point-to-point channel, and adopt the single-user max-min chordal distance criterion 
%%Grassmannian constellation in $G\big(\CC^T,\Mtot\big)$ (although the joint symbol $\Xm \in \Xc$ is not necessarily truncated unitary). Then, by mimicking the commonly used max-min chordal distance criterion, we obtain the design criterion
%$\Xc^* = \arg\displaystyle\max_{\Xc} \displaystyle\min_{\Xm \ne \Xm' \in \Xc} \textstyle \trace\Big(\Id - \frac{\Xm^\H \Xm'{\Xm'}^\H\Xm}{\|\Xm\|_{\rm F}^2 \|\Xm'\|_{\rm F}^2}\Big)$, or equivalently, %\footnote{As long as $\Xm_1,\dots,\Xm_K$ are linearly independent, $\Span[[\Xm_1 \dots \Xm_K]]$ is a $\Mtot$-dimensional subspace of $\CC^T$, which is an element of $G\big(\CC^T,\Mtot\big)$.} 
%	%which leads to a design criterion mimicking sphere packing in this Grassmannian by maximizing the minimum pairwise chordal distance.
}

{The criteria Min-$m_1$~\eqref{eq:criterion_maxminChordal}, Min-$m_2$~\eqref{eq:criterion_minsumDet}, and Max-$e_{\min}$~\eqref{eq:criterion_minMean} serve as baselines for our proposed criteria.}
In the following, we present our approach and derive new design criteria.

%However, in general, $\Lambdam$ can be indefinite, 
%and the distribution of $\Lambdam \rvMat{Y}_0\rvMat{Y}_0^\H$ is hard to characterize. Therefore, $\mathbb{P}(\rv{L}(\Xm\to{\Xm'}) \le 0 )$ is hard to compute in general. %\footnote{We can compute $\mathbb{P}(\rv{L}(\Xm\to{\Xm'}) \le 0 )$ in closed form when $\Xm \Xm^\H - {\Xm'} {\Xm'}^\H$ is semidefinite, as presented in the Appendix. However, this is generally not the case for any pair of joint symbols.} 
%\vspace{-.4cm}
\subsection{{Proposed} Criteria} \label{sec:emin_dmin}
%\vspace{-.2cm}
{Let us rewrite the PEP as
\begin{align}
%$
\mathbb{P}(\Xm\to{\Xm'}) 
&= \mathbb{P}\left(\ln
\frac{p_{\rvMat{Y} | \rvMat{X}}(\rvMat{Y} | \Xm)}{p_{\rvMat{Y} | \rvMat{X}}(\rvMat{Y} | {\Xm'})} \le 0 \right) \\
&= \mathbb{P}\big(\rv{L}(\Xm\to{\Xm'}) \le 0 \big)
%$
\label{eq:PEP}
\end{align}
with the pairwise log-likelihood ratio~(PLLR)  $\rv{L}(\Xm\to{\Xm'})$ defined as $\ln
\frac{p_{\rvMat{Y} | \rvMat{X}}(\rvMat{Y} | \Xm)}{p_{\rvMat{Y} | \rvMat{X}}(\rvMat{Y} | {\Xm'})}$. 
Using \eqref{eq:likelihood}, we obtain % a closed-form expression of the LLR
\begin{align}
%\lefteqn{
&\rv{L}(\Xm\to{\Xm'})
%}\notag \\ 
%&= N\ln\frac{\det\left( \Id_T+ PT{\Xm'}
%	{\Xm'}^\H \right)}{\det\left( \Id_T+ PT\Xm \Xm^\H
%	\right)} - 
%\trace\left( ( \Id_T+ PT\Xm \Xm^\H )^{-1} \rvMat{Y}^\H \rvMat{Y} \right)
%+ \trace\left( ( \Id_T+ PT{\Xm'} {\Xm'}^\H )^{-1} \rvMat{Y}^\H
%\rvMat{Y} \right) 
\notag \\
&= N\ln\frac{\det\big( \Id_T+ {\Xm'}
{\Xm'}^\H \big)}{\det\left( \Id_T+ \Xm \Xm^\H
\right)} \\
&\quad - \trace\Big( \big(( \Id_T+ \Xm \Xm^\H )^{-1} - ( \Id_T+ {\Xm'} {\Xm'}^\H )^{-1} \big) \rvMat{Y} \rvMat{Y}^\H \Big).\quad  \label{eq:LLR0}
\end{align}	
Hereafter, we use $\{\lambda_i\}_{i=1}^T$ to denote the eigenvalues of the matrix $\Gammam \defeq (\Id_T+\Xm \Xm^\H)(\Id_T+{\Xm'}{\Xm'}^\H)^{-1}$. Note that $\lambda_i \ge 0$, $\forall i \in [T]$.
The following expression of the PEP will be useful in our analysis.
\begin{lemma} \label{lemma:PEP_expression}
The PEP can be expressed as
\begin{align}
\mathbb{P}(\Xm\to{\Xm'}) = \P[\sum_{i=1}^{T} (\lambda_i - 1) \rv{g}_i \le N \sum_{i=1}^T \ln \lambda_i],
\end{align}
where $\{\rv{g}_i\}_{i=1}^T$ are independent Gamma random variables with shape $N$ and scale $1$.
\end{lemma}
\begin{proof}
See Appendix~\ref{proof:PEP_expression}.
\end{proof}
\subsubsection{A Criterion Based on Nonasymptotic Bounds}
The following proposition gives a lower bound on the PEP exponent. 
\begin{proposition}[PEP exponent's Chernoff lower bound] \label{prop:Chernoff}
It holds that, for every $s\in[0,1]$,
\begin{align}
&-\frac{1}{N} \ln \P(\Xm\to{\Xm'}) \notag \\
&\ge J_s(\Xm,\Xm')  \label{eq:bound_PEP_exponent}\\
&\defeq \ln\det(s(\Id_T + \Xm'{\Xm'}^\H)^{-1} + (1-s)(\Id_T + \Xm\Xm^\H)^{-1}) \notag \\
&\quad - \Big[s\ln\det[(\Id_T + \Xm'{\Xm'}^\H)^{-1}] \notag \\
&\qquad ~+ (1-s)\ln\det[(\Id_T + \Xm\Xm^\H)^{-1}]\Big].
\end{align}
\end{proposition}
\begin{proof}
The proof follows from the Chernoff bound, %The asymptotic PEP exponent expression follows from the Chernoff-Stein lemma for a hypothesis test of the channel output distribution. The details are 
and is provided in Appendix~\ref{proof:PEP_exponent}.
\end{proof}
In particular, with $s = \frac{1}{2}$, after some manipulations, we obtain
\begin{align} \label{eq:J_1/2}
J_{1/2}(\Xm,\Xm') &= \frac{1}{2} \ln\det\big(2\Id_T + (\Id_T \!+\! \Xm'\Xm'^\H)^{-1}(\Id_T \!+\! \Xm\Xm^\H) \notag \\
&\qquad \qquad \quad + (\Id_T + \Xm\Xm^\H)^{-1}(\Id_T + \Xm'\Xm'^\H)\big) \notag \\
&\quad - T \ln 2. % \\
%	&= \frac{1}{2} \sum_{i=1}^{T} \ln \Big(2 + \lambda_i + \frac{1}{\lambda_i}\Big).
\end{align}

The bounds of the PEP exponent can be tightened with an upper bound as follows. 
\begin{proposition}[PEP exponent's upper and lower bounds] \label{prop:lower_upper_bounds}
The PEP exponent is upper and lower-bounded as
\begin{equation} \label{eq:lower_upper_bounds}
b(\Xm,\Xm') + T \ge -\frac{1}{N} \ln \P(\Xm\to{\Xm'}) \ge \frac{1}{2} b(\Xm,\Xm') - T \ln 2. 
\end{equation}
where $b(\Xm,\Xm')$ is defined through $\{\lambda_i\}$ as $b(\Xm,\Xm') \defeq \sum_{i=1}^{T}  |\ln \lambda_i |$.
\end{proposition} 
\begin{proof}
See Appendix~\ref{proof:lower_upper_bounds}.
\end{proof}
Proposition~\ref{prop:lower_upper_bounds} states that the PEP exponent
scales linearly with $b(\Xm,\Xm')$ where the multiplicative factor
is tightly bounded between $\frac12$ and $1$. Note that the lower limit
factor $\frac12$ can be improved by optimizing the parameter $s$ in
Proposition~\ref{prop:Chernoff}. For the purpose of this paper, however,
we neglect the multiplicative and additive factors and focus on the
key part $b(\Xm,\Xm')$ in both upper and lower bounds. Define $b_{\min}(\Xc) \defeq \displaystyle\min_{\Xm\ne \Xm' \in \Xc} b(\Xm,\Xm')$. It follows from Proposition~\ref{prop:lower_upper_bounds} that the worst-case PEP exponent is sandwiched between $b_{\min} (\Xc) + T$ and $\frac{1}{2}b_{\min}(\Xc) - T \ln 2$. Motivated by this, we propose the following design criterion
\begin{equation} \label{eq:criterion_Riemannian}
(\text{Max-$b_{\min}$}) \colon \Xc^* = \arg \max_\Xc b_{\min}(\Xc).
\end{equation}

\begin{remark}
As opposed to the baseline metrics $m_1(\Xc)$, $m_2(\Xc)$, and $e_{\min}(\Xc)$ that are based on asymptotic bounds on the PEP as $P$ or $N$ goes to infinity, our metric  $b_{\min}(\Xc)$ is justified with nonasymptotic bounds.
\end{remark}
\begin{remark}
Since the metric $b_{\min} (\Xc)$ provides tight bounds on the PEP exponent, it can also be used to evaluate the error performance of a given joint constellation. The higher the value of $b_{\min} (\Xc)$, the lower the joint ML detection error is expected to be. Computing $b_{\min}(\Xc)$ is more efficient than evaluating the empirical joint ML symbol error rate.
\end{remark}

{In numerical optimization of $b_{\min} (\Xc)$, one has to compute the gradient of $b(\Xm,\Xm')$ with respect to the symbols. This can be challenging since $b(\Xm,\Xm')$ involves the eigenvalues of $\Gammam$. In this regard, it is more convenient to maximize the bound $J_s(\Xm,\Xm')$ given in Proposition~\ref{prop:Chernoff}: 
\begin{equation} 
(\text{\text{Max-$J_{s, \min}$}}) \colon \ \Xc^* = \arg\max_{\Xc} \underbrace{\min_{\Xm \ne \Xm' \in \Xc} J_{s}(\Xm,\Xm')}_{=: J_{s,\min} (\Xc)} \label{eq:criterion_Js}
\end{equation}
for some $s \in [0,1]$. In the single-user SIMO case, let $s= \frac12$ and consider Grassmannian signaling with $\|\xv\| = PT, \forall \xv \in \Xc$, then Max-$J_{s, \min}$ is equivalent to the max-min chordal distance criterion $\Xc^* = \arg \max_\Xc \min_{\xv \ne \xv' \in \Xc} \sqrt{1-\frac{1}{P^2T^2}|\xv^\H \xv'|^2}.$}

{
\subsubsection{The Relation to Riemannian Distance}

We next point out a geometric interpretation for the property of a pair of joint symbols that achieves low PEP. To this end, let $\Pc_T$ be the set of $T \times T$ Hermitian and positive definite matrices. This set is a differential manifold. At a point $\Am$ of $\Pc_T$, define the Riemannian metric $\|\Am^{-\frac12} \dif \Am \Am^{-\frac12}\|_{\rm F} = \big[\trace[\Am^{-1} \dif \Am ]^2\big]^\frac12$. This metric is used to compute the length of a piecewise differential path in $\Pc_T$. Specifically, the length of a path $\gamma \colon [a,b] \to \Pc_T$ is given by $L(\gamma) = \int_{a}^{b} \|\gamma^{-\frac12}(t) \gamma'(t)\gamma^{-\frac12}(t)\|_{\rm F} \dif t$~\cite[Ch.~6]{bhatia2009positive}. %See~\cite[Chapter~6]{bhatia2009positive} for a detailed description of this manifold and its connection with some matrix properties. We pre Of particular interest to us is the following result: 
The Riemannian distance $\delta_{\rm R}(\Am,\Bm)$ between any two points $\Am$ and $\Bm$ in $\Pc_T$ is defined as 
%	\begin{align} \label{eq:def_Riemannian}
%		\delta_{\rm R}(\Am,\Bm) \defeq \inf \{L(\gamma) \colon \text{$\gamma$ is a path from $\Am$ to $\Bm$}\}.
%	\end{align}
%	That is, $\delta_{\rm R}(\Am,\Bm)$ is 
the length of the \emph{geodesic} between $\Am$ and $\Bm$, i.e., the shortest path joining $\Am$ and $\Bm$ in the manifold. According to~\cite[Ch.~6]{bhatia2009positive}, %the infimum in~\eqref{eq:def_Riemannian} is achieved by a unique path joining $\Am$ and $\Bm$, which is called a \emph{geodesic} from $\Am$ to $\Bm$. This geodesic has a parameterization $\gamma(t) = \Am^{\frac12} \big(\Am^{-\frac12} \Bm \Am^{-\frac12}\big)^t \Am^\frac12$, $0\le t \le 1.$ 
%	Furthermore, 
$\delta_{\rm R}(\Am,\Bm)$ is explicitly given by 
\begin{equation}
\delta_{\rm R}(\Am,\Bm) = \|\ln(\Am^{-\frac12} \Bm \Am^{-\frac12})\|_{\rm F} = \Bigg( \sum_{i=1}^{T} \ln^2 \sigma_i(\Am^{-1} \Bm)\Bigg)^\frac12
\end{equation}
where $\{\sigma_i(\Mm)\}$ denote the eigenvalues of a matrix $\Mm$. The distance $\delta_{\rm R}(\Am,\Bm)$ is called the \emph{Riemannian distance} on the manifold $\Pc_T$. The readers are referred to~\cite[Ch.~6]{bhatia2009positive} for a further description of this distance and its relation to the geometry of the manifold $\Pc_T$.}

{We now present a relation between our $b$-metric and the Riemannian distance. Since the matrices $\Id_T + \Xm\Xm^\H$ and $\Id_T + \Xm'{\Xm'}^\H$ are Hermitian and positive definite, the Riemannian distance between them is given by $\delta_{\rm R}(\Id_T + \Xm\Xm^\H,\Id_T + \Xm'{\Xm'}^\H) = \big( \sum_{i=1}^{T} \ln^2 \lambda_i \big)^{1/2}$.
\begin{proposition}[Relation between the $b$-metric and Riemannian distance] \label{prop:Riemannian_metric}
The metric $b(\Xm,\Xm')$ is bounded  in terms of the Riemannian distance $\delta_{\rm R}(\Id_T + \Xm\Xm^\H,\Id_T + \Xm'{\Xm'}^\H)$ as
\begin{multline}
\sqrt{T} \delta_{\rm R}(\Id_T + \Xm\Xm^\H,\Id_T + \Xm'{\Xm'}^\H) \ge b(\Xm,{\Xm'})  \\ \ge \delta_{\rm R}(\Id_T + \Xm\Xm^\H,\Id_T + \Xm'{\Xm'}^\H). \label{eq:Riemmanian}
\end{multline}
\end{proposition}
\begin{proof}
The lower bound follows from 
\begin{multline}
	b(\Xm,\Xm') = \sum_{i=1}^T |\ln \lambda_i| = \sqrt{\bigg(\sum_{i=1}^T |\ln \lambda_i|\bigg)^2} \\ \ge \sqrt{\sum_{i=1}^T \ln^2 \lambda_i} = \delta_{\rm R}(\Id_T + \Xm\Xm^\H,\Id_T + \Xm'{\Xm'}^\H),
\end{multline}
where the inequality holds because the terms $|\ln \lambda_i|$ are nonnegative. The upper bound follows directly from the Cauchy-Schwarz inequality.
\end{proof}

\begin{remark}
Proposition~\ref{prop:Riemannian_metric} says that the metric $b(\Xm,\Xm')$ is within a
multiplicative factor from the Riemannian distance $\delta_{\rm R}(\Id_T
+ \Xm\Xm^\H,\Id_T + \Xm'{\Xm'}^\H)$, where the factor is bounded between $1$ and
$\sqrt{T}$. Therefore, $b(\Xm,\Xm')$ is large if and only if
$\delta_{\rm R}(\Id_T + \Xm\Xm^\H,\Id_T + \Xm'{\Xm'}^\H)$ is large. It
follows that \emph{a pair of joint symbols $\Xm$ for $\Xm'$ are less
likely to be misdetected for each other if the geodesic joining $\Id_T +
\Xm\Xm^\H$ and $\Id_T + \Xm'{\Xm'}^\H$ in $\Pc_T$ is longer.} If $\Xm
\Xm = \Xm' {\Xm'}^\H$, this geodesic has length zero, thus
$b(\Xm,\Xm') = 0$ and the PEP exponent is upper bounded by a
constant. This agrees with the identifiability condition in Proposition~\ref{prop:identifiability}.
\end{remark}

{
\subsubsection{Simplified Criteria}

In the following, we further simplify the design criteria. As we shall see, this simplification leads to simpler metrics which can be optimized at reduced complexity, and to simple constructions allowing to efficiently generate the joint constellation. We first relax the Chernoff bound in Proposition~\ref{prop:Chernoff} as follows.
\begin{proposition}[PEP exponent's relaxed lower bound] \label{prop:Chernoff_relax}
It holds that
\begin{align} \label{eq:Chernoff_relax}
-\frac{1}{N} \ln \P(\Xm\to{\Xm'}) %\ge \ln \left(1+ \frac{1}{2} \trace[\Gammam] + \frac{1}{2} \trace[\Gammam^{-1}] \right) - \frac{T}{2} \ln 2 
\ge  \ln \left(1+ \frac{1}{2} \trace[\Gammam]\right) - \frac{T}{2} \ln 2.
\end{align}
\end{proposition}
\begin{proof}
If $\Xm\Xm^\H = \Xm'{\Xm'}^\H$, then $\Gammam = \Id_T$ and \eqref{eq:Chernoff_relax} is trivial since the right-hand side is at most $0$ for $T \ge 2$.  If $\Xm\Xm^\H \ne \Xm'{\Xm'}^\H$, applying the Chernoff bound in Proposition~\ref{prop:Chernoff} with $s = 1/2$, we get
\begin{align}
&-\frac{1}{N} \ln \P(\Xm\to{\Xm'}) \notag \\ &\ge J_{1/2}(\Xm,\Xm') \notag \\
&= \frac{1}{2} \sum_{i=1}^{T} \ln \Big(2 + \lambda_i + \frac{1}{\lambda_i}\Big) - T \ln 2 \\
&= \frac{1}{2}  \ln \left(\prod_{i=1}^{T} \Big(2 + \lambda_i + \frac{1}{\lambda_i}\Big)\right) - T \ln 2 \\
&\ge \frac{1}{2}  \ln \left(2^T + 2^{T-1} \sum_{i=1}^{T}\Big(\lambda_i + \frac{1}{\lambda_i}\Big)\right) - T \ln 2 \\
&= \ln \left(1+ \frac{1}{2} \trace[\Gammam] + \frac{1}{2} \trace[\Gammam^{-1}] \right) - \frac{T}{2} \ln 2 \\
&\ge \ln \left(1+ \frac{1}{2} \trace[\Gammam]\right) - \frac{T}{2} \ln 2,
\end{align}
where the inequalities follow from the fact that $\{\lambda_i\}$ are positive for $\Xm\Xm^\H \ne \Xm'{\Xm'}^\H$.
\end{proof}
Hence, maximizing $\trace[\Gammam]$ can lead to large PEP exponent. We have that $\trace[\Gammam] = \trace\big( ( \Id
+ {\Xm'}{\Xm'}^\H )^{-1} \big)  + \trace\big( ( \Id_T+ {\Xm'}{\Xm'}^\H )^{-1} \Xm
\Xm^\H \big).$ The next proposition characterizes how the terms in the right-hand side scale with the transmit power.
\begin{proposition}[{Dominating term in $\trace[\Gammam]$}] \label{lem:trGamma_scaling}
Let $\Xm$ and ${\Xm'}$ be such that $\|\Xm\vv\|_{\rm F}^2 =
\Theta(P)$ and $\|{\Xm'}\vv\|_{\rm F}^2 = \Theta(P)$ as
$P\to\infty$ for any unit-norm vector $\vv \in \CC^{\Mtot}$. We have that $\trace\big( ( \Id_T	+ {\Xm'}{\Xm'}^\H )^{-1} \big)$ scales as $O(1)$, while %$\ln\frac{\det( \Id_T+ {\Xm'}{\Xm'}^\H )}{\det( \Id_T+ \Xm  \Xm^\H )} = O(1)$ if	$\Span(\Xm) = \Span({\Xm'})$ and $\Theta(\ln P)$ otherwise. Furthermore, 
$\trace\big( ( \Id_T+ {\Xm'}{\Xm'}^\H )^{-1} \Xm \Xm^\H \big) = O(1)$ if $\Span(\Xm) = \Span({\Xm'})$ and $\Theta(P)$ otherwise.
%	\begin{align}
%	\trace\left( \big( \Id_T	+ {\Xm'}{\Xm'}^\H \big)^{-1} \right) &= O(1), \label{eq:scaling2}\\
%	\ln\frac{\det\big( \Id_T+ {\Xm'}{\Xm'}^\H
%		\big)}{\det\left( \Id_T+ \Xm \Xm^\H \right)} &= \begin{cases}
%		O(1), &\text{if~} \Span[\Xm] = \Span\big({\Xm'}\big), \\
%		\Theta(\ln P), &\text{otherwise},
%	\end{cases} \label{eq:scaling1}\\
%	\trace\left( \big( \Id_T+ {\Xm'}{\Xm'}^\H \big)^{-1} \Xm \Xm^\H \right) &= \begin{cases}
%	O(1), &\text{if~} \Span[\Xm] = \Span\big({\Xm'}\big), \\
%	\Theta(P), &\text{otherwise}. 
%	\end{cases} \label{eq:scaling3}
%	\end{align}
%	in the high-SNR regime $P\to\infty$.
\end{proposition}
%\vspace{-.4cm}
\begin{proof}
See Appendix~\ref{proof:lem:meanLLR_scaling}. %the extended version~\cite{Hoang2020arXiv_MAC}.
\end{proof}
From this proposition, we see that $d(\Xm\to{\Xm'}) := \trace\big((\Id_T + {\Xm'} {\Xm'}^\H
)^{-1} \Xm \Xm^\H\big)$ is the only term in $\trace[\Gammam]$ that can scale up linearly with $P$. 
\begin{remark} \label{remark:dmetric_emetric}
Following similar lines as in Appendix~\ref{proof:lem:meanLLR_scaling}, we can show that $\ln\frac{\det( \Id_T+ {\Xm'}{\Xm'}^\H )}{\det( \Id_T+ \Xm  \Xm^\H )}$ scales as $O(1)$ if $\Span(\Xm) = \Span({\Xm'})$ and $\Theta(\ln P)$ otherwise. Therefore, $d(\Xm\to{\Xm'})$ is also the only term in $\meanLLR$ (see~\eqref{eq:meanLLR}) that can scale up linearly with $P$.
\end{remark}
}

{By focusing on the dominating term $d(\Xm\to{\Xm'})$ in $\trace[\Gammam]$ (and in $\meanLLR$),} letting  
$d_{\min}(\Xcal) \defeq \displaystyle\min_{\Xm \ne {\Xm'}\in\Xcal}  d(\Xm\to{\Xm'})$, we have the following design
criterion
\begin{equation}
(\text{Max-$d_{\min}$}) \colon \Xc^* = \arg\max_{\Xc} d_{\min}(\Xcal). \label{eq:criterion_minTrace}
\end{equation}

Hereafter, we assume for simplicity that all users have the same number of antennas, i.e. $M_1=\dots =M_K=M$. %,although the general case follows in a straightforward manner. 
{We further analyze the metric $d_{\min}(\Xcal)$ in the following.}

%\vspace{-.2cm}
\paragraph{The Single-User Case} \label{sec:single_user}
%\vspace{-.2cm}
In the single-user case with $M$ transmit antennas, it is known that the
high-SNR optimal input signal takes the form of a truncated unitary matrix% belongs to the Grassmann manifold
~\cite{ZhengTse2002Grassman}. We consider 
this approach 
%a Grassmannian
%constellation~\cite{Gohary2009GrassmanianConstellations} $\Xc \subset
%G(\CC^T,M)$, thus 
and let ${\Xm}^\H{\Xm} = \frac{PT}{M}\Id_M, \forall \Xm \in \Xc$. %(Recall that we use truncated unitary matrices to represent elements of the Grassmannian.) 
Using the Woodbury identity $(\Id_T + {\Xm'} {\Xm'}^\H)^{-1} = \Id_T - \Xm'(\Id_M + {\Xm'}^\H\Xm')^{-1}{\Xm'}^\H$, we have that
\begin{align}
d(\Xm\to{\Xm'}) 
%&= \trace\left( PT( \Id_T+ PT{\Xm'}{\Xm'}^\H )^{-1} \Xm \Xm^\H \right) \\
&= 
\trace\big(\big(\Id_T- {\Xm'} ( \Id_M+ {\Xm'}^\H {\Xm'} )^{-1}
{\Xm'}^\H \big) \Xm \Xm^\H \big) \\ 
&= \trace(\Xm^\H\Xm) - \trace\big(\Xm^\H {\Xm'} ( \Id_M+ {\Xm'}^\H {\Xm'} )^{-1}
{\Xm'}^\H \Xm\big) \\
&= PT\bigg(1 - {\alpha^{-1}_{P,T,M}}\frac{\| {\Xm'}^\H \Xm \|_{\rm F}^2}{(PT)^2}\bigg) , 
\end{align}%
where 
%where we applied the Woodbury identity, and we define
%; we define the constant 
$
{\alpha_{P,T,M} := \tfrac{1}{PT}+\frac{1}{M}}
$
and the last equality follows from ${\Xm}^\H{\Xm} = {\Xm'}^\H{\Xm'} = \frac{PT}{M}\Id_M$.
Therefore, the design criterion~\eqref{eq:criterion_minTrace} is equivalent to
%  finding $\Xcal$ such that the following term is small
$
\Xc %= \arg\max_{\Xcal}d_{\min}(\Xcal) 
=
\arg\displaystyle\min_{\Xcal} \max_{\Xm,{\Xm'}\in\Xcal:\;\Xm\ne{\Xm'}} \| {\Xm'}^\H \Xm \|_{\rm F}^2.
$
%where the metric lies in the interval $\bigl[0,\frac{1}{M}\bigr]$. 
This coincides with the common criterion consisting in maximizing the minimum pairwise chordal distance between the symbol subspaces~\cite{Conway1996packing,Kammoun2007noncoherentCodes,Hoang_cubesplit_journal,Dhillon2008constructing}.

\paragraph{The Multi-User Case}

In the $K$-user case, we have the following bounds on $d_{\min}(\Xc)$.

\begin{proposition}[Bounds on the $d_{\min}(\Xc)$ metric] \label{prop:min_dk}
It holds that 
\begin{align}
\min_{k\in [K]} d_k(\Xc) \le d_{\min}(\Xc) \le \displaystyle\min_{k\in [K]} d_k(\Xc) + (K-1)M, \label{eq:min_dk}
\end{align}
where
\begin{align}
&d_k(\Xc) \defeq \notag \\ 
&\quad \min_{\Xm_k\ne\Xm'_k\in\Xcal_k\atop
{\Xm}_j\in\Xcal_j, j\ne k} \trace\bigg(\Xm_k^\H \Big( \Id_T + \Xm'_k{\Xm'_k}^\H + \sum_{j\ne k}\Xm_j\Xm_j^\H \Big)^{-1} \Xm_k
\bigg). \label{eq:dk_Xc_1}
\end{align}

\end{proposition}
\begin{proof}
See Appendix~\ref{proof:min_dk}.
\end{proof}

\begin{corollary} \label{coro:min_dk_2user}
In the two-user case ($K=2$), it holds that 
\begin{align}
\min\left\{ d_{1}(\Xcal),
d_{2}(\Xcal) \right\} &\le d_{\min}(\Xcal) \notag \\
&\le  \min\left\{
d_{1}(\Xcal), d_{2}(\Xcal) \right\} + M,  \label{eq:tmp736} %\\
%d_{\min}(\Xcal) &\ge \min\left\{ d_{1}(\Xcal), d_{2}(\Xcal),
%d_{1}(\Xcal)+d_{2}(\Xcal) \right\} = \min\left\{ d_{1}(\Xcal),
%d_{2}(\Xcal) \right\}. 
\end{align}
where
\begin{align} 
&d_{1}(\Xcal) := \notag \\ 
&\quad  \min_{\Xm_1\ne\Xm'_1\in\Xcal_1,
{\Xm}_2\in\Xcal_2} \trace\Big(\Xm_1^\H \big( \Id_T + \Xm'_1 {\Xm'}_1^\H + {\Xm}_2 {\Xm}_2^\H \big)^{-1} \Xm_1
\Big), \label{eq:d12} \\
%\end{align}%
%\begin{align}
&d_{2}(\Xcal) := \notag \\ 
&\quad  \min_{\Xm_2\ne\Xm'_2\in\Xcal_2,
\Xm_1\in\Xcal_1} \trace\Big(\Xm_2^\H \big( \Id_T + {\Xm}_1 {\Xm}_1^\H + \Xm'_2 {\Xm'}_2^\H \big)^{-1} \Xm_2
\Big). \label{eq:d21}
\end{align}%
\end{corollary}

Proposition~\ref{prop:min_dk} says that $d_{\min}(\Xc)$ is within a constant gap to $\min_{k\in [K]} d_k(\Xc)$, and thus $d_{\min}(\Xcal)$ scales linearly with $P$ when $P$ is large if and only if $\min_{k\in [K]} d_k(\Xc)$ does so. Based on this observation, we propose the
following design criterion
%% 
%% \fbox{
%% \begin{minipage}{\textwidth}
\begin{align} \label{eq:criterion_minAltTrace_Kuser}
\Xcal^* = \arg\max_{\Xcal} \ \min_{k\in [K]} d_k(\Xc). 
\end{align}%
{This criterion is the basis for the simple constellation construction presented in Section~\ref{sec:constructions}.}

\subsection{Practical Approaches to Numerical Optimization}
In this section, given the proposed criteria, we present two practical approaches to reduce the complexity of the constellation optimization using any metric.  
\subsubsection{Alternating Optimization}
To simplify the constellation optimization, we propose an {\em alternating optimization} approach as follows.
First $\{\Xc_k\}_{k=1}^K$ are initialized. Then, for $k = 1,\dots, K$, we iteratively optimize $\Xc_k$ by $\Xcal^*_k = \arg\displaystyle\max_{\Xcal_k} m(\Xc)$ for fixed $\{\Xc_l\}_{l\ne k}$ in a round robin manner, where $m(\Xc)$ is the considered metric. At each iteration, it has fewer variables to optimize than directly solving~\eqref{eq:criterion_minMean}, \eqref{eq:criterion_minTrace}, or \eqref{eq:criterion_minAltTrace_Kuser}. %We refer to this as {\em alternating optimization}.
{Since the objective function is nondecreasing across iterations, the solution of alternating optimization converges to a local minimum.}

\subsubsection{Solution Space Reduction}
In the most general setting,
the simplified criteria~\eqref{eq:criterion_minMean}, \eqref{eq:criterion_minTrace}, \eqref{eq:criterion_minAltTrace_Kuser} still have a large solution space. Specifically, $\Xc$ belongs to the product space $$\Bigg\{\Xm_k^{(1)}, \dots,
\Xm_k^{(|\Xc_k|)} \in \CC^{T\times M_k} \colon \frac{1}{|\Xc_k|} \sum_{i=1}^{|\Xc_k|} \big\|\Xm_k^{(i)}\big\|_{{\rm F}}^2\le PT\Bigg\},$$ and thus has $\prod_{k=1}^{K}(T M_k)^{|\Xc_k|}$ free variables to optimize. To reduce the solution space, we make the suboptimal assumption that the individual constellations $\Xc_k$ follow from {USTM,} i.e., they contain scaled-truncated-unitary-matrix symbols.
From \sheng{a practical perspective}, this is desirable since the constellation is oblivious to the presence of the other users and {USTM} is high-SNR optimal, or near optimal, for the single-user channel. {Furthermore, it was shown in \cite{Devassy2015finite} that letting each user employ {USTM} independently from the other users entails a small loss in terms of sum capacity for the noncoherent MIMO MAC even at moderate SNR.}
Under this assumption, we let 
%the symbol $\Xm_k$ be a scaled truncated unitary matrix representative of a point on the Grassmann manifold $G(\CC^T,M)$, i.e.,
$
\Xm_k^\H \Xm_k = \frac{P_kT}{T} \Id_M, ~\forall \Xm_k \in \Xc_k, k\in [K].
$
%That is, the symbol $\Xm_k$ is a scaled truncated unitary matrix representative of a point on the Grassmann manifold $G(\CC^T,M)$.
Thus, the solution space is reduced to the Cartesian product of $\sum_{k=1}^{K}|\Xc_k|$ instances of the set of truncated unitary matrices (for the signal subspace) and $K$ instances of the interval $[0,P]$ (for the signal power). Furthermore, {we can choose to optimize} the signal subspace and power separately. {Specifically, using the proposed metrics, we first optimize the signal subspace for given transmit power, and then optimize the power for given signal subspace. In the following, we consider each problem.} %, as shown in the next sections. 

%\vspace{-.1cm}
\section{{A} Simple Construction for Fixed Transmit Power} \label{sec:constructions}
%\vspace{-.2cm}
In this section, inspired by the proposed criteria, we propose {a} simple constellation construction for fixed powers $\{P_k\}_{k\in [K]}.$\footnote{{In Appendix~\ref{sec:precoding}, we provide another simple constellation construction based on precoding, which is a generalization of our design for the SIMO MAC in~\cite{HoangAsilomar2018multipleAccess}.}} %We also analyze our previously proposed constellation construction using the proposed metrics.
%\vspace{-.4cm}
%\subsection{Constellation Design Based on Partitioning} \label{sec:partitioning}
%\vspace{-.2cm}
We consider the symmetrical power case $P_k = P, \; \forall k \in [K]$. This is a reasonable assumption if the rates are symmetric $R_1 = \dots = R_K$. Also, {following USTM,} we let $\Xm_k^\H\Xm_k = \frac{PT}{M} \Id_M$,
$\forall\,\Xm_k\in\Xcal_k$, $k\in[K]$. Nevertheless, there must be constraints
between the symbols of different users. For instance, if the
constellations are such that $\Xm_1=\Xm_2$ can occur, then $d_{k}(\Xcal)$ is upper-bounded by a constant for any $k$ and any $P$. 
%	Therefore,
%	$\Xcal_1$ and $\Xcal_2$ should not have identical symbols. In fact,
%	\textbf{symbols from different users should have a comparable 
%		$d$-value as those from the same user}. 
This can be developed in a formal
way as follows. 
%Due to the symmetry, it is enough to consider $d_{1}(\Xcal)$.  

By removing the terms inside the
inverse in $d_k(\Xc)$, we obtain an upper bound:
\begin{align}
d_{k}(\Xcal) \le \min&\Big\{ \min_{\Xm_k\ne\Xm'_k\in\Xcal_k} \trace\big(\Xm_k^\H( \Id_T + \Xm'_k{\Xm'_k}^\H)^{-1} \Xm_k \big), \notag \\ 
&\min_{\Xm_k\in\Xcal_k, {\Xm}_l\in\Xcal_l, l\ne k}  \trace\big(\Xm_k^\H (\Id_T + {\Xm}_l \Xm_l^\H )^{-1} \Xm_k \big) \Big\} \label{eq:dk_ub}.
\end{align}
%\begin{align}
%d_{1}(\Xcal) &\le \min\Big\{ \min_{\Xm_1\ne\Xm'_1\in\Xcal_1} \trace\big(\Xm_1 ^\H( \Id_T+ \Xm'_1{\Xm'}_1^\H)^{-1} \Xm_1 \big), \min_{\Xm_1\in\Xcal_1, {\Xm}_2\in\Xcal_2} \trace\big(\Xm_1^\H (\Id_T+ {\Xm}_2 {\Xm}_2^\H )^{-1} \Xm_1 \big) \Big\}, \label{eq:d12_ub} \\
%d_{2}(\Xcal) &\le \min\Big\{ \min_{\Xm_2\ne\Xm'_2\in\Xcal_2}
%\trace\big(\Xm_2^\H ( \Id_T+ \Xm'_2 {\Xm'}_2^\H)^{-1} \Xm_2
%\big), \min_{\Xm_1\in\Xcal_1, {\Xm}_2\in\Xcal_2}
%\trace\big(\Xm_2^\H ( \Id_T+ {\Xm}_1 {\Xm}_1^\H)^{-1} \Xm_2\big)
%\Big\}. \label{eq:d21_ub} 
%\end{align}%
For $d_k(\Xcal)$ to be large,
the upper bound must
be large. This is made precise in the next proposition.

%\vspace{-.3cm}
\begin{proposition}[Necessary condition] \label{prop:necessary}
Let $\{\Xcal_k\}_{k=1}^K$ be such that $\Xm_k^\H\Xm_k = \frac{PT}{M} \Id_M$,
$\forall\,\Xm_k\in\Xcal_k$, $k\in[K]$. If the following lower bound on the
$d$-values holds for some $c\in[0,1/M]$
\begin{equation}
\min_{k\in [K]} d_k(\Xcal) \ge 
PT\left(1 - {\alpha^{-1}_{P,T,M}} \, c \right), \label{eq:necc0}
\end{equation}%
where ${\alpha_{P,T,M} := \frac{1}{PT}+\frac{1}{M}}$, then we must have
\begin{multline}
\frac{1}{(PT)^2}\max\Big\{ \max_{\Xm_k\ne\Xm'_k\in\Xcal_k, k \in [K]} \big\|
{\Xm'_k}^\H \Xm_k
\big\|_{\rm F}^2, \\
\max_{\Xm_k\in\Xcal_k, \Xm_l\in\Xcal_l,k\ne l \in [K]} \|
{\Xm}_k^\H\Xm_l \|_{\rm F}^2
\Big\} \le c.
\label{eq:necc}
%% \max_{\Xm_1\ne\Xm'_{1}\in\Xcal_1} \left\|
%%   \frac{\Xm'_1}{\|\Xm'_1\|} \frac{\Xm_1^\H}{\|\Xm_1\|}
%%   \right\| &\le
%%   c, \label{eq:necc1}\\
%% \max_{\Xm_2\ne\Xm'_{2}\in\Xcal_2} \left\|
%%   \frac{\Xm'_2}{\|\Xm'_2\|} \frac{\Xm_2^\H}{\|\Xm_2\|}
%%   \right\| &\le
%% c,\\
%% \max_{\Xm_1\in\Xcal_1, \Xm_2\in\Xcal_2} \left\|
%% \frac{{\Xm}_2}{\|\Xm_2\|} \frac{\Xm_1^\H}{\|\Xm_1\|} \right\| &\le
%% c.\label{eq:necc3}
\end{multline}%
\end{proposition}
%\vspace{-.2cm}
\begin{proof}
The proof follows the same steps as the single-user case in Section~\ref{sec:single_user}, applying
to the upper bound \eqref{eq:dk_ub}.  
\end{proof}
%\vspace{-.2cm}
The above proposition shows that symbol pairs from
different users should fulfill similar distance criteria as symbol pairs
from the same user when it comes to identifiability conditions. However, it is unclear whether
\eqref{eq:necc} alone is enough to guarantee a large value of $d_{\min}(\Xc)$. 
In the following, we shall show that these conditions are indeed sufficient if $c$ is small.  
\begin{proposition}[Sufficient condition] \label{prop:sufficient}
Let $\{\Xcal_k\}_{k=1}^K$ be such that {$\trace[\Xm_k^\H\Xm_k] = PT$,}
$\forall\,\Xm_k\in\Xcal_k$, $k\in[K]$. 
If
\begin{multline}
\frac{1}{(PT)^2}\max\Big\{ \max_{\Xm_k\ne\Xm'_k\in\Xcal_k, k \in [K]} \big\|
{\Xm'_k}^\H \Xm_k
\big\|_{\rm F}^2, \\
\max_{\Xm_k\in\Xcal_k, \Xm_l\in\Xcal_l,k\ne l \in [K]} \|
{\Xm}_k^\H\Xm_l \|_{\rm F}^2
\Big\} \le c
\end{multline}%
for some $c\in[0,1/M]$, then we have 
\begin{equation}
\min_{k\in [K]} d_k(\Xcal) \ge 
PT\Bigg(1 - K \bigg( {\alpha_{P,T,M}}-\sqrt{\frac{K(K-1)c}{2^{\ind{K=2}}}} \bigg)^{-1}
c \Bigg). \label{eq:suff2}
\end{equation}%
%Furthermore, if $c < \left[\left(\frac{1}{2PT} + \frac{1}{2M} +\frac{1}{16}\right)^{1/2}-\frac{1}{4}\right]^2$, then $\min\left\{ d_{1}(\Xcal), d_{2}(\Xcal) \right\}$ scales linearly with the power $P$.
\end{proposition}
\begin{proof}
See Appendix~\ref{proof:sufficient}.
%the extended version~\cite{Hoang2020arXiv_MAC}.
\end{proof}
%\vspace{-.2cm}

{\begin{remark}
Proposition~\eqref{prop:sufficient} only requires the joint constellation to satisfy $\trace[\Xm_k^\H\Xm_k] = PT$ rather than $\Xm_k^\H\Xm_k = \frac{PT}{M} \Id_M$ for all
$\Xm_k\in\Xcal_k$, i.e., the joint constellation does not necessarily follow USTM.
\end{remark}}

{The two propositions above give necessary and sufficient conditions for the metric $\min_{k\in [K]} d_k(\Xcal)$ to scale linearly with $P$. The joint constellation attains a high value of this metric if and only if every pair of individual symbols either from the same user or different users are well separated in terms of the chordal distance. This is illustrated for the two-user case in Fig.~\ref{fig:dist_condition}. These propositions}
motivate the following simplified design criterion
\begin{align} 
\Xc^* = \arg\min_{\Xcal} \ 
\max&\Big\{ \max_{\Xm_k\ne\Xm'_k\in\Xcal_k, k \in [K]} \big\|
{\Xm'_k}^\H \Xm_k
\big\|_{\rm F}^2, \notag \\
&\max_{\Xm_k\in\Xcal_k, \Xm_l\in\Xcal_l,k\ne l \in [K]} \|
{\Xm}_k^\H\Xm_l \|_{\rm F}^2
\Big\}. \quad \label{eq:criterion_minChordal}
\end{align}
%Intuitively, a pair of symbols from different users ($\Xm_k$ and $\Xm_l$) should have a comparable distance as those from the same user ($\Xm_k$ and $\Xm'_k$).
\begin{figure}[t!]
\vspace{-2cm}
\centering
\includegraphics[width=.45\textwidth]{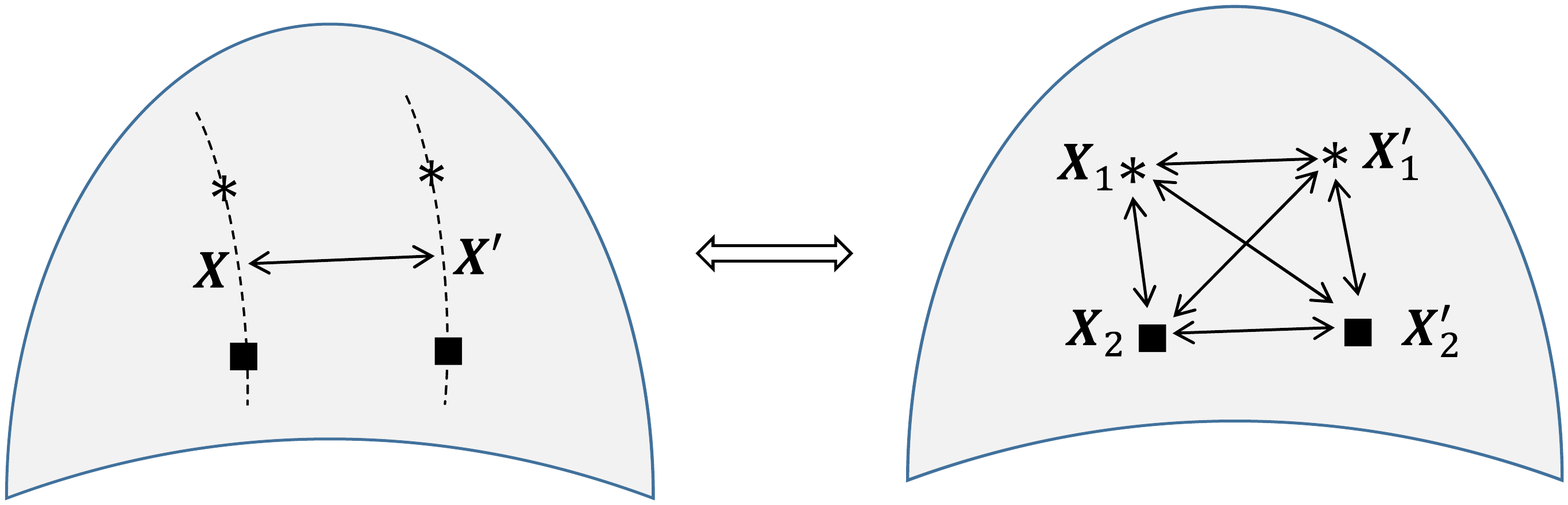}
\vspace{-1.7cm}
\caption{{Illustration of the necessary and sufficient conditions for the metric $\min_{k\in [K]} d_k(\Xcal)$ to scale linearly with $P$ in the two-user case. A pair of joint symbols $\Xm = [\Xm_1 \ \Xm_2]$ and $\Xm' = [\Xm_1' \ \Xm_2']$ attains a high value of the $d(\Xm\to{\Xm'})$ metric if and only if each pair of individual symbols among $\{\Xm_1,\Xm_2,\Xm'_1,\Xm_2'\}$ are well separated in terms of the chordal distance.}}
\label{fig:dist_condition}
%\vspace{-.8cm}
\end{figure}

%%\vspace{-.6cm}
%\subsection{Partitioning Design} 
%%\vspace{-.2cm}
Based on~\eqref{eq:criterion_minChordal}, we propose a simple construction as follows. Let $\Xcal_{\textnormal{SU}}$ be a single-user constellation and let
$
c \defeq \frac{1}{(PT)^2}\displaystyle\max_{\Xm \ne {\Xm'} \in \Xc_{\textnormal{SU}}}  \|
{\Xm'}^\H \Xm \|_{\rm F}^2
\in\bigl[0,\tfrac{1}{M}\bigr]$. 
%\IncMargin{.2em}
%\begin{algorithm}[h]
%	%	\setstretch{1.2}
%	\SetKwData{Left}{left}\SetKwData{This}{this}\SetKwData{Up}{up}
%	\SetKwFunction{Union}{Union}\SetKwFunction{FindCompress}{FindCompress}
%	\SetKwInOut{Input}{input}\SetKwInOut{Output}{output}
%	
%	\SetKwRepeat{Repeat}{repeat}{until}%
%	\SetAlgoLined
%	%	\Indm
%	%	\Input{}
%	%	\Output{}
%	%	\Indp
%	%	\Indm
%	\KwIn{$T$, $P$, $M$, $|\Xc_1|$, $|\Xc_2|$}
%	%	\Indp
%	generate a single-user Grassmannian constellation $\Xc_{\textnormal{SU}}$ of size $|\Xc_1| + |\Xc_2|$ by solving $$
%	\min_{\Xc_{\textnormal{SU}} \subset G(\CC^T,M)} \max_{\Xm \ne {\Xm'} \in \Xc_{\textnormal{SU}}} \big\| {\Xm'}^\H \Xm \big\|_{\rm F}^2
%	$$ \;
%	partition $\Xc_{\textnormal{SU}}$ into disjoint subsets $\Xc_1$ of size $|\Xc_1|$ and $\Xc_2$ of size $|\Xc_2|$ \;
%	\KwRet{\em $\Xc_1$ and $\Xc_2$}
%	\caption{Partitioning a singe-user constellation}
%	\label{construction:partitioning}
%\end{algorithm}
%\DecMargin{.2em}
%
%
%%\subsection{Construction based on partitioning a single-user constellation}
%
%If the single-user constellation $\Xcal_{\textnormal{SU}}$ satisfy 
%$
%\max_{\Xm \ne {\Xm'} \in \Xc_{\textnormal{SU}}} \big\| {\Xm'}^\H \Xm \big\|_{\rm F}^2 \le c,
%$
%for some $c\in\bigl[0,\frac{1}{M}\bigr]$. 
%That is, 
%\begin{equation}
%d_{\min}(\Xcal_{\textnormal{SU}}) \ge PT\left(1 - \alpha_{P,T,M} \, c \right).
%\end{equation}%
We can generate $\{\Xc_k\}_{k=1}^K$ by partitioning  $\Xcal_{\textnormal{SU}}$ into $K$ disjoint subsets. Then, from \eqref{eq:dmin_lb} and Proposition~\ref{prop:sufficient}, we can
guarantee  
\begin{equation}
d_{\min}(\Xcal) \ge 
PT\Bigg(1 - K \bigg( {\alpha_{P,T,M}}-\sqrt{\frac{K(K-1)c}{2^{\ind{K=2}}}} \bigg)^{-1}
c \Bigg). \label{eq:tmp833}
\end{equation}%
With such a construction, the joint constellation design
problem becomes essentially an individual constellation design problem. A random partition suffices to guarantee \eqref{eq:tmp833}, although one can smartly partition the set
$\Xcal_{\textnormal{SU}}$ to improve over \eqref{eq:tmp833}. The optimal partition problem is equivalent to a min-max graph partitioning~\cite{Bulucc2016recent}. Note that for the right-hand side of \eqref{eq:tmp833} to scale linearly with $P$, $c$ must be small enough, which requires the initial single-user constellation $\Xc_{\textnormal{SU}}$ to be sparse enough in $G(\CC^T,M)$ and thus limits the size of $\Xcal_{\textnormal{SU}}$. This is made precise in the following.
%The following proposition gives a necessary condition for $d_{\min}(\Xcal_1,\Xcal_2)$ to scale linearly with $P$.

%\vspace{-.2cm}
\begin{proposition}[Requirement for the single-user constellation $\Xc_{\rm SU}$] \label{prop:lowerBound_minDist_Xsu}
Consider a joint constellation $\Xc$ %$\{\Xc_k\}_{k=1}^K$ 
generated by partitioning a single-user constellation $\Xc_{\textnormal{SU}}$. For the lower bound of $d_{\min}(\Xc)$ in \eqref{eq:tmp833} to scale linearly with $P$, the minimum pairwise chordal distance between elements of $\Xc_{\textnormal{SU}}$, i.e. $\delta_{\min}(\Xc_{\textnormal{SU}}) \defeq
\displaystyle\min_{\Xm \ne {\Xm'} \in \Xc_{\textnormal{SU}}} \sqrt{M-\tfrac{1}{P^2T^2}\|{\Xm'}^\H \Xm\|_{\rm F}^2}$, must satisfy
\begin{align}
\delta_{\min}(\Xc_{\textnormal{SU}}) 
%	&\defeq 
%	\min_{\Xm \ne {\Xm'} \in \Xc_{\textnormal{SU}}} \sqrt{M-\big\| {\Xm'}^\H \Xm \big\|_{\rm F}^2} \\
%	&\ge \sqrt{M-c} \\
&> \sqrt{M-\Big[\Big({\frac{\alpha_{P,T,M}}{K} +\phi_K}\Big)^{\frac12}-{\phi_K^\frac12}\Big]^2}, \label{eq:tmp867}
\end{align}
{where $\phi_K \defeq \frac{K-1}{4K2^{\ind{K=2}}}$. The condition~\eqref{eq:tmp867} implies that} the cardinality of $\Xc_{\textnormal{SU}}$ is bounded as
\begin{align}
|\Xc_{\textnormal{SU}}| &<{\kappa^{-1}_{T,M}} 2^{2M(T-M)} \notag \\
&\quad \cdot \left(M-\left[{\left({\frac{\alpha_{P,T,M}}{K}} +{\phi_K}\right)}^\frac12-{\phi_K^\frac12}\right]^2\right)^{-M(T-M)} \label{eq:tmp852},
%	\\
%	&\overset{P\to \infty}{\longrightarrow} \beta(T,K,M) \defeq c^{-1}_{T,M} 2^{2M(T-M)} \Bigg(M-\bigg[\left(\frac{1}{KM} +\frac{K-1}{4K2^{\ind{K=2}}}\right)^\frac12-\sqrt{\frac{K-1}{4K2^{\ind{K=2}}}}\bigg]^2\Bigg)^{-M(T-M)}, \label{eq:tmp853}
\end{align} 
with
\begin{align}
&{\kappa_{T,M}} \defeq \notag \\
&~~ \frac{1}{(M(T\!-\!M))!} \displaystyle\prod_{i=1}^{\min\{M,T-M\}} \frac{(T-i)!}{(\min\{M,T\!-\!M\}-i)!}. \quad   \label{eq:c(T,M)}
\end{align}
\end{proposition}
\begin{proof}
The right-hand side of \eqref{eq:tmp833} scales linearly with $P$ if $1 - K \Big( {\alpha_{P,T,M}}-\sqrt{\frac{K(K-1)c}{2^{\ind{K=2}}}} \Big)^{-1}
c > 0$, i.e.,
$
c <  \left[\left({\frac{\alpha_{P,T,M}}{K}} +{\phi_K}\right)^{\frac12}-{\phi_K^\frac12}\right]^2. %\label{eq:tmp912}
$
%	Let $\delta_{\min}(\Xc_{\textnormal{SU}}) \defeq 
%	\min_{\Xm \ne {\Xm'} \in \Xc_{\textnormal{SU}}} \sqrt{M-\frac{1}{P^2T^2}\|{\Xm'}^\H \Xm\|_{\rm F}^2}$ be the minimum pairwise chordal distance between symbols in $\Xc_{\textnormal{SU}}$. 
This is equivalent to \eqref{eq:tmp867} since $\delta_{\min}(\Xc_{\textnormal{SU}}) = \sqrt{M-c}$ by definition.
%	\begin{align}
%	\delta_{\min}(\Xc_{\textnormal{SU}}) 
%%	&\defeq 
%%	\min_{\Xm \ne {\Xm'} \in \Xc_{\textnormal{SU}}} \sqrt{M-\big\| {\Xm'}^\H \Xm \big\|_{\rm F}^2} \\
%%	&\ge \sqrt{M-c} \\
%	&> \sqrt{M-\left[\left(\frac{1}{KPT} + \frac{1}{KM} +\frac{K-1}{4K2^{\ind{K=2}}}\right)^{1/2}-\sqrt{\frac{K-1}{4K2^{\ind{K=2}}}}\right]^2}. \label{eq:tmp867}
%	\end{align}
%That is to say, the single-user constellation $\Xc_{\textnormal{SU}}$ must be sparse enough in $G(\CC^T,M)$. 
On the other hand, according to~\cite[Corollary 1]{Dai2008quantizationBounds}, the volume of a metric ball $\Bc(\delta)$ of radius $\delta$ (in chordal distance) in $G(\CC^T,M)$ with the normalized invariant measure $\mu(\cdot)$ is given by
$
\mu(\Bc(\delta)) = {\kappa_{T,M}} \delta^{2M(T-M)}
$
with ${\kappa_{T,M}}$ defined in~\eqref{eq:c(T,M)}. Since $\Xc_{\textnormal{SU}}$ is a packing on $G(\CC^T,M)$ with minimum chordal distance $\delta_{\min}(\Xc_{\textnormal{SU}})$, 
the Hamming upper bound~\cite[Eq.~(3)]{Dai2008quantizationBounds} yields %$\Xc_{\textnormal{SU}}$ as 
$|\Xc_{\textnormal{SU}}| \le \frac{1}{\mu(\Bc(\delta_{\min}(\Xc_{\textnormal{SU}})/2))}$. From this and \eqref{eq:tmp867}, we obtain \eqref{eq:tmp852}. %, and \eqref{eq:tmp853} follows readily. 
\end{proof}

%Proposition~\ref{prop:lowerBound_minDist_Xsu} says that, to guarantee the linear scaling of $d_{\min}(\Xc)$ through the lower bound in \eqref{eq:tmp833}, the initial single-user constellation $\Xc_{\textnormal{SU}}$ must be sparse enough in $G(\CC^T,M)$. 
%In particular, when $M=1$, the upper bounds in Proposition~\ref{prop:upperBound_Xsu} reads
%\begin{align}
%|\Xc_{\textnormal{SU}}| &< \left(1-\left[\left(\frac{1}{2PT} + \frac{9}{16}\right)^{1/2}-\frac{1}{4}\right]^2\right)^{1-T} \\
%&\overset{P\to \infty}{\longrightarrow} \left(\frac{16}{3}\right)^{T-1}. \label{eq:const_size_bound_M1}
%\end{align}

At high SNR ($P\to \infty$), the bounds on $\delta_{\min}(\Xc_{\textnormal{SU}})$ in \eqref{eq:tmp867} and $|\Xc_{\textnormal{SU}}|$ in \eqref{eq:tmp852} converge to
\begin{align}
\nu(K,M) &\defeq \sqrt{M-\Big[\Big(\frac{1}{KM} +{\phi_K}\Big)^{\frac12}-{\phi_K^\frac12}\Big]^2}
\end{align}
and
\begin{align}
\beta(T,K,M) &\defeq c^{-1}_{T,M} 2^{2M(T-M)} \notag \\
&\quad \cdot \bigg(M\!-\!\Big[\Big(\frac{1}{KM} \!+\!{\phi_K}\Big)^\frac12\!-\!{\phi_K^\frac12}\Big]^2\bigg)^{-M(T-M)},
\label{eq:tmp1079}
\end{align}
respectively.
Fig.~\ref{fig:bound_size_Xsu} shows the values of $\log_2(\beta(T,K,M))$, which is the high-SNR upper bound on the number of bits per symbol $\log_2(|\Xc_{\textnormal{SU}}|)$ in $\Xc_{\textnormal{SU}}$, for $K=4$ and some values of $T$ and $M$. As can be seen, for a fixed $M$, the bound monotonically increases with $T$; for a fixed $T$, the bound first increases with $M$ then decreases after a peak value and becomes $0$ (imposing a zero transmission rate) when $M \approx 0.73 T$. %Note that the bound is not sensitive to the value of $K$ for the considered parameters. 
\begin{figure}[t!]
\vspace{-.4cm}
\centering
%	\hspace{-.4cm}
%	\subfigure[$K = 2$]{\includegraphics[width=.47\textwidth]{fig/beta_TKM_K2_2D.eps}}
%	\hspace{-1cm}
%	\subfigure[$K = 4$]{
\includegraphics[width=.5\textwidth]{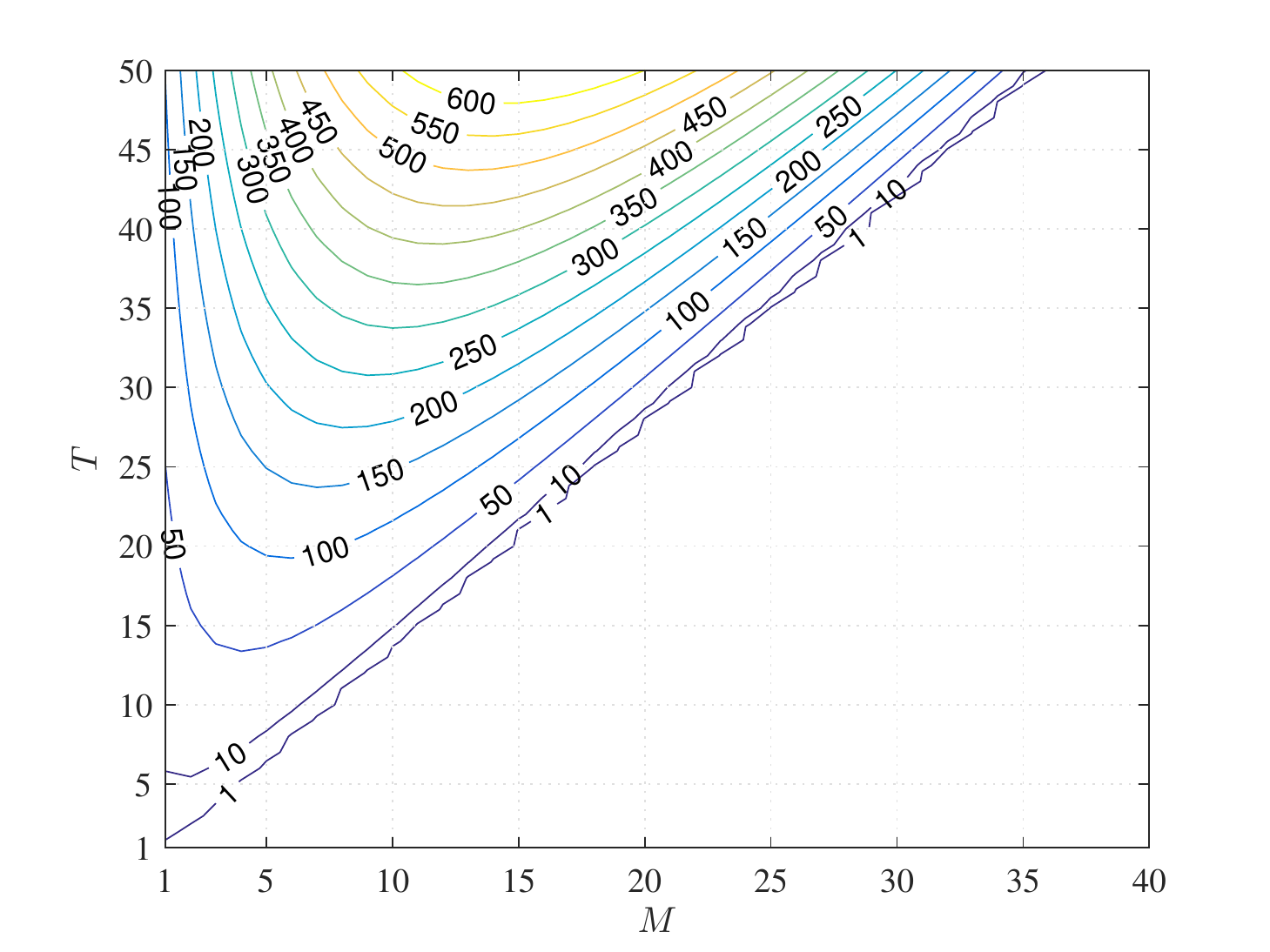}
%}
%	\hspace{-.6cm}

\vspace{-.4cm}
\caption{The upper bound $\log_2(\beta(T,K,M))$ on the number of bits per symbol $\log_2(|\Xc_{\textnormal{SU}}|)$ necessary for the lower bound \eqref{eq:tmp833} of $d_{\min}(\Xc)$ to scale linearly with $P$ with $K=4$.}
\label{fig:bound_size_Xsu}
%\vspace{-.8cm}
\end{figure}

\begin{remark}
The Grassmann manifold $G(\CC^T,M)$ has $2M(T-M)$ real dimensions. From \eqref{eq:tmp1079}, an upper bound on the number of bits per real dimension for $\Xc_{\textnormal{SU}}$ is given by
\begin{equation} \notag
\frac{\log_2 \beta(T,K,M)}{2M(T-M)} \le \zeta(K,M) 
\end{equation}
where
\begin{equation} \notag
\zeta(K,M) \defeq 1 - \frac{1}{2}\log_2\bigg(1-\frac{1}{M}\bigg[\Big(\frac{1}{KM} \!+\!{\phi_K}\Big)^{\frac12}-{\phi_K^\frac12}\bigg]^2\bigg). %\\
%	&< \log_2 3\approx 1.585,
\end{equation}
In fact, using Stirling's formula $\sqrt{2\pi} n^{n+1/2}e^{-n} \le n! \le e n^{n+1/2}e^{-n}$~\cite{Robbins1955remark_Stirling}, we can show that $\frac{\log_2 \beta(T,M)}{2M(T-M)} \uparrow \zeta(K,M)$ as $T\to \infty$, where ``$\uparrow$'' means ``approach from below''. After some simple manipulations, we have that $\zeta(K,M) \le 2- \frac{1}{2}\log_2 3 < \log_2 3$ for any $K$ and $M$. That is, roughly speaking, one should not pack more than $2$ symbols of $\Xc_{\textnormal{SU}}$ in each real dimension of the manifold in average if the partitioning approach is used. %This guarantees that no triple of symbols are in the same subspace. 
\end{remark}

\section{Power Optimization} \label{sec:power_opt}
%\vspace{-.3cm}
When the users transmit at different rates, letting the users transmit at equal power might not be optimal. For example, in the extreme case where only one of the users transmits at non-zero rate, all other users should remain silent, i.e., transmit at {zero power,} to avoid causing interference. Therefore, power optimization also plays a key role. For a fixed constellation $\Xc$ (possibly generated with equal transmit power), let us now consider the problem of optimizing the transmit power so as to maximize the proposed metrics.\footnote{{Note that the power optimization can be used not only to further optimize a joint constellation whose signal subspaces have been optimized, but also to improve any given joint constellation.}}
%\subsubsection{The Two-User Case}

{Let us first focus on the two-user case with per-user USTM. For convenience, we write the constellation symbols as truncated unitary matrices scaled with the transmit powers $\{P_1,P_2\}$, that is, $\Xc_k = \Big\{\sqrt{\frac{P_k T}{M}}\Xm_k\of{i} : \Xm_k\ofH{i}\Xm_k\of{i} = \Id_M, i \in [|\Xc_k|] \Big\}, ~k = 1,2.$ Let $\bar{\Xc} = \bar{\Xc}_1 \times \bar{\Xc}_2$ where $\bar{\Xc}_k = \{\Xm_k^{(i)}\}_{i=1}^{|\Xc_k|}$ is the set of the {\em normalized} symbols of user $k$. We assume that $\bar{\Xc}$ is fixed and would like to optimize the transmit powers $\{P_1,P_2\}$. To this end, we define $\theta \defeq P_2/P_1$, denote $\Xc$ as $\Xc^\theta$ for convenience, and seek to optimize $\theta$ as
\begin{align} \label{eq:opt_theta}
\theta^* = \arg \max_{\theta} m(\Xc^\theta)
\end{align}	
where $m(\Xc)$ is the considered metric. Recall that we assume w.l.o.g. that $\max_{k} P_k = P$. The optimal value of $\theta$ cannot be found in closed-form in general. We propose a procedure to optimize $\theta$ as follows.
\begin{enumerate}
\item Let $P_1 = P$, i.e., user~$1$ transmits at full power, optimize $\theta$ as in \eqref{eq:opt_theta} with the constraint $\theta \in [0,1]$. Let $\hat{\theta}$ be the optimal value.
\item Let $P_2 = P$, i.e., user~$2$ transmits at full power, optimize $\theta$\textemdash or equivalently $\frac{1}{\theta}$\textemdash as in \eqref{eq:opt_theta} with the constraint $\frac{1}{\theta} \in [0,1]$. Let $\breve{\theta}$ be the optimal value.
\item The optimal value of $\theta$ is given by\footnote{{In the numerical result in the next section, we shall see that it is favorable to let the user with higher transmission rate transmit at full power $P$ in the considered setting.}} $\arg\max_{\theta \in \{\hat{\theta},\breve{\theta}\}} m(\Xc^\theta)$.
\end{enumerate}
In Steps 1 and 2, one finds an extremum of the metric over $\theta$ or $\frac{1}{\theta}$ inside the interval $[0,1]$. Well-known extremum search algorithms, such as the golden-section search, can be employed. In the SIMO case, the optimization of $\theta$ in these steps can be done more efficiently as follows. In this case, the individual constellations are $\Xc_k = \{\sqrt{P_k T}\xv_k^{(i)} : \|\xv_k^{(i)}\| = 1, i \in [|\Xc_k|] \}, ~k = 1,2.$ The set of the {\em normalized} symbols are $\bar{\Xc}_k = \{\xv_k^{(i)}\}_{i=1}^{|\Xc_k|}$, $k=1,2$. Consider the metric $d_{\min}(\Xc)$.} %Recall that we define $\theta \defeq P_2/P_1$.
From Corollary~\ref{coro:min_dk_2user}, we deduce that $\min\{d_1(\Xc^\theta),d_2(\Xc^\theta)\} \le d_{\min}(\Xc^\theta) \le \min\{d_1(\Xc^\theta),d_2(\Xc^\theta)\} + 1$ where $d_1(\Xc^\theta) = \min_{\xv_1,\xv'_1,\xv_2} \delta_1(\theta, \xv_1,\xv'_1,\xv_2)$ and $d_2(\Xc^\theta) = \min_{\hat{\xv}_1,\hat{\xv}_2,\hat{\xv}'_2} \delta_2(\theta, \hat{\xv}_1,\hat{\xv}_2,\hat{\xv}'_2)$ with 
\begin{align}
&\delta_1(\theta, \xv_1,\xv'_1,\xv_2) \defeq \notag \\
&\quad P_1T\xv_1^\H(\Id_T + P_1T \xv'_1 {\xv'}_1^\H + \theta P_1 T \xv_2 \xv_2^\H)^{-1} \xv_1, \\ 
&\delta_2(\theta, \hat{\xv}_1,\hat{\xv}_2,\hat{\xv}'_2) \defeq \notag \\
&\quad \theta P_1T \hat{\xv}_2^\H(\Id_T + P_1T\hat{\xv}_1 \hat{\xv}_1^\H + \theta P_1 T\hat{\xv}'_2 {{\hat{\xv}}}_2^{\prime \H}  )^{-1} \hat{\xv}_2,
\end{align}
for $\{\xv_{1}, \xv'_1, \hat{\xv}_1\} \subset \bar{\Xc}_1, \{\xv_2, \hat{\xv}_2, \hat{\xv}'_2\} \subset \bar{\Xc}_2$ such that $\xv_{1} \ne \xv'_1$ and $\hat{\xv}_2 \ne \hat{\xv}'_2$. %Note that $\delta_1(\theta)$ is also a function of $\xv_1$, $\xv_1'$, and $\xv_2$, and $\delta_2(\theta)$ is also a function of $\hat{\xv}_1$, $\hat{\xv}_2$, and $\hat{\xv}'_2$.
The optimal value of $\theta$ can be found by analyzing $d_1(\Xc^\theta)$ and $d_2(\Xc^\theta)$, as stated in the following proposition. 
%\vspace{-.3cm}
\begin{proposition}[Power optimization] \label{prop:theta}
In the two-user SIMO case, the following results hold.
\begin{enumerate} %[leftmargin = *]
%		\item $d_1(\Xc^\theta)$ is monotonically decreasing in $\theta$, while $d_2(\Xc^\theta)$ is strictly increasing in $\theta$ for $\theta \ge 0$;  
%\vspace{-.2cm}
\item %$d_1(\Xc^\theta)$ is monotonically decreasing in $\theta$, while $d_2(\Xc^\theta)$ is strictly increasing in $\theta$ for $\theta \ge 0$; 
$\min\{d_1(\Xc^\theta),d_2(\Xc^\theta)\}$ is maximized at $\theta = {\tilde{\theta}}$ such that $d_1(\Xc^{{\tilde{\theta}}}) = d_2(\Xc^{{\tilde{\theta}}})$, and 
\begin{align} \label{eq:constant_gap_2user}
d_{\rm min}(\Xc^{{\tilde{\theta}}}) \le \max_\theta d_{\rm min}(\Xc^{\theta}) \le d_{\rm min}(\Xc^{{\tilde{\theta}}}) + 1.
\end{align}

\item For each $\{\xv_{1}, \xv'_1, \hat{\xv}_1\} \subset \bar{\Xc}_1$, $\{\xv_2, \hat{\xv}_2, \hat{\xv}'_2\} \subset \bar{\Xc}_2$ such that $\xv_{1} \ne \xv'_1$ and $\hat{\xv}_2 \ne \hat{\xv}'_2$, there exists a unique value of $\theta$ such that $\delta_1(\theta, \xv_1,\xv'_1,\xv_2) = \delta_2(\theta, \hat{\xv}_1,\hat{\xv}_2,\hat{\xv}'_2)$. Denote this value by $\hat{\theta}$ which is implicitly a function of $\xv_{1}, \xv'_1, \hat{\xv}_1,\xv_2, \hat{\xv}_2, \hat{\xv}'_2$. Denote $\delta(\hat{\theta}) \defeq \delta_1(\hat{\theta}, \xv_1,\xv'_1,\xv_2) = \delta_2(\hat{\theta}, \hat{\xv}_1,\hat{\xv}_2,\hat{\xv}'_2)$. It holds that 
%\begin{subequations}
\label{eq:tmp1344}
\begin{align} 
\hat{\theta} &= \frac{1}{3a} \bigg[2\sqrt{\Delta} \notag \\
&\qquad ~~ \cdot \cos\bigg(\frac{1}{3} \arccos \bigg( \frac{9abc -2b^3 - 27a^2d}{2\sqrt{\Delta^3}}\bigg) \bigg) - b\bigg] \label{eq:tmp1344-a} 
\end{align}
with
\begin{align}
a &\defeq P_1T\big[1 + P_1T(1-|{\xv'}_1^\H\xv_2|^2) \big] e_{2}, \\
b &\defeq (1+ P_1T) e_{2} \notag \\
&\quad + \big[1 + P_1T(1-|{\xv'}_1^\H\xv_2|^2) \big] \big[1 + P_1T(1-|\hat{\xv}_1^\H \hat{\xv}_2|^2) \big] \notag \\
&\quad- \big[P_1T + P_1^2T^2(1-|\hat{\xv}_1^\H\hat{\xv}'_2|^2) \big] e_{1},\\ 
c &\defeq - (1+P_1T)e_{1} \notag \\
&\quad - \big[1 + P_1T(1-|\hat{\xv}_1^\H\hat{\xv}'_2|^2) \big] \big[1 + P_1T(1-|\xv_1^\H{\xv'}_1|^2) \big] \notag \\
&\quad + \big(1+\tfrac{1}{P_1T}\big)\big[1 + P_1T(1-|\hat{\xv}_1^\H \hat{\xv}_2|^2) \big], \\
d &\defeq - \big(1+\tfrac{1}{P_1T}\big) \big[1 + P_1T(1-|\xv_1^\H{\xv'}_1|^2) \big], \\
e_{1} &\defeq 1-|\xv_1^\H\xv_2|^2 + P_1T\big[(1-|\xv_1^\H\xv'_1|^2)(1-|{\xv'}_1^\H\xv_2|^2) \notag \\
&\qquad \qquad \qquad \qquad \qquad - |\xv_1^\H\xv'_1{\xv'}_1^\H\xv_2 - \xv_1^\H\xv_2|^2\big], \\
e_{2} &\defeq 1-|\hat{\xv}_2^\H\hat{\xv}'_2|^2 + P_1T\big[(1-|\hat{\xv}_2^\H\hat{\xv}_1|^2)(1-|\hat{\xv}_1^\H\hat{\xv}'_2|^2) \notag \\
&\qquad \qquad \qquad \qquad \qquad - |\hat{\xv}_2^\H\hat{\xv}_1\hat{\xv}_1^\H\hat{\xv}'_2 - \hat{\xv}_2^\H \hat{\xv}'_2|^2\big], \\
\Delta &\defeq b^2 - 3ac,
\end{align}
%\end{subequations}
and that 
\begin{align} \label{eq:tmp1706}
{\tilde{\theta}} = \arg\min_{\hat{\theta} \in \Thetam} \delta(\hat{\theta}),
\end{align}
where  $\Thetam$ is the set of values of $\hat{\theta}$ for all possible 6-tuple of symbols $\{\xv_{1}, \xv'_1, \hat{\xv}_1\} \subset \bar{\Xc}_1$, $\{\xv_2, \hat{\xv}_2, \hat{\xv}'_2\} \subset \bar{\Xc}_2$ such that $\xv_{1} \ne \xv'_1$ and $\hat{\xv}_2 \ne \hat{\xv}'_2$.
%		, where 
%		\begin{align}
%		\Thetam &\defeq \{ \hat{\theta} : \{\xv_{1}, \xv'_1, \hat{\xv}_1\} \subset \bar{\Xc}_1, \{\xv_2, \hat{\xv}_2, \hat{\xv}'_2\} \subset \bar{\Xc}_2, \xv_{1} \ne \xv'_1, \hat{\xv}_2 \ne \hat{\xv}'_2 \} \\
%		&= \{\theta : \exists \{\xv_{1}, \xv'_1, \hat{\xv}_1\} \subset \bar{\Xc}_1, \{\xv_2, \hat{\xv}_2, \hat{\xv}'_2\} \subset \bar{\Xc}_2, \xv_{1} \ne \xv'_1, \hat{\xv}_2 \ne \hat{\xv}'_2~\text{s.t}~\delta_1({\theta}, \xv_1,\xv'_1,\xv_2) = \delta_2({\theta}, \hat{\xv}_1,\hat{\xv}_2,\hat{\xv}'_2)\}.
%		\end{align}
\end{enumerate}
\end{proposition}
%%\vspace{-.7cm}
\begin{proof} 
See Appendix~\ref{proof:prop:theta}.
\end{proof}

%%\vspace{-.3cm}
The first part of Proposition~\ref{prop:theta} says that there exists a unique ${\tilde{\theta}}$ that maximizes $\displaystyle\min_{k\in \{1,2\}} d_k(\Xc^\theta)$, and this ${\tilde{\theta}}$ is also approximately the value of $\theta$ maximizing $d_{\min}(\Xc^\theta)$. The second part states that ${\tilde{\theta}}$ can be found by enumerating the closed-form expression \eqref{eq:tmp1344-a} over the set of normalized symbols $\bar{\Xc}_k$, $k = 1,2.$ This is simpler than enumerating $d_{\min}(\Xc^\theta)$ over the whole range of $\theta$. 
In Fig.~\ref{fig:metrics_vs_P2/P}, we numerically verify Proposition~\ref{prop:theta} by plotting the values of $e_{\min} (\Xc^\theta)$ and $d_{\min}(\Xc^\theta)$, as well as $d_1(\Xc^\theta)$ and $d_2(\Xc^\theta)$, as a function of $\theta$ for $P_1 = 20$~dB and different $\bar{\Xc}$ with $T =4$, $B_1 = 6$, and $B_2 = 2$. 
We see that $d_{\min}(\Xc^\theta)$ is within a constant gap from the minimum of $d_1(\Xc^\theta)$, which decreases with $\theta$, and $d_2(\Xc^\theta)$, which increases with $\theta$. The metric $d_{\min}(\Xc^\theta)$ is maximized approximately at ${\tilde{\theta}}$ such that $d_1(\Xc^{{\tilde{\theta}}}) = d_2(\Xc^{{\tilde{\theta}}})$. These observations agree with \eqref{eq:constant_gap_2user} in Proposition~\ref{prop:theta}. Furthermore, ${\tilde{\theta}}$ is also near the value of $\theta$ that maximizes the metric $e_{\min} (\Xc^\theta)$. {Following Proposition~\ref{prop:theta}, when the metric $d_{\min}(\Xc)$ is considered, $\theta^*$ in \eqref{eq:opt_theta} can be approximated by ${\tilde{\theta}}$ from \eqref{eq:tmp1706}.}

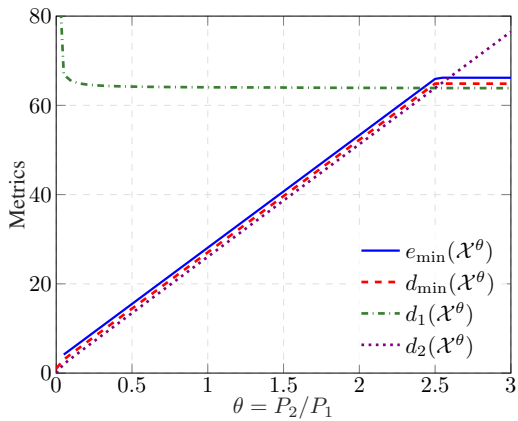
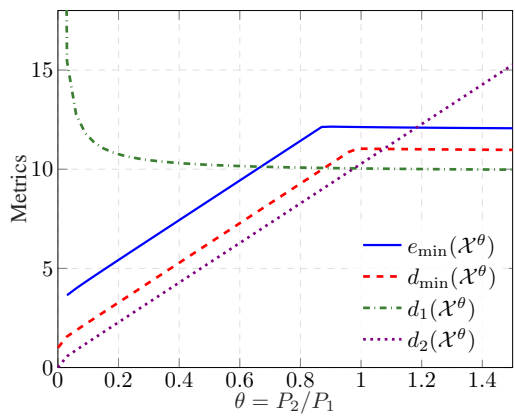
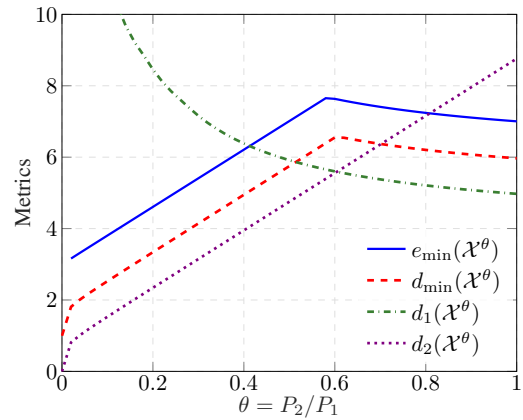
\begin{figure}[t!]
%\vspace{-.7cm}
\centering
\hspace{-.1cm}
\subfigure[$\bar{\Xc}$ obtained by maximizing $d_{\min}(\Xc)$~\eqref{eq:criterion_minTrace} at $30$~dB for both users.]{\input{fig/T4_1B6_2B2_P30dB_minTrace_metricsVSpower_new.tex}}
\hspace{.1cm}
%	\subfigure[$\Xc$ obtained by maximizing $d_{\min}(\Xc)$~\eqref{eq:criterion_minTrace} at $P_1 = 30$~dB, $P_2 = 0.9 P_1$]{\input{fig/T4_1B6_2B2_P30dB_minTraceP2p9_metricsVSpower.tex}}
%	\hspace{-.2cm}
\subfigure[$\bar{\Xc}$ obtained by precoding with Type-II Precoder.]{\input{fig/T4_1B6_2B2_P30dB_precode_metricsVSpower_new.tex}}
\hspace{.1cm}
\subfigure[$\bar{\Xc}$ obtained by partitioning.]{\input{fig/T4_1B6_2B2_P30dB_partition_metricsVSpower_new.tex}}
\hspace{-.3cm}
%\vspace{-.4cm}
\caption{The values of the metrics $e_{\min}(\Xc^\theta)$, $d_{\min}(\Xc^\theta)$, $d_1(\Xc^\theta)$ and $d_2(\Xc^\theta)$ as a function of $\theta$ for $P_1 = 20$~dB, $P_2 = \theta P_1$, $T =4$, $B_1 = 6$, $B_2 = 2$, $M = 1$, and different given normalized constellation $\bar{\Xc}$.}
\label{fig:metrics_vs_P2/P}
%\vspace{-.6cm}
\end{figure}

{In the $K$-user case with $K > 2$, one can use a similar procedure to optimize the power. Specifically, letting one user transmit at full power, one optimizes the fraction of power used by other users. In this case, however, a multidimensional extremum search algorithm should be used, such as the Nelder-Mead method~\cite{nelder1965simplex}. The optimal power allocation is given by the option resulting in the highest metric value. More efficient power optimization methods are open for investigation.}

%-----------------------------------------------
%\vspace{-.4cm}
\section{Numerical Results} \label{sec:performance}

%\subsection{Summary of Design Criteria/Constructions}
We summarize the baseline/proposed design criteria/constructions in Table~\ref{tab:criteria}. %(not including the proposed criteria \eqref{eq:criterion_minAltTrace_Kuser}, \eqref{eq:criterion_minChordal}, and the alternating optimization).
%\vspace{-.2cm}

\begin{table*}[t!]
\begin{center}
\caption{The baseline/proposed joint constellation design criteria/constructions}
			\def\arraystretch{1.8}
\small
%\vspace{-.4cm}
\begin{tabular}{|p{1.2cm}|p{7.1cm}|p{1.85cm}|p{5.2cm}|}
\hline
& \multicolumn{1}{c|}{\textbf{Criterion/Construction}} & \textbf{Shorthand} & \multicolumn{1}{c|}{\textbf{Motivation}}
\\ \hline \hline
\multirow{3}{*}{\bf {Baseline}} & $\Xc^* = \arg \displaystyle\min_{\Xc} m_1(\Xc)$ \eqref{eq:criterion_maxminChordal}  & Min-$m_1$ & \multirow{3}{5.2cm}{Treating $\Xc$ as a single-user constellation for the $\Mtot \times N$ MIMO channel and adapt existing criteria for that channel}   \\ \cline{2-3}
& $\Xc^* = \arg\displaystyle\min_{\Xc} m_2(\Xc)$ \eqref{eq:criterion_minsumDet} \newline (from \cite[Eq.~(8)]{McCloudIT2002signalDesignAndConvCode})  & Min-$m_2$ & \\ 
\cline{2-3}
%				& Each individual symbol comprises pilot and  & Pilot-based & The coherent approach \\
& $\Xc^* = \arg\displaystyle\max_{\Xc} e_{\min}(\Xc)$ \eqref{eq:criterion_minMean} & Max-$e_{\min}$ & %Minimizing the  {KL distance related to the PEP exponent} 
\\ 
\hline \hline
\multirow{6}{*}{\bf Proposed} & {$\Xc^* = \arg\displaystyle\max_{\Xc} b_{\min}(\Xc)$~\eqref{eq:criterion_Riemannian}} & {Max-$b_{\min}$} & \multirow{2}{5cm}{{Minimizing PEP exponent bounds}} \\ 
\cline{2-3}
& {$\Xc^* = \arg\displaystyle\max_{\Xc} J_{s,\min}(\Xc)$~\eqref{eq:criterion_Js}} & {Max-$J_{s, \min}$} & \\ \cline{2-4}
& $\Xc^* = \arg\displaystyle\max_{\Xc} d_{\min}(\Xcal)$  \eqref{eq:criterion_minTrace} & Max-$d_{\min}$ & Maximizing the high-SNR dominant term {in a PEP exponent bound}\\ \cline{2-4}
& $\Xcal^* = \arg\displaystyle\max_{\Xcal} \displaystyle\min_{k\in [K]} d_k(\Xc)$ \eqref{eq:criterion_minAltTrace_Kuser} & & $\min_{k\in [K]} d_k(\Xc)$ is within a constant gap from $d_{\min}(\Xcal)$ \\ \cline{2-4} 
%				& $\Xcal^*_k = \arg\displaystyle\max_{\Xcal_k} m(\Xc)$ for fixed $\{\Xc_l\}_{l\ne k}$ in a round-robin manner for $k \in [K]$, for a metric $m(\Xc)$
%				& Alternating optimization & Reduce the complexity \\ \cline{2-4} 
& $\Xc^* = \arg\displaystyle\min_{\Xcal}
\max\Big\{ \displaystyle\max_{\Xm_k\ne\Xm'_k\in\Xcal_k, k \in [K]} \big\|
{\Xm'_k}^\H  \Xm_k
\big\|_{\rm F}^2,$ 
$~~~~~~~ \displaystyle\max_{\Xm_k\in\Xcal_k, \Xm_l\in\Xcal_l,k\ne l \in [K]} \|
{\Xm}_k^\H\Xm_l \|_{\rm F}^2
\Big\}$ \eqref{eq:criterion_minChordal} & & {Necessary and sufficient conditions for $\min_{k\in [K]} d_k(\Xc)$ to scale linearly with $P$} (Propositions \ref{prop:necessary} and \ref{prop:sufficient}) \\ \cline{2-4} 
& Partitioning a single-user constellation $\Xc_{\textnormal{SU}}$ (Section~\ref{sec:constructions}) & Partitioning & $d_{\min}(\Xcal)$ is large if $\Xc_{\textnormal{SU}}$ is sparse enough\\ \cline{2-4}
& Precoding single-user constellations in $G(\CC^{T-K+1},1)$ {(Appendix \ref{sec:precoding})}& Precoding & Imposing a geometric separation between individual constellations \\ \hline
\end{tabular}
\label{tab:criteria}
\end{center}
%\vspace{-1cm}
\end{table*}
\normalsize

%\vspace{-.2cm}
In the following, {we generate and compare the joint constellations with different design criteria/constructions.} 
%Let us first focus on} the single transmit antenna case ($M_k = M =1$, $k\in [K]$) and
% assume {USTM}, i.e., %each constellation symbol is \sheng{a} unit vector representative of a point in $G(\CC^T,1)$ scaled with $\sqrt{P_kT}$. Specifically, 
% $\Xc_k = \big\{\sqrt{P_kT}\xv_{k}^{(i)}\big\}_{i=1}^{|\Xc_k|}$ with each $\xv_{k}^{(i)}$ being a unit-norm vector representative of a point in $G(\CC^T,1)$, $k \in [K]$, and
% a joint symbol $\Xm \in \Xc$ is formed as $\Xm = [\sqrt{P_1 T}\xv_1 \dots \sqrt{P_K T}\xv_K]$ for $\xv_k \in \Xc_k$. 
{For the partitioning design, we optimize the constellation $\Xc_{\rm SU}$ following the max-min chordal distance criterion, then apply a random partition.} For the precoding design {(see Appendix \ref{sec:precoding}),} we consider %Type II Precoder and 
a {common initial constellation for all users, which is numerically optimized according to the max-min chordal distance criterion.} We will compare our design to the constellations optimized with the criteria Min-$m_1$ \eqref{eq:criterion_maxminChordal}, Min-$m_2$ \eqref{eq:criterion_minsumDet}, {and Max-$e_{\min}$~\eqref{eq:criterion_minMean}} in terms of joint symbol error rate~(SER)~\eqref{eq:joint_SER}. We also consider the joint constellation in which the symbols %in $\Xc_k$ are generated as $\xv_k = \sqrt{P_k} \Big[\sqrt{K}\tilde{\ev}_k^\T \ \sqrt{\frac{T-K}{P_{\rm avg}}}\tilde{\xv}_k^\T\Big]^\T$ where $\tilde{\ev}_k$ is the $k$-th column of $\Id_{K}$ and $\tilde{\xv}_k$ is a vector of scalar symbols in a QAM constellation with average power $P_{\rm avg}$. This corresponds to the scenario where $K$ users transmit 
contain orthogonal pilot sequences followed by spatially multiplexed QAM data symbols. %With this {\em pilot-based constellation}, the receiver uses either an ML detector \eqref{eq:MLdecoder} or a linear {MMSE} detector consisting of MMSE channel estimation, MMSE equalization, and component-wise QAM demapper. %Note that ML detection \eqref{eq:MLdecoder} is used for all noncoherent schemes.
{We use the joint ML detector for all schemes.}

%%\vspace{-.2cm}
\subsection{Numerical Optimization} \label{sec:optimization}
We solve numerically {Max-$J_{1/2,\min}$~\eqref{eq:criterion_Js},} Max-$e_{\min}$~\eqref{eq:criterion_minMean}, Max-$d_{\min}$~\eqref{eq:criterion_minTrace}, Min-$m_1$~\eqref{eq:criterion_maxminChordal}, Min-$m_2$~\eqref{eq:criterion_minsumDet}, and the alternating optimization of the $d_{\min}$ metric for given powers $\{P_k\}$. %(even when the performance of the resulting constellations is benchmarked at other SNR values). 
In general, we want to solve the manifold-constrained optimization 
\begin{align}
\displaystyle\max_{\Xc = \Xc_1 \times \dots \times \Xc_K} \displaystyle\min_{\Xm \ne {\Xm'} \in \Xc} f(\Xm, {\Xm'}), %\label{eq:opt_1}
\end{align}
where $f(\Xm, {\Xm'})$ is given by $J_{1/2}(\Xm,\Xm')$, $\frac{1}{N} \meanLLR$, $d(\Xm\to\Xm')$, and $-\trace\big(\Xm\Xm^\H \Xm' {\Xm'}^\H\big)$ for the Max-$J_{1/2,\min}$, Max-$e_{\min}$, Max-$d_{\min}$, and Min-$m_1$ criteria, respectively. (Note that $\|\Xm\|_{\rm F}^2$ and $\|\Xm'\|_{\rm F}^2$ are constants for given $\{P_k\}$.)
We remark that the objective function is not smooth because of the $\min$.
To smooth it, we use the {well-known} approximation $\max_i x_i \approx \epsilon \ln \sum_i \exp (x_i/\epsilon)$ with a small $\epsilon$ and obtain %a new optimization problem
\begin{align}
\min_{\Xc  = \Xc_1 \times \dots \times \Xc_K} \underbrace{\epsilon \ln \sum_{\Xm \ne {\Xm'} \in \Xc} \exp \Big(-\frac{f(\Xm, {\Xm'})}{\epsilon}\Big)}_{\large {\eqdef g(\Xc)}}. \label{eq:tmp1102}
\end{align}
For Min-$m_2$, the optimization problem is similar to~\eqref{eq:tmp1102} with $g(\Xc)$ replaced by $m_2(\Xc)$.

Each joint constellation symbol $\Xm$ can be seen as a collection of matrix representatives of $K$ points in the Grassmann manifold $G(\CC^T,M)$. }
% This smooth optimization is, however, over multiple points on the set of lines. To tackle this, we construct the matrix $\Cm \defeq \big[\xv_{1}^{(1)} \dots \xv_{1}^{(|\Xc_1|)} \ \dots \ \xv_{K}^{(1)} \dots \xv_{K}^{(|\Xc_2|)}\big] \in \CC^{T\times \sum_{k=1}^K|\Xc_k|}$, then $\Cm$ belongs to the oblique manifold  $\Oc\Bc\big(T,\sum_{k=1}^K|\Xc_k|\big)$ defined as 
%$
%\Oc\Bc(n,m) \defeq \big\{ \Vm = [\vv_1 \dots \vv_m]\in \CC^{n\times m}: \|\vv_1\| = \dots = \|\vv_m\| = 1\big\}.
%$
%The oblique manifold $\Oc\Bc(n,m)$ can be seen as an embedded Riemannian manifold of $\CC^{n\times m}$ endowed with the usual inner product, or as 
%the product manifold of $m$ unit spheres in $\CC^T$~\cite[Sec.~3.4.1]{AbsMahSep2008optManifolds}. Then, the optimization problem \eqref{eq:tmp1102} can be reformulated as an {optimization over a single point} on this oblique manifold as
%\begin{equation}
%\min_{\Cm \in \Oc\Bc\big(T,\sum_{k=1}^K|\Xc_k|\big)} \ \underbrace{\epsilon \ln\sum_{\Xm = \big[\sqrt{P_1T}\xv_{1}^{(i_1)} \ \dots \ \sqrt{P_KT}\xv_{K}^{(i_K)}\big] \atop \ne {\Xm'} = \big[\sqrt{P_1T}\xv_{1}^{(j_1)} \dots \sqrt{P_KT}\xv_{K}^{(j_K)}\big]} \exp \Big(-\frac{f(\Xm, {\Xm'})}{\epsilon}\Big)}_{\large \eqdef g(\Cm)}. \label{eq:opt_oblique}
%\end{equation}
%In Appendix~\ref{app:gradient}, we compute the Riemannian gradient of {$g(\Xc)$.}
The Riemannian gradient of $g(\Xc)$ can be computed from its Euclidean gradient following~\cite[Sec.~3.6]{AbsMahSep2008optManifolds}, and the details are given in Appendix~\ref{app:gradient}. We resort to the Manopt toolbox~\cite{manopt} to solve the optimization by conjugate gradient descent on the manifold. Note that the optimization space is not an Euclidean space and the objective function {$g(\Xc)$} is in general nonconvex, thus most descent algorithms only guarantee to return an (approximate) critical point. In order to ensure that this point is a local minimum and not a saddle point, the search direction needs to be carefully constructed. Several rules to construct the new search direction based on a linear combination of the previous search direction and the new (preconditioned) gradient are provided for the Euclidean space in~\cite{Hager2006survey}. The Manopt toolbox adapts these rules to the Riemannian space. If no descent direction is found, one can restart, i.e., switch to the negative gradient. This is equivalent to resetting the direction to a steepest descent step, which discards the past information. The Manopt toolbox implements Powell's restart strategy~\cite{Powell1977restart}. We optimize the joint constellations at $P = 30$~dB, although the constellations are then benchmarked at other SNR values.

\subsubsection{Initialization}
Note that the objective function $g(\Xc)$ in~\eqref{eq:tmp1102} is in general nonconvex and can have multiple local optima. With different initializations, the optimization converges to different local optima. We observe from numerical experiments that different local optima obtained from different initial points can yield significantly different values of $g(\Xc)$. Furthermore, the best initial point for direct optimization of the metric might not be the best initial point for alternating optimization. In our simulation, we try multiple initial points that can be easily generated, namely, the precoding-based constellation, partitioning-based constellation, the pilot-based constellation, and random constellations sampled from the manifold. We eventually choose the option that results in the highest metric value.

\subsubsection{Complexity Analysis}
In each gradient descent step, the objective function $g(\Xc)$ and its Riemannian gradient $\nabla_{\rm R} g(\Xc)$ (given in Appendix~\ref{app:gradient}) need to be computed. In Table~\ref{tab:complexity}, we give the complexity order of these operations for different criteria, where we assume that $|\Xc_k| = \Theta(2^B)$, $\forall k \in [K]$. Note that the complexity orders of computing $g(\Xc)$ and $\nabla_{\rm R} g(\Xc)$ are $O(2^{2KB} (T^3 + T^2 M_{\rm tot})$ and $O(K 2^{(2K+1)B} (T^3 + T^2 M_{\rm tot}))$, respectively, for all considered criteria. However, the absolute number of operations of complexity order $O(T^3)$ or $O(T^2 M_{\rm tot})$ varies for these metrics. Therefore, to facilitate the comparison, we include a scaling factor indicating the number of these dominating operations in Table~\ref{tab:complexity}. The Max-$J_{1/2,\min}$ criterion has the highest complexity, which shall be justified by its performance advantage in the next subsections. The Min-$m_1$ criterion has the lowest complexity. The scaling factors for the Max-$d_{\min}$ criterion is lower than that for Max-$e_{\min}$, although we shall see that they lead to similar performance. Alternating optimization allows to reduce the complexity order of computing $\nabla_{\rm R} g(\Xc)$ by a factor of $K$.
\begin{table*}[t!]
\begin{center}
\caption{The complexity order of computing the objective function $g(\Xc)$ and its Riemannian gradient $\nabla_{\rm R} g(\Xc)$ for different criteria}
			\def\arraystretch{1.5}
\small
\begin{tabular}{ c c c }
\hline 
Criterion & Complexity of computing $g(\Xc)$ & Complexity of computing $\nabla_{\rm R} g(\Xc)$ \\ \hline
Max-$J_{1/2,\min}$ & $\Theta(2^{2KB} (5T^3 + 2T^2 M_{\rm tot}))$ & $\Theta(K 2^{(2K+1)B} (9T^3 + 2T^2 M_{\rm tot}))$\\ %\hline
Max-$d_{\min}$ & $\Theta(2^{2KB} (2T^3 + 2T^2 M_{\rm tot}))$ & $\Theta(K 2^{(2K+1)B} (2.5T^3 + 2T^2 M_{\rm tot}))$ \\ %\hline
Max-$e_{\min}$ & $\Theta(2^{2KB} (4T^3 + 2T^2 M_{\rm tot}))$ & $\Theta(K 2^{(2K+1)B} (4.5T^3 + 2T^2 M_{\rm tot}))$ \\ %\hline
Min-$m_1$ & $\Theta(2^{2KB} (T^3 + 2T^2 M_{\rm tot}))$ & $\Theta(K 2^{(2K+1)B} (T^3 + 2T^2 M_{\rm tot}))$ \\ %\hline
Min-$m_2$ & $\Theta(2^{2KB} (2T^3 + 2T^2 M_{\rm tot}))$ & $\Theta(K 2^{(2K+1)B} (2T^3 + 2T^2 M_{\rm tot}))$ \\ \hline
\end{tabular}
\label{tab:complexity}
\end{center}
\end{table*}

Hereafter, in all figures, the legends representing our proposed schemes are in bold face. 

\subsection{The Single-User Case}
We first consider the single-user case, i.e., $K=1$, with coherence interval $T = 4$, $B \in \{5, 6\}$ bits/symbol, $M = 2$ transmit antennas, and $N = 2$ receive antennas. In Fig.~\ref{fig:SER_P2P}, we show the SER as a function of the SNR $P$ for the constellations obtained by optimizing different metrics. We see that the constellations optimized with the proposed criteria Max-$J_{1/2,\min}$ is on par with that optimized with Min-$m_2$, and outperforms the constellations optimized with the other metrics. The constellation obtained with Max-$d_{\min}$ is on par with that obtained with Max-$e_{\min}$, and slightly better than that with Min-$m_1$ in the high-SNR regime. This shows that for the single-user case where the truncated unitary structure of the symbols is guaranteed, our proposed metrics perform as  well as state-of-the-art metrics. On the other hand, in the multi-user case where the symbols are not necessarily truncated unitary matrices, our metrics have advantages over the existing ones, as we shall show next.
\begin{figure}[t!]
%	\vspace{-.5cm}
\centering
\subfigure[{$B = 5$ bits/symbol}]{
\begin{tikzpicture}[scale=.85,style={mark size=3pt,line width=3pt}]
\begin{axis}[%
width=3in,
height=2.2in,
at={(0.758in,0.481in)},
scale only axis,
xmin=10, xmax=20,
xtick={10,12,14,16,18,20},
xlabel style={font=\color{white!15!black}},
xlabel={SNR (dB)},
ymode=log,
ymin=1e-5, ymax=1e-1,
yminorticks=true,
ylabel style={font=\color{white!15!black}},
ylabel={Symbol Error Rate},
axis background/.style={fill=white},
xmajorgrids,
ymajorgrids,
yminorgrids,
legend style={at={(0.02,0.02)}, anchor=south west, legend cell align=left, align=left, fill=white, fill opacity=0.3, text opacity=1, draw=none, %draw=white!15!black
, nodes={scale=.85}}
]

\addplot [color=red, line width=1pt, mark=diamond, mark size=4pt, mark options={solid, red}]
table[row sep=crcr]{%
6.000000000000000   0.209700000000000  \\
8.000000000000000   0.093580000000000 \\
10.000000000000000   0.036000000000000 \\
12.000000000000000   0.011100000000000 \\
14.000000000000000   0.002690000000000 \\
16.000000000000000   0.000643000000000 \\
18.000000000000000   0.000130000000000 \\
20.000000000000000   0.000031000000000 \\
22.000000000000000                   0 \\
};
\addlegendentry{ \textbf{Max-$J_{1/2,\min}$}}

\addplot [color=blue, dotted, line width=1pt, mark=+, mark size=3pt, mark options={solid, blue}]
table[row sep=crcr]{%
6.000000000000000   0.216560000000000 \\
8.000000000000000   0.102860000000000 \\
10.000000000000000   0.041270000000000 \\
12.000000000000000   0.014650000000000 \\
14.000000000000000   0.004360000000000 \\
16.000000000000000   0.001245000000000 \\
18.000000000000000   0.000322000000000 \\
20.000000000000000   0.000074000000000 \\
22.000000000000000                   0 \\
};
\addlegendentry{ \textbf{Max-$d_{\min}$}}

\addplot [color=pimentgreen, dotted, line width=1pt, mark=x, mark size=3pt, mark options={solid, pimentgreen}]
table[row sep=crcr]{%
6.000000000000000   0.216600000000000 \\
8.000000000000000   0.100770000000000 \\
10.000000000000000   0.040770000000000 \\
12.000000000000000   0.014710000000000 \\
14.000000000000000   0.004520000000000 \\
16.000000000000000   0.001169000000000 \\
18.000000000000000   0.000307000000000 \\
20.000000000000000   0.000069000000000 \\
22.000000000000000                   0 \\
};
\addlegendentry{ {Max-$e_{\min}$}}

\addplot [color=pimentgreen, dashdotted, line width=1pt, mark=o, mark size=3pt, mark options={solid, pimentgreen}]
table[row sep=crcr]{%
6.000000000000000   0.222960000000000 \\
8.000000000000000   0.106400000000000 \\
10.000000000000000   0.043050000000000 \\
12.000000000000000   0.015920000000000 \\
14.000000000000000   0.004980000000000 \\
16.000000000000000   0.001343000000000 \\
18.000000000000000   0.000337000000000 \\
20.000000000000000   0.000103000000000 \\
22.000000000000000                   0 \\
};
\addlegendentry{ Min-$m_1$}

\addplot [color=pimentgreen, dashdotted, line width=1pt, mark=square, mark size=3pt, mark options={solid, pimentgreen}]
table[row sep=crcr]{%
6.000000000000000   0.209400000000000 \\
8.000000000000000   0.095790000000000 \\
10.000000000000000   0.035420000000000 \\
12.000000000000000   0.010450000000000 \\
14.000000000000000   0.002720000000000 \\
16.000000000000000   0.000603000000000 \\
18.000000000000000   0.000132000000000 \\
20.000000000000000   0.000028000000000 \\
22.000000000000000                   0 \\
};
\addlegendentry{ Min-$m_2$}

\end{axis}
\end{tikzpicture}%
}
\subfigure[{$B = 6$ bits/symbol}]{
\begin{tikzpicture}[scale=.85,style={mark size=3pt,line width=3pt}]
\begin{axis}[%
width=3in,
height=2.2in,
at={(0.758in,0.481in)},
scale only axis,
xmin=10, xmax=22,
xtick={10,12,14,16,18,20,22},
xlabel style={font=\color{white!15!black}},
xlabel={SNR (dB)},
ymode=log,
ymin=1e-5, ymax=1e-1,
yminorticks=true,
ylabel style={font=\color{white!15!black}},
ylabel={Symbol Error Rate},
axis background/.style={fill=white},
xmajorgrids,
ymajorgrids,
yminorgrids,
legend style={at={(0.02,0.02)}, anchor=south west, legend cell align=left, align=left, fill=white, fill opacity=0.3, text opacity=1, draw=none, %draw=white!15!black
, nodes={scale=.85}}
]

\addplot [color=red, line width=1pt, mark=diamond, mark size=4pt, mark options={solid, red}]
table[row sep=crcr]{%
6.000000000000000   0.320400000000000 \\
8.000000000000000   0.164800000000000 \\
10.000000000000000   0.073900000000000 \\
12.000000000000000   0.027300000000000 \\
14.000000000000000   7.930e-03 \\
16.000000000000000   1.930e-03 \\
18.000000000000000   3.700e-04 \\
20.000000000000000   0.000087000000000 \\
22.000000000000000   0.000012000000000 \\
};
\addlegendentry{ \textbf{Max-$J_{1/2,\min}$}}

\addplot [color=blue, dotted, line width=1pt, mark=+, mark size=3pt, mark options={solid, blue}]
table[row sep=crcr]{%
6.000000000000000   0.320800000000000 \\
8.000000000000000   0.176400000000000 \\
10.000000000000000   0.081400000000000 \\
12.000000000000000   0.035400000000000 \\
14.000000000000000   1.207e-02 \\
16.000000000000000   3.930e-03 \\
18.000000000000000   1.100e-03 \\
20.000000000000000   0.000341000000000 \\
22.000000000000000   0.000064000000000 \\
};
\addlegendentry{ \textbf{Max-$d_{\min}$}}

\addplot [color=pimentgreen, dotted, line width=1pt, mark=x, mark size=3pt, mark options={solid, pimentgreen}]
table[row sep=crcr]{%
6.000000000000000   0.328300000000000 \\
8.000000000000000   0.169800000000000 \\
10.000000000000000   0.081000000000000 \\
12.000000000000000   0.033000000000000 \\
14.000000000000000   1.163e-02 \\
16.000000000000000   3.490e-03 \\
18.000000000000000   1.190e-03 \\
20.000000000000000   0.000309000000000 \\
22.000000000000000   0.000055000000000 \\
};
\addlegendentry{ {Max-$e_{\min}$}}

\addplot [color=pimentgreen, dashdotted, line width=1pt, mark=o, mark size=3pt, mark options={solid, pimentgreen}]
table[row sep=crcr]{%
6.000000000000000   0.319500000000000 \\
8.000000000000000   0.175100000000000 \\
10.000000000000000   0.081700000000000 \\
12.000000000000000   0.032400000000000 \\
14.000000000000000   1.135e-02 \\
16.000000000000000   3.550e-03 \\
18.000000000000000   9.800e-04 \\
20.000000000000000   0.000310000000000 \\
22.000000000000000   0.000083000000000 \\
};
\addlegendentry{ Min-$m_1$}

\addplot [color=pimentgreen, dashdotted, line width=1pt, mark=square, mark size=3pt, mark options={solid, pimentgreen}]
table[row sep=crcr]{%
6.000000000000000   0.315400000000000 \\
8.000000000000000   0.161800000000000 \\
10.000000000000000   0.068400000000000 \\
12.000000000000000   0.026600000000000 \\
14.000000000000000   7.920e-03 \\
16.000000000000000   1.680e-03 \\
18.000000000000000   4.000e-04 \\
20.000000000000000   0.000072000000000 \\
22.000000000000000   0.000009800000000 \\
};
\addlegendentry{ Min-$m_2$}

\end{axis}
\end{tikzpicture}%
}
\caption{{The SER of the constellations optimized with different criteria for $K=1$ user, coherent interval $T = 4$, $B \in \{5, 6\}$ bits/symbol, $M = 2$ transmit antennas, and $N = 2$ receive antennas.}}
\label{fig:SER_P2P}
%\vspace{-.4cm}
\end{figure}
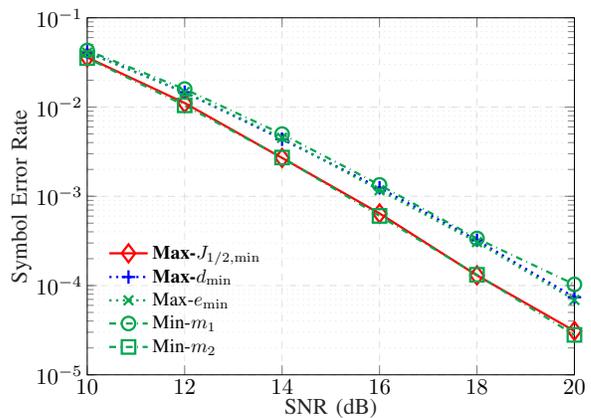
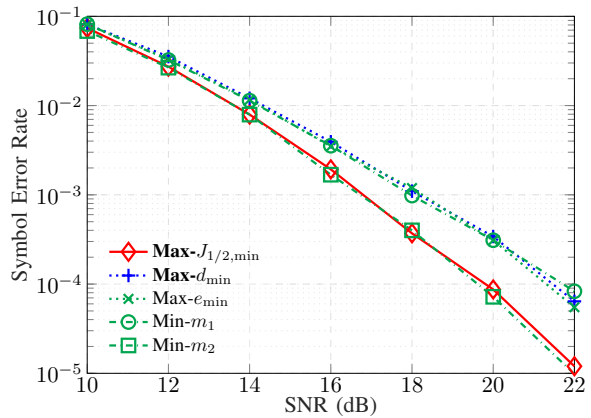

\subsection{{The Multi-User Case With Symmetrical Rate and Equal Power}}
%\vspace{-.2cm}
{In the multi-user case,} we first consider the \sheng{symmetrical rate} setting
$R_1 = \dots = R_K = B/T$ \sheng{with equal power $P_1=\dots = P_K =P$ for all users}. %(Thus $P$ is identified with the SNR.) 

\subsubsection{Two-User Case}
%For the two-user ($K=2$) case,
%in Fig.~\ref{fig:joint_SER_T5B4K2M2}, we plot the joint SER \eqref{eq:joint_SER} of the considered schemes for $T = 5$, $B = 5$,
%and $N = 4$. We observe that the constellations optimized with our
%metrics {$J_{1/2,\min}(\Xc)$,} $e_{\min}(\Xc)$
%\eqref{eq:criterion_minMean} and $d_{\min}(\Xc)$
%\eqref{eq:criterion_minTrace} achieve similar performance and are the
%best among the schemes with \sheng{the same rate pair}. The performance of the alternatively optimized constellation is slightly inferior and better than the pilot-based scheme. The partitioning design (with random partition) and the precoding design have similar performance. The constellations optimized with the Min-$m_1$ and Min-$m_2$ criteria, especially the latter, perform worse than that with our criteria.

For the two-user ($K=2$) case,
in Fig.~\ref{fig:joint_SER_T5B4K2M2}, we plot the joint SER \eqref{eq:joint_SER} of the considered constellations for $T = 5$, $B = 4$,
{$M = 2$} and $N = 4$. {We observe that the constellation optimized with the
$J_{1/2,\min}(\Xc)$ metric achieves the best performance among the schemes pertaining to the same rate pair, while the constellation optimized with the $d_{\min}(\Xc)$ metric
\eqref{eq:criterion_minTrace} achieves similar performance as that with the  $e_{\min}(\Xc)$
\eqref{eq:criterion_minMean} metric and outperforms the other constellations for medium and large SNR. The performance of the alternatively optimized constellation with the $d_{\min}(\Xc)$ metric is only slightly inferior to the direct optimization, and better than the pilot-based scheme. The partitioning design (with random partition) and the precoding design respectively outperform the constellations optimized with the Min-$m_2$ and Min-$m_1$ criteria.}
\begin{figure}[t!]
%\vspace{-.3cm}
\centering
\input{fig/joint_SER_T5B4M2N4.tex}
%\vspace{-.5cm}
\caption{The joint SER of the proposed constellations \sheng{compared
to the baselines} for $T = 5$, $K = 2$, {$B = 4$, $M = 2$,} and $N = 4$.}
\label{fig:joint_SER_T5B4K2M2}
%\vspace{-.5cm}
\end{figure}
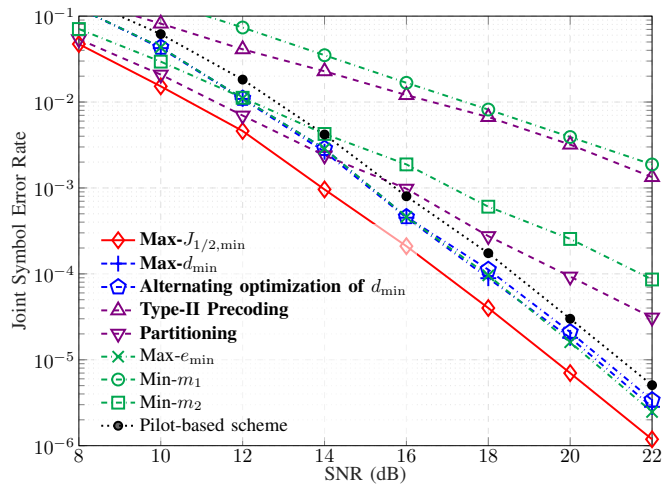

To assess the effectiveness of the design metrics, in Fig.~\ref{fig:metric_T5B4K2M2}, we show the values of our metrics {$b_{\min}(\Xc)$, $J_{1/2,\min}(\Xc)$,} $d_{\min}(\Xc)$  and the {baseline} metrics $e_{\min}(\Xc)$, $m_1(\Xc)$, and $m_2(\Xc)$ for the constellations considered in Fig.~\ref{fig:joint_SER_T5B4K2M2}. {As shown in Fig.~\ref{fig:metric_T5B4K2M2_R} and Fig.~\ref{fig:metric_T5B4K2M2_Js}, the relative order of the constellations in terms of the metrics $b_{\min}(\Xc)$ and $J_{1/2,\min}(\Xc)$ is identical to their relative order in terms of joint-ML SER at moderate/high SNR in Fig.~\ref{fig:joint_SER_T5B4K2M2}. Similar observation holds for the metrics $d_{\min}(\Xc)$ and $e_{\min}(\Xc)$ in Fig.~\ref{fig:metric_T5B4K2M2_e_d}.\footnote{{An exception is that the constellation obtained with Max-$J_{1/2,\min}$ does not have the highest value of $d_{\min}(\Xc)$ and $e_{\min}(\Xc)$, which is speculated to result from the suboptimality of the optimization solution.}} This confirms that our proposed metrics and the $e_{\min}(\Xc)$ metric are meaningful for constellation design and evaluation. We also see that $d_{\min}(\Xc)$ is very close to $e_{\min}(\Xc)$ for $\SNR \ge 20$~dB.} From Fig.~\ref{fig:metric_T5B4K2M2_m1} and Fig.~\ref{fig:metric_T5B4K2M2_m2}, we see that the relative order of the constellations in terms of the value of the baseline metrics $m_1(\Xc)$ in \eqref{eq:criterion_maxminChordal} and $m_2(\Xc)$ in \eqref{eq:criterion_minsumDet} is rather unrevealing about the SER performance in Fig.~\ref{fig:joint_SER_T5B4K2M2}. {For example, the constellations optimized with Max-$e_{\min}$, Max-$d_{\min}$, and alternating optimization perform well although they have high values of the $m_2(\Xc)$ metric.} %In particular, although the constellation optimized with the metric $m_1(\Xc)$ also achieves a low joint-ML SER (and a high value of our metrics), this is not true for other constellations. For example, the partitioning design has a low value of $m_1(\Xc)$ but a high SER; the pilot-based constellation has a much lower value of $m_2(\Xc)$ than the Max-$e_{\min}$ design, but has a higher SER.
%\begin{figure}[h!]
%	%\vspace{-.3cm}
%	\centering
%	\input{fig/metric_T5B5N4_journal.tex}
%	%\vspace{-.4cm}
%	\caption{The value of the metrics $\frac{1}{N}\min_{\Xm \ne {\Xm'} \in \Xc} \E\big[\rv{L}(\Xm\to{\Xm'})\big]$ (lines) and $d_{\min}(\Xc)$ (markers) for the considered schemes for $T = 5$ and $B = 5$.}
%	\label{fig:metric_T5B5N4}
%	%\vspace{-.3cm}
%\end{figure}
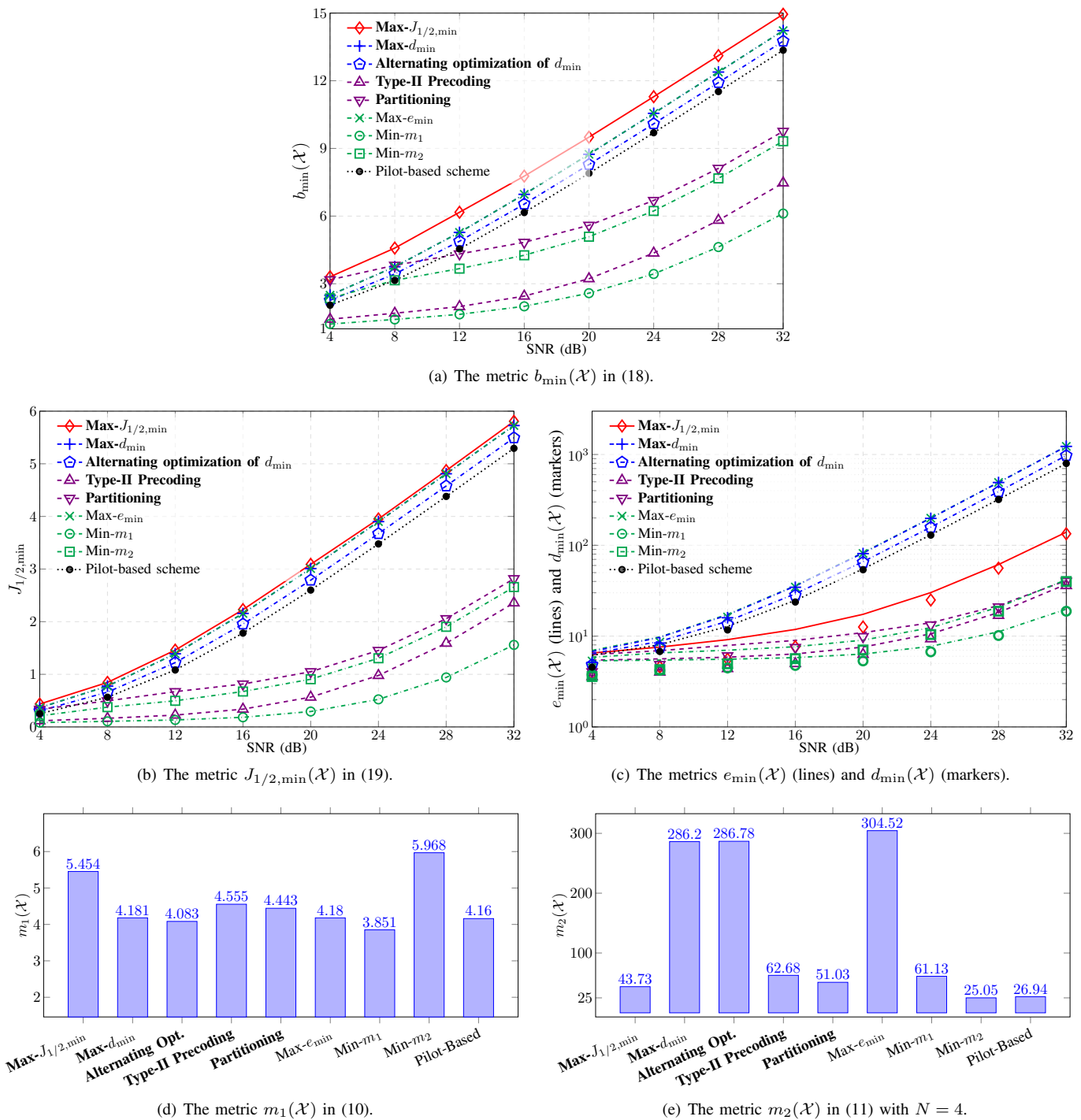
\begin{figure*}[t!]
%\vspace{-.4cm}
\centering
\subfigure[{The metric $b_{\min}(\Xc)$ in~\eqref{eq:criterion_Riemannian}.}]{\input{fig/metric_beta_T5B4K2M2.tex}\label{fig:metric_T5B4K2M2_R}} \\
%\hspace{-.3cm}
\subfigure[{The metric $J_{1/2,\min}(\Xc)$ in~\eqref{eq:criterion_Js}.}]{\input{fig/metric_Js_T5B4K2M2.tex}\label{fig:metric_T5B4K2M2_Js}}
\hspace{.2cm}
\subfigure[The metrics $e_{\min}(\Xc)$ (lines) and $d_{\min}(\Xc)$ (markers).]{\input{fig/metric_emin_dmin_T5B4K2M2.tex}\label{fig:metric_T5B4K2M2_e_d}}
%\hspace{-.3cm}
\\
%\hspace{-.2cm}
\subfigure[The metric $m_1(\Xc)$ in~\eqref{eq:criterion_maxminChordal}.]{\input{fig/metric_minChordal_T5B4K2M2.tex}\label{fig:metric_T5B4K2M2_m1}}
\hspace{.2cm}
\subfigure[The metric $m_2(\Xc)$ in~\eqref{eq:criterion_minsumDet} with {$N=4$.}]{\input{fig/metric_sumDet_T5B4K2M2.tex}\label{fig:metric_T5B4K2M2_m2}}
%\hspace{-.2cm}
%\vspace{-.3cm}
\caption{The value of the {design} metrics for the considered constellations for $T = 5$, {$K=2$,} $B = 4$, {and $M = 2$}.}
\label{fig:metric_T5B4K2M2}
%\vspace{-.4cm}
\end{figure*}

%In Fig.~\ref{fig:joint_SER_B8N4}, we consider larger constellations ($B = 8$) for which numerical optimization of the joint metrics 
%becomes cumbersome. However, the partitioning and precoding constructions, which are based on our metrics, achieve good performance and outperform the pilot-based constellation.
%\begin{figure}[h!]
%%	\vspace{-.5cm}
%	\centering
%	\hspace{-.2cm}
%	\subfigure[$T=5, B = 8, N = 4$.]{\input{fig/joint_SER_T5B8N4_journal.tex}}
%	\hspace{.2cm}
%	\subfigure[$T=6, B = 8, N = 4$.]{\input{fig/joint_SER_T6B8N4_journal.tex}}
%	\hspace{-.2cm}
%
%	%\vspace{-.3cm}
%	\caption{The joint SER of the partitioning and precoding designs compared with a pilot-based scheme for $T \in \{5,6\}$, $K=2$, $B = 8$, {$M = 1$,} and $N = 4$.}
%	\label{fig:joint_SER_B8N4}
%	%\vspace{-.4cm}
%\end{figure}

%\vspace{-.2cm}
\subsubsection{Three-User Case}
In the three-user ($K=3$) case, we consider {$T = 7$,} $B=3$, {$N=6$,} and plot the joint SER of various constellations in Fig.~\ref{fig:joint_SER_T7B3K3N6}. We observe again that maximizing {$J_{1/2,\min} (\Xc)$ results in the best performance, while maximizing $e_{\min} (\Xc)$ and $d_{\min} (\Xc)$ yields similar performance and outperforms the other constellations in the moderate/high SNR regime. The pilot-based scheme is outperformed by the constellation obtained from Min-$m_1$. The SER of the Min-$m_2$ constellation and the partitioning-based constellation are low in the low-SNR regime but then decreases more slowly with the SNR.} %The constellations obtained from Min-$m_1$ and Min-$m_2$ are outperformed by the pilot-based and precoding-based constellations. The SER of the Min-$m_2$ constellation and the partitioning-based constellation decreases slower with the SNR than the other schemes. 

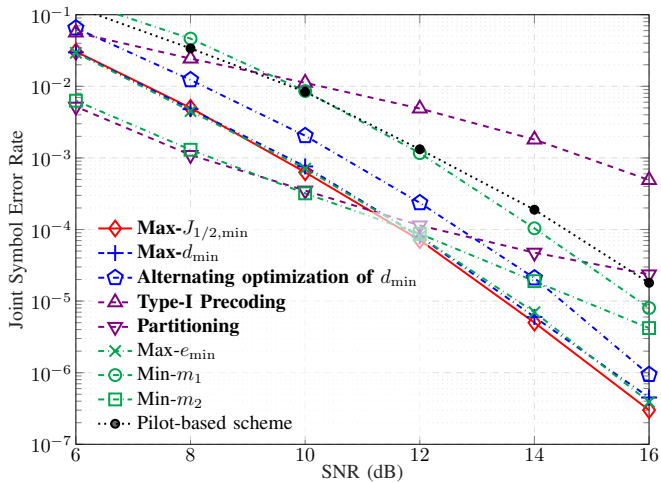
\begin{figure}[t!]
%\vspace{-.3cm}
\centering
\input{fig/joint_SER_T7B3K3M2N6.tex}
%\vspace{-.6cm}
\caption{{The joint SER of the proposed constellations compared
to the baselines for $T = 7$, $K=3$, $B = 3$, $M = 2$, and $N = 6$.}}
\label{fig:joint_SER_T7B3K3N6}
%\vspace{-.7cm}
\end{figure}

%\begin{figure}[h!]
%	%\vspace{-.3cm}
%	\centering
%	\input{fig/joint_SER_T4B3K3N4.tex}
%	%\vspace{-.6cm}
%	\caption{The joint SER of the proposed constellations \sheng{compared
%			to the baselines} for $T = 4$, $K=3$, $B = 3$, {$M = 1$,} and $N = 4$.}
%	\label{fig:joint_SER_T4B3K3N4}
%	%\vspace{-.7cm}
%\end{figure}

Fig.~\ref{fig:metric_T7B3K3M2} depicts the values of the proposed metrics {$b_{\min}(\Xc)$, $J_{1/2,\min}(\Xc)$,}  $d_{\min}(\Xc)$, and the {baseline} metrics {$e_{\min}(\Xc)$,} $m_1(\Xc)$ and $m_2(\Xc)$ for the considered constellations. 
{As for the two-user case, the relative order of the constellations in terms of the metrics $b_{\min}(\Xc)$, $J_{1/2,\min}(\Xc)$, $e_{\min}(\Xc)$ and $d_{\min}(\Xc)$ predicts well the relative order in terms of the joint SER in Fig.~\ref{fig:joint_SER_T7B3K3N6}. %For the Min-$m_2$ constellation and the partitioning-based constellation, the metrics $e_{\min}(\Xc)$ and $d_{\min}(\Xc)$ do not grow linearly in $P$, which is consistent with their SER performance. In fact, we can verify that these constellations violate \eqref{eq:tmp867} in Proposition~\ref{prop:lowerBound_minDist_Xsu} when $P$ is large: the minimum pairwise chordal distance $\delta_{\min}(\Xc_{\textnormal{SU}})$ of the concatenation $\Xc_{\textnormal{SU}}$ of their individual constellations are about $0.2958$ and $0.8264$ respectively, both smaller than $\nu(K =3,M=1) \approx 0.9543$. 
On the other hand,} from Fig.~\ref{fig:metric_T7B3K3M2_m1} and Fig.~\ref{fig:metric_T7B3K3M2_m2}, we further observe that the metrics $m_1(\Xc)$ and $m_2(\Xc)$ are not meaningful  for constellation evaluation.
\begin{figure*}[t!]
%\vspace{-.4cm}
\centering
\subfigure[{The metric $b_{\min}(\Xc)$ in~\eqref{eq:criterion_Riemannian}.}]{\input{fig/metric_beta_T7B3K3M2.tex}\label{fig:metric_T7B3K3M2_beta}}\\
%\hspace{-.3cm}
\subfigure[{The metric $J_{1/2,\min}(\Xc)$ in~\eqref{eq:criterion_Js}.}]{\input{fig/metric_Js_T7B3K3M2.tex}\label{fig:metric_T7B3K3M2_Js}}
\hspace{.2cm}
\subfigure[The metrics $e_{\min}(\Xc)$ (lines) and $d_{\min}(\Xc)$ (markers).]{\input{fig/metric_emin_dmin_T7B3K3M2.tex}\label{fig:metric_T7B3K3M2_e_d}}
%\hspace{-.3cm}
\\
%\hspace{-.2cm}
\subfigure[The metric $m_1(\Xc)$ in~\eqref{eq:criterion_maxminChordal}.]{\input{fig/metric_minChordal_T7B3K3M2.tex}\label{fig:metric_T7B3K3M2_m1}}
\hspace{.2cm}
\subfigure[The metric $m_2(\Xc)$ in~\eqref{eq:criterion_minsumDet} with {$N=6$.}]{\input{fig/metric_sumDet_T7B3K3M2.tex}\label{fig:metric_T7B3K3M2_m2}}
%\hspace{-.2cm}
%\vspace{-.3cm}
\caption{The value of the {design} metrics for the considered constellations for {$T = 7$, $K=3$, $B = 3$, and $M = 2$}.}
\label{fig:metric_T7B3K3M2}
%\vspace{-.4cm}
\end{figure*}
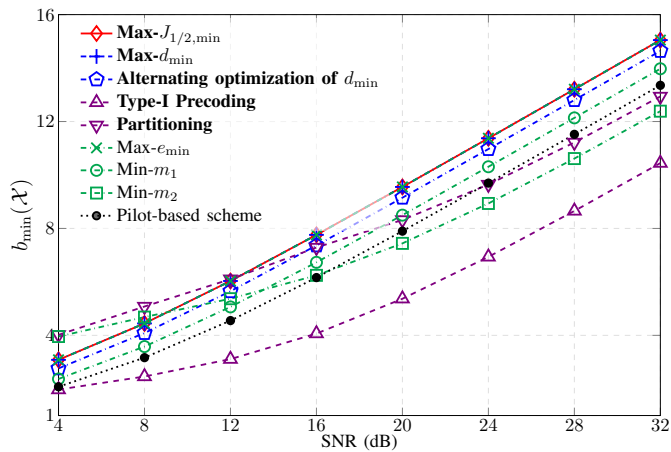
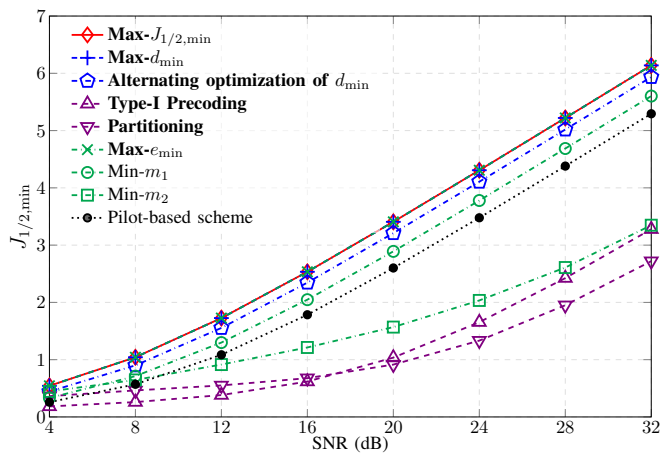
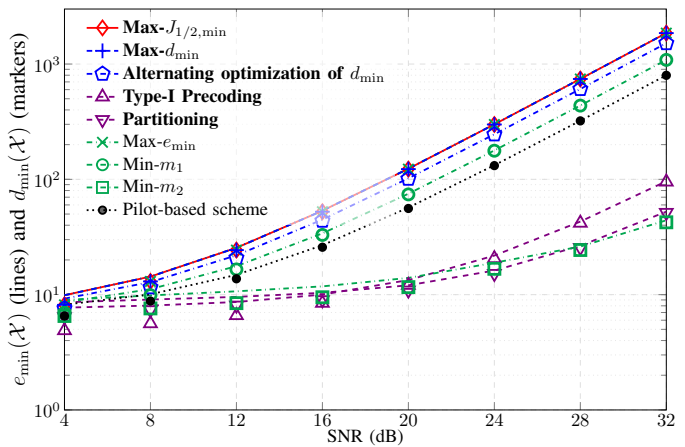
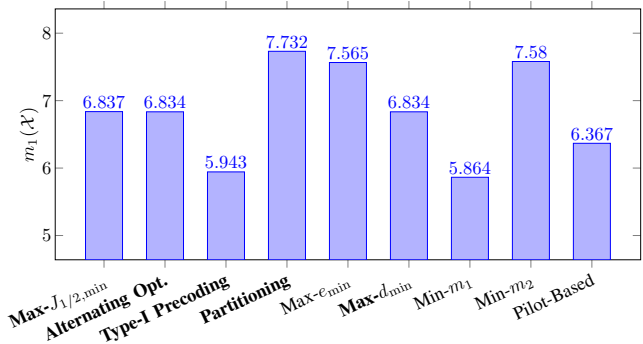
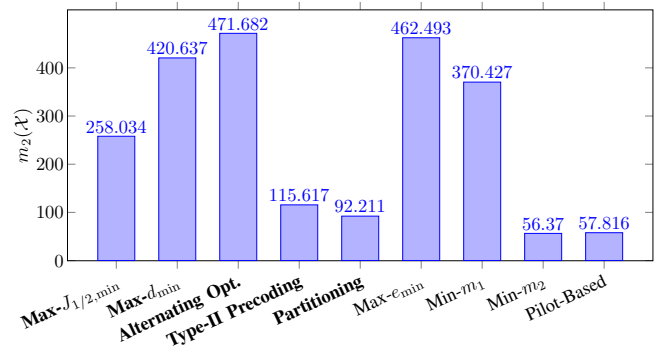

%\begin{figure}[h!]
%	%\vspace{-.6cm}
%	\centering
%	\subfigure[The metrics $e_{\min}(\Xc)$ (lines) and $d_{\min}(\Xc)$ (markers).]{\input{fig/metric_T4B3K3N4.tex}} 
%	\\
%	\hspace{-.3cm}
%	\subfigure[{The metric $J_{1/2,\min}(\Xc)$ in~\eqref{eq:criterion_Js}.}]{\input{fig/metric_T4B3K3N4_Js.tex}}
%	\hspace{.05cm}
%	\subfigure[{The Riemannian distance $\delta_{\rm R,\min}(\Xc)$ in~\eqref{eq:criterion_Riemannian}.}]{\input{fig/metric_T4B3K3N4_Riemannian.tex}}
%	\hspace{-.3cm}
%	\\
%	\hspace{-.2cm}
%	\subfigure[The metric $m_1(\Xc)$ in \eqref{eq:criterion_maxminChordal}.]{\input{fig/metric_minChordal_T4B3K3N4.tex}}
%	\hspace{.1cm}
%	\subfigure[The metric $m_2(\Xc)$ in \eqref{eq:criterion_minsumDet} with {$N=6$.}]{\input{fig/metric_sumDet_T4B3K3N4.tex}}
%	\hspace{-.2cm}
%	%\vspace{-.4cm}
%	\caption{The value of the {design} metrics  for the considered constellations for $T = 4$, $K = 3$, $B = 3$, {and $M = 1$}.}
%	\label{fig:metric_minChordal_sumDet_T4B3K3N4}
%	%\vspace{-.9cm}
%\end{figure}

%\vspace{-.3cm}
\subsection{The Multi-User Case With Asymmetrical Rate and Power Optimization} 
%\vspace{-.2cm}
We now consider the asymmetrical rate case and focus on the two-user SIMO ($M=1$) case. We set $T = 4$, $B_1 = 6$, and $B_2 = 2$ (as in Fig.~\ref{fig:metrics_vs_P2/P}). In Fig.~\ref{fig:SER_T4_1B6_2B2_N4}, we plot the joint SER of the constellations generated by Max-$d_{\min}$, precoding, or partitioning and compare with a pilot-based constellation with the same transmission rate for each user. Furthermore, we consider equal and full transmit power $P_1 = P_2 = P$, or optimized power as in Section~\ref{sec:power_opt}. The constellations obtained by Max-$d_{\min}$ significantly {outperform} other schemes. For this constellation, the optimal power coincides with full power $P_1 = P_2 = P$ for all $P > 4$~dB. For the precoding and partitioning designs, the optimal power allocation is to let user 1 (which has higher transmission rate) transmit at full power $P_1 = P$ and user 2 at lower power $P_2 = {{\theta^*}} P$ with ${{\theta^*}}$ obtained from optimizing $d_{\min}(\Xc)$ as in {Proposition~\ref{prop:theta}.} The SER with optimized power is only slightly lower than the SER with full power.
This is because the values of the metrics with optimized power are not significantly higher than that with full power, as seen in Fig.~\ref{fig:metrics_vs_P2/P}. However, using optimized power helps reduce the transmit power of user $2$, {thus save energy for this user.} The lower ${{\theta}^*}$ is, the further the power of user $2$ is saved with respect to transmitting at full power.
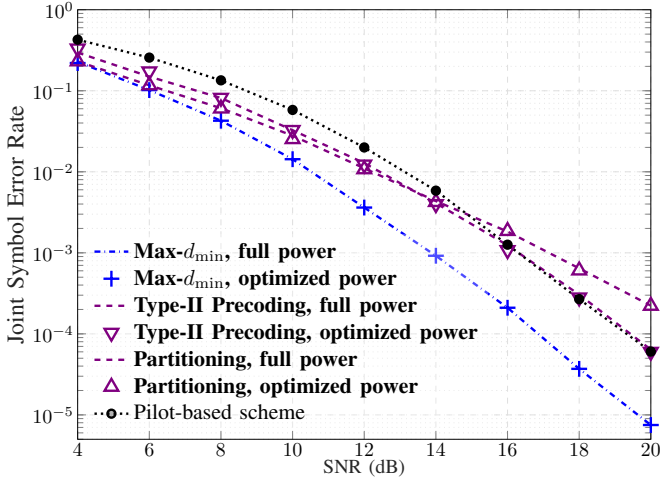
\begin{figure}[t!]
%\vspace{-.2cm}
\centering
\input{fig/T4_1B6_2B2_P30dB_SER_fullvsOptPower.tex}
%\vspace{-.6cm}
\caption{The joint SER of the proposed constellations with full transmit power $P_1 = P_2= P$ or optimized transmit power as in Section~\ref{sec:power_opt}, compared to a pilot-based constellation for $T=4$, $K=2$, $B_1 = 6$, $B_2 = 2$, {$M = 1$,} and $N=4$.}
\label{fig:SER_T4_1B6_2B2_N4}
%\vspace{-.3cm}
\end{figure}
%We found that the optimal power for the precoding and partitioning design is to let user 1 (which has higher transmission rate) transmits at full power $P_1 = P$ and user 2 at lower power $P_2 = \tilde{\theta} P$ with $\tilde{\theta}$ obtained from \eqref{eq:opt_P2}.
%Although the SER performance gain of using optimized power is marginal, this helps reducing the transmit power of user $2$. 

In Fig.~\ref{fig:powerSaving}, we plot the optimized power fraction ${{\theta}^*}$ for user 2  obtained using~{Proposition~\ref{prop:theta}.} %The power saving in dB is evaluated as $-10\log_{10} \theta_{\rm opt}$. 
For the precoding design, as the power constraint $P$ grows, $\tilde{{\theta}}^*$ increases, i.e., user $2$ should use more power. Conversely, for the partitioning design, user $2$ should use less power as $P$ grows. We note that this behavior might not hold for all constellations of the kind.
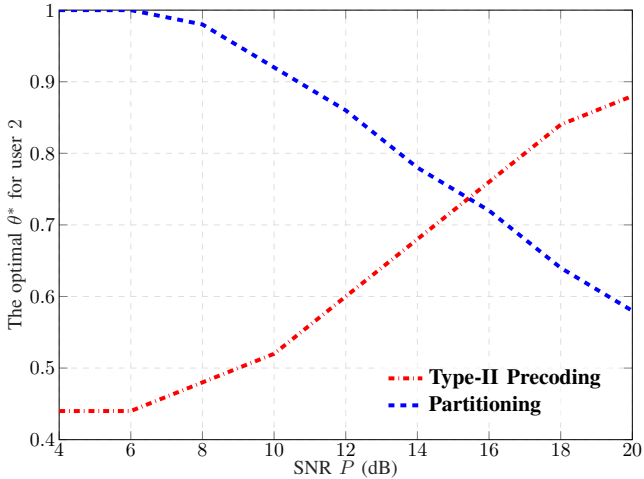
\begin{figure}[t!]
%\vspace{-.2cm}
\centering
\input{fig/theta_opt.tex}%T4_1B6_2B2_P30dB_powerSavingUser2.tex}
%\vspace{-.6cm}
\caption{The optimal power fraction $\theta^*$ for user $2$ for the precoding and partitioning designs with $T=4$, $B_1 = 6$, $B_2 = 2$, {and $M = 1$}.}
\label{fig:powerSaving}
%\vspace{-.6cm}
\end{figure}

%{
%\subsection{The MIMO Case}
%Finally, in this subsection we consider the two-user MIMO case with $M_1 = M_2 = M = 2$ antennas. For $M > 1$, one needs to modify the optimization techniques presented in Section~\ref{sec:optimization}, which is nontrivial. However, the precoding-based and partitioning-based constellation can be still efficiently generated. In Fig.~\ref{fig:joint_SER_MIMO}, we compare the joint SER of the precoding and partitioning constructions with the pilot-based scheme for $T = 6$, $B \in \{4,8\}$, and $N = 4$. For $B= 4$, it can be clearly seen that our constructions significantly outperform the pilot-based scheme. The two types of precoder lead to similar performance. For $B = 8$, the partitioning-based constellation outperforms the pilot-based constellation. In fact, with a higher number of transmit antennas, the pilot-based scheme dedicates a larger proportion of a coherence block to transmit pilot symbols. In this example, the pilot-based scheme spends $4$ channel uses for pilot transmission, equivalent to $2/3$  of a coherence block. 
%}
%\begin{figure}[h!]
%	%	\vspace{-.5cm}
%	\centering
%	\hspace{-.2cm}
%	\subfigure[{$B = 4$.}]{
%		\input{fig/joint_SER_T6B4K2M2N4.tex}
%	}
%	\hspace{.2cm}
%	\subfigure[{$B = 8$.}]{\input{fig/joint_SER_T6B8K2M2N4.tex}}
%	\hspace{-.2cm}
%	\caption{{The joint SER of the partitioning and precoding designs compared with the pilot-based scheme for $T = 6$, $K=2$, $B \in \{4, 8\}$, $M = 2$, $N = 4$.}}
%	\label{fig:joint_SER_MIMO}
%	%\vspace{-.4cm}
%\end{figure}

%\vspace{-.3cm}
%-----------------------------------------------
\section{Conclusion} \label{sec:conclusion}
%\vspace{-.2cm}
In this work, we studied the joint constellation design for
noncoherent MIMO MAC in Rayleigh block fading. By analyzing the joint detection error, we have derived closed-form metrics which
are effective for designing joint constellations that achieve a low error rate. Our metrics are applicable to both the single-user and multi-user scenarios, but are especially suited for the latter case. Specifically, the metric $b_{\min}(\Xc)$ provides tight bounds on the worst-case pairwise error probability, and was shown through numerical experiments to be a good indicator for the joint detection error of different constellations. Therefore, it can be used as a tool to evaluate the error performance of a given joint constellation. Joint constellations that achieve high values of $b_{\min}(\Xc)$ and low error probability can be obtained by maximizing the Chernoff-bound-based metric $J_{s,\min}(\Xc)$. Our $b_{\min}(\Xc)$ metric leads to a geometric interpretation: a joint constellation that achieves low detection error must have good Riemannian distance property in the manifold of Hermitian positive definite matrices. Our investigation of several baseline metrics adapted from existing criteria for the MIMO point-to-point  channel shows that the KL-distance-based metric $e_{\min}(\Xc)$ exhibits good performance, which can be closely approached with our simplified metric $d_{\min}(\Xc)$. To reduce the complexity of the maximization of these metrics, we proposed and demonstrated the effectiveness of two practical approaches, namely, alternating optimization and reduction of the solution space to the class of unitary space-time modulation. Inspired by our metrics, we  proposed a simple constellation construction consisting in partitioning a single-user constellation. We also provided another simple construction based on precoding individual constellations of lower dimension, which is a generalization of our previous design for the SIMO MAC. Furthermore, we investigated the optimization of the per-user symbol power.

In this work, we have focused on the optimality with respect to the joint ML detector. Note that it is common in the literature to use insights from the joint ML detector performance as guidelines to design constellation and detection for the single-user case. However, this detector has high complexity in general. It would be interesting, as in the single-user case~\cite{Kammoun2007noncoherentCodes,Hoang_cubesplit_journal}, to construct joint constellations that allow for effective low-complexity detection. This is normally done by imposing a particular structure on the constellation. (For example, we exploited the geometric structure of the precoding-based constellation to design efficient multi-user detection in the SIMO case in \cite{HoangAsilomar2018multipleAccess}.) With the proposed metrics, this problem can be solved by introducing additional constraints on the constellation.
%A next step is to
%investigate the asymmetrical rate case in which power optimization also
%plays a key role. 
%We proposed novel metrics for joint constellation design for the two-user noncoherent MIMO MAC in block fading, and generate constellations based on these metrics. Numerical results verify that the proposed metrics are meaningful and the resulting constellations achieve good performance.
%\vspace{-.3cm}

%\end{document}
%\newpage
\appendices

\section{{Disscussion on the Extension to Correlated Fading}}
\label{app:correlated_fading}
{We discuss the possible extension to correlated fading in the following. At the users' side, spatial correlation between the antennas of different users is not likely since the users are not colocated. For the case where there is correlation between the antennas of the same user, the optimized joint constellation can be obtained from the optimized joint constellation for uncorrelated fading with a modified power constraint. This is stated in the following proposition.}
{
\begin{proposition} \label{rem:correlation}
Consider the case where there is correlation between the antennas of the same user, namely, the rows of $\rvMat{H}$ are independent and follow $\Cc\Nc(\mathbf{0},\Psim)$ with $\Psim \defeq \Bigg[\begin{smallmatrix}
\Psim_1 &  & \mathbf{0} \\ & \ddots & \\ \mathbf{0}  &  & \Psim_K
\end{smallmatrix}\Bigg]$ where $\Psim_k$ is an $M_k \times M_k$ positive definite matrix. In this case, the solution to the ML error minimization \eqref{eq:criterion_MLerror} can be expressed as $\tilde{\Xc}_k = \{\Xm_k \Psim_k^{\frac12}: \Xm_k \in \Xc^*_k\}$ where $\{\Xc^*_k\}_{k=1}^K$ is the solution to \eqref{eq:criterion_MLerror} for uncorrelated fading where the power constraint is replaced with  
\begin{align} \label{eq:const_power_constraint_corr}
\frac{1}{|\Xc_k|} \sum_{\Xm_k \in \Xc_k} \|\Xm_k \Psim_k^{-\frac12} \|_{\rm F}^2 \le P T, \quad k \in [K].
\end{align} 
%{Correlation between the antennas of different users is not likely since the users are not fully colocated. Correlation between the receive antennas can be canceled at the receiver.}
\end{proposition}}
{\begin{proof}
In the mentioned correlated fading case, the channel output $\rvMat{Y}$ can be written as
$$\rvMat{Y} = \sum_{k=1}^{K}\rvMat{X}_k \rvMat{H}_k^\T + \rvMat{Z} = \sum_{k=1}^{K}\tilde{\rvMat{X}}_k \tilde{\rvMat{H}}_k^\T + \rvMat{Z} = \tilde{\rvMat{X}} \tilde{\rvMat{H}}^\T + \rvMat{Z}$$ 
where $\tilde{\rvMat{X}}_k \defeq {\rvMat{X}}_k \Psim_k^{\frac12}$, $\tilde{\rvMat{X}} = [\tilde{\rvMat{X}}_1 \ \tilde{\rvMat{X}}_2 \ \dots \ \tilde{\rvMat{X}}_K]$, $\tilde{\rvMat{H}}_k \defeq {\rvMat{H}}_k \Psim_k^{-\frac12}$, and $\tilde{\rvMat{H}} = [\tilde{\rvMat{H}}_1 \ \tilde{\rvMat{H}}_2 \ \dots \ \tilde{\rvMat{H}}_K]$. Note that $\tilde{\rvMat{H}}$ is a uncorrelated matrix with i.i.d $\Cc\Nc(0,1)$ entries. The likelihood function is given by $p_{\rvMat{Y} | \rvMat{X}}(\Ym | \Xm) = \frac{\exp(-\trace(\Ym^\H(\Id_T+\Xm \Psim \Xm^\H)^{-1}\Ym ))}{\pi^{NT}\det^N(\Id_T+\Xm\Psim\Xm^\H)} = \frac{\exp(-\trace(\Ym^\H(\Id_T+\tilde{\Xm} \tilde{\Xm}^\H)^{-1}\Ym ))}{\pi^{NT}\det^N(\Id_T+\tilde{\Xm}\tilde{\Xm}^\H)}$. Therefore, ML detection in the correlated channel from $\rvMat{X}$ to $\rvMat{Y}$ is equivalent to ML detection in the uncorrelated channel from $\tilde{\rvMat{X}}$ to $\rvMat{Y}$, where the power constraint becomes $	\frac{1}{|\Xc_k|} \sum_{\tilde{\Xm}_k \in \Xc_k} \|\tilde{\Xm}_k \Psim_k^{-\frac12} \|_{\rm F}^2 \le P T,  k \in [K].$ As a consequence, one can obtain the optimal constellation for the correlated fading case from that for the uncorrelated fading case with the constraint \eqref{eq:const_power_constraint_corr} on the constellation symbols.
\end{proof}}

{The optimization \eqref{eq:criterion_MLerror} with constraint \eqref{eq:const_power_constraint_corr} is a generalization of the problem considered in this paper, and is a subject for future works. In particular, if $\Psim_k = \psi_k\Id_{M_k}$, i.e., the users experience different path losses given by $\{\psi_k\}_{k=1}^K$, an outer power-loading algorithm could be used to manage the path loss such that the effective channel gain of different users are equal.
}

{Correlation at the receiver's side is likely if the receive antennas are placed close to each other. In this case, the constellation optimization is equivalent to the case with uncorrelated fading with colored noise, as stated in the following proposition.
\begin{proposition}
Consider the case where there is correlation between the receive antennas, namely, the columns of $\rvMat{H}$ are independent an follow $\Cc\Nc(\mathbf{0}, \Psim)$ where $\Psim$ is an $N \times N$  positive definite matrix. In this case, the solution to the ML error minimization~\eqref{eq:criterion_MLerror} is identical to that for the uncorrelated fading case with a colored noise matrix having independent rows following $\Cc\Nc(\mathbf{0},\Psim^{-1})$.
\end{proposition}
\begin{proof}
Upon receiving the signal $\rvMat{Y} = \rvMat{X}\rvMat{H}^\T + \rvMat{Z}$, the receiver can cancel the correlation by computing
\begin{align}
\tilde{\rvMat{Y}} = \rvMat{Y} \Psim^{-1/2} = \rvMat{X}\tilde{\rvMat{H}}^\T + \rvMat{Z} \Psim^{-1/2},
\end{align}
where the equivalent channel matrix $\tilde{\rvMat{H}} \defeq \Psim^{-1/2} \rvMat{H}$ has i.i.d. entries following $\Cc\Nc(0,1)$. The channel from $\rvMat{X}$ to $\tilde{\rvMat{Y}}$ has uncorrelated fading and colored noise matrix $\rvMat{Z} \Psim^{-1/2}$ with independent rows following $\Cc\Nc(\mathbf{0},\Psim^{-1})$. Since $\Psim$ is positive definite, the mapping $\rvMat{Y} \mapsto \tilde{\rvMat{Y}}$ is one-to-one. Therefore, ML detection of $\rvMat{X}$ from $\rvMat{Y}$ is equivalent to that from $\tilde{\rvMat{Y}}$. The corresponding constellation optimization to minimize the ML error are thus identical.
\end{proof}

Solving the joint constellation optimization for the MAC with colored noise is also a subject for future works. The single-user counterpart has been investigated in~\cite{BekoTSP2007noncoherentColoredNoise}.
}

\section{A closed-form expression of the PEP} \label{app:closedFormPEP}
%Following the footstep of~\cite[Proposition 1]{Brehler2001noncoherent}, the PEP $\mathbb{P}(\Xm\to{\Xm'})$ can be derived in closed form as follows.

%\vspace{-.2cm}
\begin{proposition} \label{prop:PEP_general}
Let $\{{\hat{\lambda}}_l\}_{l=1}^L$ be the distinct non-zero eigenvalues of 
%	\begin{align} \label{eq:Lambdam}
$\Lambdam \defeq ( \Id_T+ \Xm \Xm^\H) \big( \Id_T+ {\Xm'}{\Xm'}^\H \big)^{-1} - \Id_T$
%	\end{align}
with multiplicities $\{\mu_l\}_{l=1}^L$, and let $\{{\hat{\lambda}}_l\}_{l=1}^{L_p}$ be positive and $\{{\hat{\lambda}}_l\}_{l=L_p+1}^L$ negative. The PEP is given by
\begin{equation}
\P(\Xm\to{\Xm'}) = \begin{cases}
1+\sum_{k=1}^{L_p} \xi_k\Big(N \ln\frac{\det( \Id_T+ \Xm \Xm^\H
)}{\det( \Id_T+ {\Xm'}
{\Xm'}^\H )}\Big), \\
\quad  \text{if~} \det( \Id_T+ \Xm
\Xm^\H ) \ge \det( \Id_T+ {\Xm'}
{\Xm'}^\H ), \\
-\sum_{k=L_p+1}^{L} \xi_k\Big(N \ln\frac{\det( \Id_T+ \Xm \Xm^\H
)}{\det( \Id_T+ {\Xm'}
{\Xm'}^\H )}\Big), \\
\quad  \text{if~} \det( \Id_T+ \Xm
\Xm^\H ) \le \det( \Id_T+ {\Xm'}
{\Xm'}^\H ),
\end{cases}
\end{equation}
with $\xi_k(c) \defeq {\rm Res}\bigg(\frac{e^{sc}}{s\prod_{l=1}^L {\hat{\lambda}}_l^{\mu_l N}\big(s+1/{\hat{\lambda}}_l\big)^{\mu_l N}}, \frac{-1}{{\hat{\lambda}}_k}\bigg)$ where $${\rm Res}(f(s),x) \defeq \frac{1}{(m-1)!}\lim_{s\to x} \frac{\dif^{m-1}}{\dif s^{m-1}} \big[(s-x)^m f(s)\big]$$ is the residue of a function $f(s)$ in a pole $x$ of multiplicity $m$.
\end{proposition}
%\vspace{-.3cm}
\begin{proof}
The closed-form expression of the PEP follows readily from \cite[Proposition 1]{Brehler2001noncoherent} by noting that the matrix $\textbf{C}^{\rm NC}_{ij}$ therein is equal to $\Lambdam \otimes \Id_N$ in our setting, thus has the same nonzero eigenvalues as $\Lambdam$ with multiplicities $N$.
\end{proof}

{
\section{Proof that $\lim\limits_{N\to\infty}\P({\Xm}\to {\Xm'}) = 0$ for any pair of identifiable symbols $\Xm$ and $\Xm'$} \label{proof:KL}
We invoke Cantelli's inequality to get that\footnote{{Cantelli's inequality~\cite[Sec.~II.8]{Savage1961probability} states that $\P(\rv{x} - \mu \le \lambda) \le \frac{\sigma^2}{\sigma^2+\lambda^2}$ for a real-valued random variable $\rv{x}$ with mean $\mu$ and variance $\sigma^2$, and $\lambda<0$. Applying this with $\rv{x} = \LLR$ and $\lambda = -\meanLLR$, we obtain \eqref{eq:Cantelli}.}} 
\begin{align} 
\P({\Xm}\to {\Xm'}) &= \P[\rv{L}(\Xm\to{\Xm'}) \le 0] \\
&\le 
\frac{\varLLR}{\varLLR + \meanLLR^2} \qquad  \label{eq:Cantelli} \\
&= \left(1+ \frac{\meanLLR^2}{\varLLR}\right)^{-1}. % &\ge \P[{\left|\rv{L}(\Xm\to{\Xm'}) - \meanLLR\right| \ge \meanLLR}] \\
% &\ge \P[{\rv{L}(\Xm\to{\Xm'}) \le 0, \rv{L}(\Xm\to{\Xm'}) \le \meanLLR}] \\
%&\ge .
\end{align}
Therefore, it suffices to show that $\frac{\meanLLR^2}{\varLLR} \to \infty$ as $N\to\infty$.
%Let $\rvMat{Y}_0 \defeq {(\Id_T+ \rvMat{X}\rvMat{X}^\H)^{-1/2}}\rvMat{Y}$ be a ``whitened'' version of $\rvMat{Y}$, then {$\rvMat{Y}_0$ is a Gaussian matrix with $T$ independent rows following $
%\Cc\Nc(\mathbf{0},\Id_N)$. From \eqref{eq:LLR0}, the PLLR $\LLR$ can be expressed  as
%\begin{equation}
%\LLR %&= N\ln\frac{\det\left( \Id_T+ {\Xm'}^\H
%%	{\Xm'} \right)}{\det\left( \Id_T+ \Xm \Xm^\H
%%	\right)} - 
%%\trace\left( \rvMat{Y}_0^\H\rvMat{Y}_0 \right) \notag \\
%%&\quad
%%+ \trace\left( ( \Id_T+ \Xm \Xm^\H)^{1/2} \big( \Id_T+ {\Xm'}{\Xm'}^\H \big)^{-1} ( \Id_T+ \Xm \Xm^\H)^{1/2} \rvMat{Y}_0^\H\rvMat{Y}_0 \right) \\ 
%= -N\ln\det[\Id_T + \bar{\Lambdam}]  + \trace[\bar{\Lambdam} \rvMat{Y}_0\rvMat{Y}_0^\H]. 
%%= -N\ln\det(\Id_T+ \bar{\Lambdam}) + \sum_{i=1}^{T} \lambda_i \rv{g}_i,
%\end{equation}
%where $\bar{\Lambdam} \defeq ( \Id_T+ \Xm \Xm^\H)^\frac12 \big( \Id_T+ {\Xm'}{\Xm'}^\H \big)^{-1} ( \Id_T+ \Xm \Xm^\H)^{\frac12} - \Id_T$. Let $T$ nonnegative eigenvalues of $\bar{\Lambdam}$ be $\{\lambda_i\}_{i=1}^T$. We write $\bar{\Lambdam} = \bar{\Um} \diag[\lambda_1,\dots,\lambda_T] \bar{\Um}^\H$ where $\bar{\Um}$ is a unitary matrix. We further expand the PLLR as
We recall from~\eqref{eq:LLR_expression} that 
\begin{align}
\LLR %&= -N\ln\det[\Id_T + \bar{\Lambdam}]  + \trace[ \diag[\lambda_1,\dots,\lambda_T] \bar{\Um}^\H \rvMat{Y}_0\rvMat{Y}_0^\H \bar{\Um}] \\
&= -N \sum_{i=1}^{T}\ln \lambda_i + \sum_{i=1}^{T} (\lambda_i - 1) \rv{g}_i
\end{align}
where $\{\rv{g}_i\}_{i=1}^T$ are independent Gamma random variables with shape $N$ and scale 1. It follows that
\begin{align}
\meanLLR &= %-N\ln\det(\Id_T+ \bar{\Lambdam}) + \textstyle\sum_{i=1}^{T} \lambda_i = 
-N \sum_{i=1}^{T}\ln \lambda_i + N\sum_{i=1}^{T} (\lambda_i - 1) \\ 
&= N \sum_{i=1}^{T} (\lambda_i - 1 - \ln \lambda_i), \\
\varLLR &= N \textstyle\sum_{i=1}^{T} (\lambda_i-1)^2 %= N\trace\big(\bar{\Lambdam}^2\big).
\end{align}
For any joint constellation satisfying the identifiability condition in Proposition~\ref{prop:identifiability}, we have $\Xm\Xm^\H \ne \ \Xm'{\Xm'}^\H$, thus $\Gammam \ne \Id_T$  and thus $\{\lambda_i \colon \lambda_i \ne 1\} \ne \emptyset$. It follows that $\sum_{i=1}^{T} (\lambda_i - 1 - \ln \lambda_i)$ is positive definite since $\ln x < x-1, \forall x > 0, x \ne 1$. Also, $\textstyle\sum_{i=1}^{T} (\lambda_i-1)^2$ is positive definite. %Therefore, $\trace(\bar{\Lambdam}) - \ln\det(\Id_T+ \bar{\Lambdam})$ is positive definite. 
Therefore
\begin{equation}
\frac{\meanLLR^2}{\varLLR} = N \frac{\big(\sum_{i=1}^{T} (\lambda_i - 1 - \ln \lambda_i)\big)^2}{\sum_{i=1}^{T} (\lambda_i-1)^2} \to \infty,
\end{equation}
as $N\to \infty$. 
This completes the proof.%From this and \eqref{eq:Cantelli}, we conclude that $\lim\limits_{N\to\infty}\P({\Xm}\to {\Xm'}) = 0$ for any pair of distinct symbols $\Xm$ and $\Xm'$ of a joint constellation satisfying the identifiability condition.
\section{Proof of Lemma~\ref{lemma:PEP_expression}} \label{proof:PEP_expression}
From the second equality in~\eqref{eq:PEP}, it suffices to show that the PLLR can be written as $-N \sum_{i=1}^T \ln \lambda_i + \sum_{i=1}^{T} (\lambda_i - 1) \rv{g}_i$. Let $\rvMat{Y}_0 \defeq {(\Id_T+ \rvMat{X}\rvMat{X}^\H)^{-1/2}}\rvMat{Y}$ be a ``whitened'' version of $\rvMat{Y}$, then $\rvMat{Y}_0$ is a Gaussian matrix with $T$ independent rows following $
\Cc\Nc(\mathbf{0},\Id_N)$. From \eqref{eq:LLR0}, the PLLR $\LLR$ can be expressed  as
\begin{align}
&\LLR \notag \\ %&= N\ln\frac{\det\left( \Id_T+ {\Xm'}^\H
%	{\Xm'} \right)}{\det\left( \Id_T+ \Xm \Xm^\H
%	\right)} - 
%\trace\left( \rvMat{Y}_0^\H\rvMat{Y}_0 \right) \notag \\
%&\quad
%+ \trace\left( ( \Id_T+ \Xm \Xm^\H)^{1/2} \big( \Id_T+ {\Xm'}{\Xm'}^\H \big)^{-1} ( \Id_T+ \Xm \Xm^\H)^{1/2} \rvMat{Y}_0^\H\rvMat{Y}_0 \right) \\ 
&= -N\ln\det [\Gammam]  \notag \\
&\quad + \trace\Big((( \Id_T\!+\! \Xm \Xm^\H)^\frac12 \big( \Id_T\!+\! {\Xm'}{\Xm'}^\H \big)^{-1} ( \Id_T\!+\! \Xm \Xm^\H)^{\frac12} - \Id_T) \notag \\
&\qquad \qquad  \cdot \rvMat{Y}_0\rvMat{Y}_0^\H\Big). 
%= -N\ln\det(\Id_T+ \bar{\Lambdam}) + \sum_{i=1}^{T} \lambda_i \rv{g}_i,
\end{align}
%where $\bar{\Lambdam} \defeq ( \Id_T+ \Xm \Xm^\H)^\frac12 \big( \Id_T+ {\Xm'}{\Xm'}^\H \big)^{-1} ( \Id_T+ \Xm \Xm^\H)^{\frac12} - \Id_T$. Let $T$ nonnegative eigenvalues of $\bar{\Lambdam}$ be $\{\lambda_i\}_{i=1}^T$. 
Since $\Gammam$ and $( \Id_T+ \Xm \Xm^\H)^\frac12 \big( \Id_T+ {\Xm'}{\Xm'}^\H \big)^{-1} ( \Id_T+ \Xm \Xm^\H)^{\frac12}$ share the same eigenvalues $\{\lambda_i\}_{i=1}^T$, we can decompose 
\begin{multline}
	( \Id_T+ \Xm \Xm^\H)^\frac12 \big( \Id_T+ {\Xm'}{\Xm'}^\H \big)^{-1} ( \Id_T+ \Xm \Xm^\H)^{\frac12} - \Id_T \\ = \bar{\Um} \diag[\lambda_1 - 1, \lambda_2 - 1,\dots,\lambda_T - 1] \bar{\Um}^\H
\end{multline}
where $\bar{\Um}$ is a $T \times T$ unitary matrix. We further expand the PLLR as
\begin{align}
\LLR &= -N\ln\det [\Gammam]  \notag \\
&\quad + \trace[ \diag[\lambda_1-1,\dots,\lambda_T-1] \bar{\Um}^\H \rvMat{Y}_0\rvMat{Y}_0^\H \bar{\Um}] \\
&= -N\sum_{i=1}^{T}\ln \lambda_i + \sum_{i=1}^{T} (\lambda_i - 1) \rv{g}_i \label{eq:LLR_expression}
\end{align}
where $\rv{g}_i \defeq \|\bar{\rvVec{y}}_{0,i}\|^2$ with $\bar{\rvVec{y}}_{0,i}$ being the $i$-th row of $\bar{\Um}^\H \rvMat{Y}_0$. Note that since $\bar{\Um}$ is unitary and deterministic, $\bar{\Um}^\H \rvMat{Y}_0$ has the same distribution  as $\rvMat{Y}_0$, i.e., $\bar{\rvVec{y}}_{0,i}$ are independent and follow $\Cc\Nc(\mathbf{0},\Id_N)$. Therefore, $\{\rv{g}_i\}_{i=1}^T$ are independent Gamma random variables with shape $N$ and scale $1$. This completes the proof.}

{
\section{Proof of Proposition~\ref{prop:Chernoff}} \label{proof:PEP_exponent}
%\subsection{Proof of~\eqref{eq:bound_PEP_exponent}}
In this proof, for convenience, we denote $\Gm_\Am \defeq (\Id + \Am \Am^\H)^{-1}$. %Then we have that $p_{\rvMat{Y}\cond \rvMat{X}}(\Ym \cond )$
We need to show that $\mathbb{P}(\Xm\to{\Xm'}) = \P[\rv{L}(\Xm\to{\Xm'}) \le 0] \le \exp(-NJ_s(\Xm,\Xm'))$.
By applying the Chernoff bound~\cite[Th. 6.2.7]{DeGroot2012ProbStats}, we obtain for every $s > 0$ that
\begin{align}
&\mathbb{P}(\Xm\to{\Xm'}) \notag \\
&\le \E_{\rvMat{Y} \cond \Xm}\left[\exp\left(-s  \rv{L}(\Xm\to{\Xm'}) \right)\right] \\
&= \E_{\rvMat{Y} \cond \Xm}\left[ \left(\frac{p_{\rvMat{Y} | \rvMat{X}}(\rvMat{Y} | \Xm')}{p_{\rvMat{Y} | \rvMat{X}}(\rvMat{Y} | {\Xm})}\right)^s\right] \\
&= \int_{\CC^{T\times N}} [p_{\rvMat{Y} | \rvMat{X}}(\Ym | \Xm')]^s [p_{\rvMat{Y} | \rvMat{X}}(\Ym | \Xm)]^{1-s} \dif \Ym \\
&= \int_{\CC^{T\times N}} \left[\frac{\exp(-\trace(\Ym^\H\Gm_{\Xm'}\Ym ))}{\pi^{NT}\det^{-N}\Gm_{\Xm'}}\right]^s \notag \\
&\qquad \cdot  \left[\frac{\exp(-\trace(\Ym^\H\Gm_{\Xm}\Ym ))}{\pi^{NT}\det^{-N}\Gm_{\Xm}}\right]^{1-s} \dif \Ym \label{eq:ref2044}\\
&= \left[\frac{\det^s(\Gm_{\Xm'}) \det^{1-s}(\Gm_\Xm)}{\det[s\Gm_{\Xm'} + (1-s)\Gm_{\Xm}]} \right]^N \notag \\
&\quad \cdot \int_{\CC^{T\times N}} \frac{\exp(-\trace(\Ym^\H(s\Gm_{\Xm'} + (1-s)\Gm_\Xm)\Ym ))}{\pi^{NT}\det^{-N}(s\Gm_{\Xm'} + (1-s)\Gm_\Xm)} \dif \Ym \qquad \label{eq:tmp2045}
\end{align}
where \eqref{eq:ref2044} follows from \eqref{eq:likelihood}, and \eqref{eq:tmp2045} follows after some simple manipulations. Next, we restrict to $s \in [0,1]$, and thus $(s\Gm_{\Xm'} + (1-s)\Gm_\Xm)^{-1}$ is a covariance matrix. Therefore, the integral in \eqref{eq:tmp2045} is an integral of a Gaussian density over the whole support, and thus equals $1$. As a consequence, $\mathbb{P}(\Xm\to{\Xm'})$ is upper-bounded by the first term in \eqref{eq:tmp2045}, which equals $\exp(-N J_s(\Xm,\Xm'))$. 
\section{Proof of Proposition~\ref{prop:lower_upper_bounds}}
\label{proof:lower_upper_bounds}
The lower bound in \eqref{eq:lower_upper_bounds} follows by taking $s = 1/2$ in Proposition~\ref{prop:Chernoff} and by bounding $J_{1/2}(\Xm,\Xm')$ in~\eqref{eq:J_1/2} as
\begin{align}
J_{1/2}(\Xm,\Xm') &= \frac{1}{2} \sum_{i=1}^{T} \ln \Big(2 + \lambda_i + \frac{1}{\lambda_i}\Big) - T \ln 2 \\
&\ge \frac{1}{2} \sum_{i=1}^{T} \ln \max\Big\{\lambda_i, \frac{1}{\lambda_i}\Big\} - T \ln 2 \\
&= \frac{1}{2} \sum_{i=1}^{T} | \ln \lambda_i | - T \ln 2.
%	&= \frac{1}{2} \sqrt{\left( \sum_{i=1}^{T} \ln \Big(2 + \lambda_i + \frac{1}{\lambda_i}\Big) \right)^2} - T\ln 2 \\
%	&\ge \frac{1}{2} \sqrt{\sum_{i=1}^{T} \ln^2 \Big(2 + \lambda_i + \frac{1}{\lambda_i}\Big)} - T \ln 2  \label{eq:tmp691}\\
%	&\ge \frac{1}{2} \sqrt{\sum_{i=1}^{T} \ln^2 \lambda_i} - T \ln 2 \\
%	&= \frac{1}{2} \delta_{\rm R}(\Id_T + \Xm'{\Xm'}^\H,\Id_T + \Xm\Xm^\H) - T \ln 2
\end{align}
%where \eqref{eq:tmp691} holds because the terms $\ln \Big(2 + \lambda_i + \frac{1}{\lambda_i}\Big)$ are nonnegative. \\
%
To show the upper bound, we first write the Gamma random variables $\rv{g}_i$ as $\rv{g}_i = \sum_{j=1}^N \rv{e}_{i,j}$, $i \in [T]$, where $\{\rv{e}_{i,j}\}_{i\in [T],j\in [N]}$ are independent exponential random variables with parameter~$1$. From this and Lemma~\ref{lemma:PEP_expression}, we can bound the PEP as
\begin{align}
\mathbb{P}(\Xm\to{\Xm'}) &= \P\bigg(\sum_{i=1}^{T} \sum_{j=1}^{N} (\lambda_i - 1) \rv{e}_{i,j} \le N \sum_{i=1}^{T} \ln \lambda_i\bigg) \\
&\ge \P[{(\lambda_i - 1) \rv{e}_{i,j} \le \ln \lambda_i, \forall i\in [T], j\in [N]}] \\
&= \prod_{i=1}^{T} \prod_{j=1}^{N}  \P[(\lambda_i - 1) \rv{e}_{i,j} \le \ln \lambda_i]  \\
&= \exp\bigg(-N \sum_{i=1}^{T} f(\lambda_i)\bigg) \label{eq:tmp701}
\end{align}
where $f(\lambda) \defeq - \ln \P[(\lambda-1)\rv{e} \le \ln \lambda]$ with $\rv{e}$ being an exponential random variable with parameter~$1$. We shall show that 
\begin{align} \label{eq:f_lambda}
f(\lambda) \le |\ln \lambda| + 1, \quad \forall \lambda \ge 0.
\end{align}
\begin{itemize} %[leftmargin=*]
\item If $\lambda = 1$, \eqref{eq:f_lambda} obviously holds with equality.
\item If $\lambda < 1$, we have that $\P[(\lambda-1)\rv{e} \le \ln \lambda] = \P\big(\rv{e} \ge \frac{-\ln \lambda}{1-\lambda}\big) = \exp\big(\frac{\ln \lambda}{1-\lambda}\big)$. Thus $f(\lambda) = \frac{- \ln \lambda}{1-\lambda} = - \ln \lambda + \frac{\ln \lambda^{-1}}{\lambda^{-1}-1} < |\ln \lambda| + 1$ since $\ln \lambda^{-1} < \lambda^{-1}-1$ for all $\lambda^{-1} > 1$.
\item If $\lambda > 1$, we have that 
$
\P[(\lambda-1)\rv{e} \le \ln \lambda] = \P[\rv{e} \le \frac{\ln \lambda}{\lambda - 1}] 
= 1 - \exp\left(-\frac{\ln \lambda}{\lambda - 1}\right) 
\ge \frac{1-e^{-1}}{\lambda}.
$
To verify the inequality, notice that the function $\lambda \left(1 - \exp\left(-\frac{\ln \lambda}{\lambda - 1}\right)\right) = \lambda - \lambda^{-\frac{1}{\lambda-1}+1}$ is increasing for $\lambda > 1$, and converges from above to $1-e^{-1}$ as $\lambda$ approaches $1$ from above. We deduce that $f(\lambda) \le \ln (\frac{\lambda}{1-e^{-1}}) = \ln \lambda - \ln (1-e^{-1}) < |\ln \lambda| + 1$.
\end{itemize}
Introducing~\eqref{eq:f_lambda} into~\eqref{eq:tmp701}, we upper 
bound the PEP exponent as
\begin{align}
-\frac{1}{N} \ln \P[\Xm \to \Xm'] \le \sum_{i=1}^{T} f(\lambda_i) 
\le \sum_{i=1}^{T} |\ln \lambda_i| + T.
%	&\le \sqrt{T \sum_{i=1}^{T} \ln^2 \lambda_i} + T \label{eq:tmp726} \\
%	&= \sqrt{T} \delta_{\rm R}(\Id_T + \Xm'{\Xm'}^\H,\Id_T + \Xm\Xm^\H) + T
\end{align}
%where~\eqref{eq:tmp726} follows from the Cauchy–Schwarz inequality. 
This completes the proof.
}

%\vspace{-.5cm}
\section{Proof of Proposition~\ref{lem:trGamma_scaling}} \label{proof:lem:meanLLR_scaling}
%\vspace{-.1cm}
We have $\trace\big(( \Id_T + {\Xm'}{\Xm'}^\H)^{-1} \big) = O(1)$ since the eigenvalues of $( \Id_T + {\Xm'}{\Xm'}^\H )^{-1}$ are all smaller than $1$. 
{Following the QR decomposition, the input matrix $\Xm$ can be factorized as $\Xm = \Wm\Dm$ where $\Wm \in \CC^{T\times \Mtot}$ is a truncated unitary matrix specifying the column space of $\Xm$, and $\Dm \in \CC^{\Mtot\times \Mtot}$ is a full-rank spanning matrix.
Similarly, 
$
{\Xm'} = {\Wm'}{\Dm'},
$
for some truncated unitary matrix ${\Wm'} \in \CC^{T\times \Mtot}$ and some full-rank  spanning matrix ${\Dm'} {\in \CC^{\Mtot \times \Mtot}}$. 

If $\Span(\Xm) = \Span({\Xm'})$, we get $\Wm = \Wm'$, thus
\begin{align}
&\trace\big( ( \Id_T+ {\Xm'}{\Xm'}^\H )^{-1} \Xm \Xm^\H \big) \notag \\
&= \trace\big( ( \Id_T+ \Wm \Dm {\Dm'}^\H \Wm )^{-1} \Wm \Dm \Dm^\H \Wm^\H \big) \\
&= \trace\big( ( \Id_{\Mtot} +  \Dm {\Dm'}^\H )^{-1} \Dm \Dm^\H\big) \\
&\le  \frac{\Mtot \sigma_{\max}(\Dm^\H \Dm)}{1+\sigma_{\min}({\Dm'}^\H \Dm')}, \label{eq:tmp2549}
\end{align}
where $\sigma_{\max}(\Dm^\H \Dm)$ is the largest eigenvalue of $\Dm^\H \Dm$ and $\sigma_{\min}({\Dm'}^\H \Dm')$ is the smallest eigenvalue of ${\Dm'}^\H \Dm'$. Since $\|\Xm\vv\|^2 = \|\Dm\vv\|^2$ for any unit-norm vector $\vv \in \CC^{\Mtot}$ and $\|\Xm\vv\|^2 = \Theta(P)$ by assumption, we get that $\|\Dm\vv\|^2 = \Theta(P)$ and $\|\Dm^\H\vv\|^2 = \Theta(P)$ for any $\vv$. Taking $\vv$ as one of the eigenvectors of $\Dm^\H \Dm$, we deduce that the eigenvalues of $\Dm^\H\Dm$ scale as $\Theta(P)$. Similarly, the eigenvalues of ${\Dm'}^\H \Dm'$ also scale as $\Theta(P)$. Therefore, it follows from~\eqref{eq:tmp2549} that $\trace\big( ( \Id_T+ {\Xm'}{\Xm'}^\H )^{-1} \Xm \Xm^\H \big)$ is upper bound by a constant for large $P$, i.e., $\trace\big( ( \Id_T+ {\Xm'}{\Xm'}^\H )^{-1} \Xm \Xm^\H \big) = O(1)$.

Moreover, using the Woodbury identity $(\Id_T + {\Xm'} {\Xm'}^\H)^{-1} = \Id_T - \Xm'(\Id_M + {\Xm'}^\H\Xm')^{-1}{\Xm'}^\H$, we obtain
\begin{align}
&\trace\big( ( \Id_T+ {\Xm'}{\Xm'}^\H )^{-1} \Xm \Xm^\H \big) \notag \\
&= 
\trace\big(\big(\Id_T- {\Xm'} ( \Id_M+ {\Xm'}^\H {\Xm'} )^{-1}
{\Xm'}^\H \big) \Xm \Xm^\H \big) \\ 
&= \trace(\Xm^\H\Xm) - \trace\big(\Xm^\H {\Xm'} ( \Id_M+ {\Xm'}^\H {\Xm'} )^{-1}
{\Xm'}^\H \Xm\big) \\
&= \trace[\Xm^\H\Xm] - \trace\big(\Xm^\H \Wm' {\Dm'} ( \Id_M+ {\Dm'}^\H {\Dm'} )^{-1}
{\Dm'}^\H {\Wm'}^\H \Xm\big) \\
&\ge \trace[\Xm^\H\Xm] - \trace\big(\Xm^\H \Wm' {\Wm'}^\H \Xm\big) \label{eq:tmp2564} \\
&= \trace[\Dm^\H\Dm] - \trace\big(\Dm^\H \Wm^\H \Wm' {\Wm'}^\H \Wm \Dm\big), \label{eq:tmp2565}
\end{align}
where~\eqref{eq:tmp2564} follows since ${\Dm'} ( \Id_M+ {\Dm'}^\H {\Dm'} )^{-1}
{\Dm'}^\H \preceq \Id$. Let us denote by $\{\vv_1, \dots, \vv_{\Mtot}\}$ and $\{\mu_1, \dots,\mu_{\Mtot}\} \le 1$ respectively the eigenvectors and corresponding eigenvalues of the matrix $\Wm^\H \Wm' {\Wm'}^\H \Wm$. Then it follows from~\eqref{eq:tmp2565} that
\begin{equation}
\trace\big( ( \Id_T+ {\Xm'}{\Xm'}^\H )^{-1} \Xm \Xm^\H \big) \ge \sum_{i=1}^{\Mtot} (1-\mu_i) \|\Dm^\H \vv_i\|^2. \label{eq:tmp2570}
\end{equation}
If $\Span(\Xm) \ne \Span({\Xm'})$, there exists some $i\in [\Mtot]$ such that $\mu_i < 1$. Furthermore, $\|\Dm^\H\vv_i\|^2 = \Theta(P)$ following the reasoning above. Therefore, it follows from~\eqref{eq:tmp2570} that $\trace\big( ( \Id_T+ {\Xm'}{\Xm'}^\H )^{-1} \Xm \Xm^\H \big) = \Theta(P)$.
}

\section{Proof of Proposition~\ref{prop:min_dk}} \label{proof:min_dk}
We first show the two-user case, i.e., $\min\left\{ d_{1}(\Xcal),
d_{2}(\Xcal) \right\} \le d_{\min}(\Xcal) \le  \min\left\{
d_{1}(\Xcal), d_{2}(\Xcal) \right\} + M$ with $d_{1}(\Xcal)$ and $d_{2}(\Xcal)$ defined in \eqref{eq:d12} and \eqref{eq:d21}, respectively. To this end, we first develop $d(\Xm\to{\Xm'})$ as
\begin{align}
d(\Xm\to{\Xm'})
%\trace\big( PT( \Id_T+ PT {\Xm'}{\Xm'}^\H )^{-1} \Xm \Xm^\H\big) \\
&= \trace\big( \Xm_1^\H (\Id_T+ {\Xm'}{\Xm'}^\H )^{-1} \Xm_1 \big)  \notag \\
&\quad + \trace\big( \Xm_2^\H (\Id_T + {\Xm'}{\Xm'}^\H )^{-1} \Xm_2
\big), % \label{eq:tr}
\end{align}
where we recall that $\Xm := [\Xm_1 \ \Xm_2]$, ${\Xm'} := [\Xm'_1 \ \Xm'_2]$ with $\Xm_k,\Xm'_k\in\Xcal_k$, $k\in\{1,2\}$, and ${\Xm'}\ne\Xm$. 
%All possible combinations of $(\Xm,{\Xm'})$ with $\Xm\ne{\Xm'}$
%are in one of the three
%following cases. 
Regarding $\Xm$ and $\Xm'$ as the transmitted and detected joint symbols, respectively, there are two types of error event. On one hand, if both users are in error, i.e., $\Xm_1\ne\Xm'_1$, $\Xm_2 \ne \Xm'_2$, then
\begin{align}
d(\Xm\to{\Xm'}) &= \trace\Big( \Xm_1^\H \big( \Id_T+ \Xm'_1 {\Xm'}_1^\H + \Xm'_2 {\Xm'}_2^\H \big)^{-1} \Xm_1
\Big) \notag \\
&\quad + \trace\Big( \Xm_2^\H \big( \Id_T+ {\Xm'}_1^\H
\Xm'_1 + {\Xm'}_2^\H \Xm'_2 \big)^{-1} \Xm_2 \Big).  \notag
\end{align}
On the other hand, if only one user is in error, i.e.,
$\Xm_k=\Xm'_k$, $\Xm_{l}\ne\Xm'_{l}$, $k\ne
l\in\left\{ 1,2 \right\}$, then
\begin{align}
d(\Xm\to{\Xm'}) &= %\underbrace
{\trace\Big( \Xm_k^\H \big( \Id_T+ {\Xm}_k {\Xm}_k^\H + \Xm'_l {\Xm'}_l^\H \big)^{-1} \Xm_k
\Big)}%_{\le M} 
\notag \\
&\quad + \trace\Big( \Xm_l^\H \big( \Id_T+ {\Xm'}_k^\H
\Xm'_k + {\Xm'}_l^\H \Xm'_l \big)^{-1} \Xm_l \Big). \notag
\end{align}
It holds that $d_{\min}(\Xcal)$ is the minimal value of $d(\Xm\to{\Xm'})$ over both cases. 
Since $0 \le \trace\big( \Xm_k^\H ( \Id_T + {\Xm}_k {\Xm}_k^\H + \Xm'_l {\Xm'}_l^\H )^{-1} \Xm_k
\big) \le M$, $\forall\,k\ne l$, %considering both types of error, 
we obtain the bounds on $d_{\min}(\Xcal)$ as stated.

We now generalize the analysis of the two-user case to the $K$-user case. %We continue to assume that $M_1 = \dots = M_K = M$. 
Let us develop
\begin{align}
d(\Xm \to \Xm') = \sum_{k=1}^K \trace[\Xm_k^\H(\Id_T + \Xm'{\Xm'}^\H)^{-1}\Xm_k],
\end{align}
where we recall that $\Xm = [\Xm_1 \dots \Xm_K], \Xm' = [\Xm'_1 \dots \Xm'_K]$ with $\Xm_k,\Xm'_k \in \Xc_k, k \in [K]$ and $\Xm \ne \Xm'$. $\Xm$ and $\Xm'$ are regarded as the transmitted and detected joint symbols, respectively. For any $\Kc \subset [K]$, if all users in $\Kc$ are wrongly detected, while all users in $\Lc = [K]\setminus \Kc$ are correctly detected, then
\begin{align}
&d(\Xm \to \Xm') \notag \\
&= \sum_{k\in\Kc} \trace\bigg(\Xm_k^\H\Big(\Id_T + \Xm_k'{\Xm'_k}^\H + \sum_{j\ne k}\Xm'_j{\Xm'_j}^\H\Big)^{-1}\Xm_k\bigg) \notag \\
&\quad + \sum_{l\in\Lc} \trace\bigg(\Xm_l^\H\Big(\Id_T + \Xm_l\Xm_l^\H + \sum_{j\ne l}\Xm'_j{\Xm'_j}^\H\Big)^{-1}\Xm_l\bigg).
\end{align}
In this case, the minimal value of $d(\Xm \to \Xm')$ is defined as 
\begin{align}
d_{\min}^\Kc(\Xc) \defeq \min_{\Xm_k \ne \Xm'_k \in \Xc_k, \forall k \in \Kc, \atop \Xm_l = \Xm'_l \in \Xc_l, \forall l \in [K] \setminus \Kc}  d(\Xm \to \Xm').
\end{align}
Then, it is straightforward that $d_{\min}(\Xc)$ is the minimum value of $d_{\min}^\Kc(\Xc)$ over all possible $\Kc \subset[K]$, i.e., $d_{\min}(\Xc) = \min_{\Kc \subset [K]} d_{\min}^\Kc(\Xc)$. 
%	Let us define
%	\begin{align}
%	d_k(\Xc) \defeq \min_{\Xm_k\ne\Xm'_k\in\Xcal_k\atop
%		{\Xm}_j\in\Xcal_j, j\ne k} \trace\bigg(\Xm_k^\H \Big( \Id_T + \Xm'_k{\Xm'_k}^\H + \sum_{j\ne k}\Xm_j\Xm_j^\H \Big)^{-1} \Xm_k
%	\bigg). \label{eq:dk_Xc}
%	\end{align}
With $d_k(\Xc)$ defined in \eqref{eq:dk_Xc_1}, we have that
\begin{align}
d_{\min}^\Kc(\Xc) \ge \displaystyle\min_{k\in \Kc} d_k(\Xc) \ge \min_{k\in [K]} d_k(\Xc), \forall \Kc \subset [K],
\end{align}
where the first inequality holds since the constraint under the $\min$ in $d_k(\Xc)$ subsumes that in $d_{\min}^\Kc(\Xc)$ and the trace in $d_k(\Xc)$ is one of the summands in $d_{\min}^\Kc(\Xc)$, which are all nonegative, for any $k \in \Kc \subset [K]$; the second inequality holds since $\Kc \subset[K]$. Taking $\Kc^* = \arg \displaystyle\min_{\Kc \subset [K]} d_{\min}^\Kc(\Xc)$ yields
\begin{align}
d_{\min}(\Xc) = d_{\min}^{\Kc^*}(\Xc) \ge \min_{k\in [K]} d_k(\Xc). \label{eq:dmin_lb}
\end{align}
On the other hand, since $\trace\Big(\Xm_l^\H\big(\Id_T + \Xm_l\Xm_l^\H + \sum_{j\ne l}\Xm'_j{\Xm'_j}^\H\big)^{-1}\Xm_l\Big) \le M$, $\forall\,l \in [K]$, 
we get that $d_{\min}^{\{k\}}(\Xc) \le d_k(\Xc) + (K-1)M$ for all $k\in [K]$. Letting $k^* = \arg\min_{k\in[K]} d_k(\Xc)$, we  have that
\begin{align}
d_{\min}(\Xc) &= \min_{\Kc \subset [K]} d_{\min}^\Kc(\Xc) \\
&\le d_{\min}^{\{k^*\}}(\Xc) \\
&\le d_{k^*}(\Xc) + (K-1)M \\
&\le \min_{k\in[K]} d_k(\Xc) + (K-1)M. \label{eq:dmin_ub}
\end{align}
From \eqref{eq:dmin_lb} and \eqref{eq:dmin_ub}, we have \eqref{eq:min_dk}, and the proof is concluded.

%\vspace{-.3cm}
\section{Proof of Proposition~\ref{prop:sufficient}} \label{proof:sufficient}
%\vspace{-.1cm}
%Due to the symmetry, it is enough to focus on $d_{1}(\Xcal)$. 
Let us rewrite $\Xm'_k
{\Xm'_k}^\H + \sum_{l\ne k}{\Xm}_l \Xm_l^\H = \bar{\Xm} \bar{\Xm}^\H$ where
$\bar{\Xm} := \left[\Xm_1 \ \dots \ \Xm_{k-1} \ \Xm'_k \ \Xm_{k+1} \ \dots \ \Xm_K\right] \in \Xcal$. Then, the trace in \eqref{eq:dk_Xc_1}
becomes
\begin{align}
&\trace\big(\Xm_k^\H \big(\Id_T+ \bar{\Xm} \bar{\Xm}^\H \big)^{-1} \Xm_k
\big) \notag \\
&= \trace\left(\Xm_k^\H \Xm_k\right) - \trace \big(\Xm_1^\H\bar{\Xm} \big(\Id_T+ \bar{\Xm}^\H \bar{\Xm} \big)^{-1} \bar{\Xm}^\H \Xm_k \big)  
\\
&= PT- \trace \big(\Xm_k^\H \Um
\Sigmam \big(\Id_T+ \Sigmam^2 \big)^{-1} \Sigmam \Um^\H
\Xm_k  \big), 
%% &\ge \|\Xm_1\|^2 \left( 1 - \min\left\{ \frac{\|\Xm_1 \Vm
%% \|^2}{\|\Xm_1\|^2}, \frac{\| \Xm_1 \Xm_{12}^\H 
%% \|^2}{\|\Xm_1\|^2} \right\} \right) 
\end{align}%
where $\bar{\Xm} = \Um \Sigmam \Vm^\H$ with $\Um\in\mathbb{C}^{r\times T}$, $\Vm\in\mathbb{C}^{KM\times	r}$ being orthogonal matrices, and
$r$ being the rank of $\bar{\Xm}$; $\Sigmam$ contains
$r$ singular values of $\bar{\Xm}$ in decreasing order. 
%% Let us
%% partition the matrix $\Sigmam$ into two parts, $\Sigmam'$ with the
%% $r'$ singular values larger than $1$ and $\Sigmam''$ with $r''=r-r'$ singular
%% values smaller than $1$. Accordingly, we partition $\Vm = [ \Vm' \ \Vm'' ]$
%% with $\Vm'\in\mathbb{C}^{r'\times T}$ and $\Vm''\in\mathbb{C}^{r''\times T}$. 
Then, since $\big(\Id_T+ \Sigmam^2 \big)^{-1} \preceq
(1+\sigma^2_{\min}(\bar{\Xm}))^{-1} \Id$ with $\sigma_{\min}(\bar{\Xm})$ being the
minimum non-zero singular value of $\bar{\Xm}$, we have
\begin{align}
&\trace\Big(\Xm_k^\H \big(\Id_T+ \bar{\Xm} \bar{\Xm}^\H \big)^{-1} \Xm_k
\Big) \notag \\
%%   \MoveEqLeft[1]{\trace \left( \frac{\Xm_1}{\|\Xm_1\|} \Vm \Sigmam \left(
%%   \Id_T+ \Sigmam\Sigmam \right)^{-1} \Sigmam \Vm^\H
%%   \frac{\Xm_1^\H}{\|\Xm_1\|}  \right)}\\ 
&\ge PT - \big(1+\sigma^2_{\min}(\bar{\Xm})\big)^{-1} \trace \left( \Xm_k^\H \Um \Sigmam  \Sigmam \Um^\H \Xm_k  \right) \\
&= PT- \big(1+\sigma^2_{\min}(\bar{\Xm})\big)^{-1} \big\| \bar{\Xm}^\H
\Xm_k \big\|_{\rm F}^2 \\
&= PT - \big(1+\sigma^2_{\min}(\bar{\Xm})\big)^{-1}
\bigg(\| {\Xm'_k}^\H \Xm_k \|_{\rm F}^2 +
\sum_{l\ne k}\| \Xm_l^\H \Xm_k \|_{\rm F}^2 \bigg). \label{eq:tmp922} 
\end{align}%

From
\eqref{eq:tmp922}, the key is to find a lower bound on the non-zero singular value 
$\sigma_{\min}(\bar{\Xm})$. {To this end, we shall make use of the following lemmas.}%The following lemma is useful for that  purpose.

\begin{lemma} \label{lemma:hermitian_pertubation}
Let $\Am$ and $\Bm$ be two $T\times T$ Hermitian matrices, then %$|\sigma_i(\Am+\Bm) - \sigma_i(\Am)| \le \|\Bm\|_{\rm F}, \forall i \in [T].$
%  	\begin{align}
$|\sigma_i(\Am+\Bm) - \sigma_i(\Am)| \le \|\Bm\|_{\rm F}, \forall i \in [T]$. %\label{eq:bound_spectral_gap}
%  	\end{align}
\end{lemma}

\begin{proof}
From~\cite[Corollary 8.1.6]{GoluVanl96matrix_computations}, $|\sigma_i(\Am+\Bm) - \sigma_i(\Am)|$ is upper bounded by the spectral norm of $\Bm$. Then, the lemma follows since the spectral norm is upper bounded by the Frobenius norm.
\end{proof}

\begin{lemma} \label{lem:min_singular_value}
Let $\Qm := \left[\begin{smallmatrix} \Id_m & \Am_{m\times n} \\
\Am_{m\times n}^\H &
\Id_n \end{smallmatrix}\right]$
be positive semidefinite. Then, the $m+n$ eigenvalues of $\Qm$ are 
%  	\begin{align}
$1+\sigma_1(\Am), \ldots, 1+\sigma_{\min\{m,n\}}(\Am), 1, \ldots, 1, 
1-\sigma_{\min\{m,n\}}(\Am), \ldots, 1-\sigma_1(\Am)$. 
%  	\end{align}%
%%     In part
\end{lemma}

%\vspace{-.4cm}
\begin{proof}
The singular value decomposition of $\Am$ leads to a block
diagonalization of $\Qm$ with $2\times 2$ blocks. The result then
follows immediately. 
\end{proof}

{We proceed from
\eqref{eq:tmp922} as follows.}
\begin{itemize}
\item For any $K$, applying Lemma~\ref{lemma:hermitian_pertubation} with $\Am = \Id_{KM}$ and $\Bm = \bar{\Xm}^\H \bar{\Xm} - \frac{PT}{M}\Id_{KM}$, we have that $|\sigma_{\min}(\bar{\Xm}^\H \bar{\Xm}) - \frac{PT}{M}| \le \big\|\bar{\Xm}^\H \bar{\Xm} - \frac{PT}{M}\Id_{KM}\big\| = \sqrt{\sum_{k\ne l \in [K]}\| \Xm_l^\H \Xm_k \|_{\rm F}^2} \le PT\sqrt{K(K-1)c}$, thus $\sigma^2_{\min}(\bar{\Xm}) \ge PT\big(\frac{1}{M} - \sqrt{K(K-1)c}\big)$.

\item For $K=2$, the bound can be tightened. For $k\ne l \in \{1,2\}$, applying Lemma~\ref{lem:min_singular_value} with $\Qm = \frac{M}{PT} \bar{\Xm}^\H
\bar{\Xm}$ and $\Am = \frac{M}{PT} {\Xm'_k}^\H \Xm_l$, we see that
the minimum non-zero eigenvalues of $\Qm$ is $1-\sigma^*(\frac{M}{PT}
{\Xm'_k}^\H \Xm_l)$ if there exists at least one singular value of $\frac{M}{PT}
{\Xm'_k}^\H \Xm_l$
strictly smaller than $1$ and $\sigma^*(\frac{M}{PT}
{\Xm'_k}^\H \Xm_l)$ is the largest among such
values. Otherwise, if all singular values of $\frac{M}{PT} {\Xm'_k}^\H
\Xm_l$ are $1$, the minimum non-zero eigenvalue of $\Qm$ is two. In
any case, the minimum non-zero eigenvalue of $\Qm$ is lower
bounded %\footnote{The proposed bound can be further tightened if needed. Its current form is more for simplicity. } 
by 
$
1-\|
\frac{M}{PT} \Xm'_k \Xm_l^\H \| \ge 1-M\sqrt{c}. %1- M \bigl\|\frac{\Xm'_1}{\|\Xm'_1\|} \frac{\Xm_2^\H}{\|\Xm_2\|} \bigr\|
$
Hence, $\sigma^2_{\min}(\bar{\Xm}) \ge PT\big(\frac{1}{M} - \sqrt{c}\big)$.
\end{itemize}

Finally, plugging the bound of $\sigma^2_{\min}(\bar{\Xm})$ into \eqref{eq:tmp922} yields \eqref{eq:suff2}.

{
\section{Constellation Design Based on Precoding} \label{sec:precoding}
%\vspace{-.2cm}
In~\cite{HoangAsilomar2018multipleAccess}, we have proposed a precoding-based constellation construction for the SIMO case ($M = 1$). In this appendix, we extend that construction to the MIMO case. The idea follows from the intuition that each individual constellation should have a unique signature to help the receiver separate signals transmitted from different users. %We choose to embed the users' signatures into a set of user-specific precoders. 
If one dedicates $(K-1)M$ degrees of freedom of a user's signal for this unique signature to discriminate from the signals transmitted from $(K-1)M$ antennas of all other users, the remaining degrees of freedom for communication is $T-(K-1)M$ per antenna. Following this line, we construct $\Xc_k$ as the image of a Grassmannian constellation in $G(\CC^{T-(K-1)M},M)$ through a user-specific mapping from $G(\CC^{T-(K-1)M},M)$ to $G(\CC^{T},M)$. Specifically, we first define for each user an initial constellation $\Cc_k = \{\Cm^{(1)},\dots,\Cm^{(|\Xc_k|)}\}$ in $G(\CC^{T-(K-1)M},M)$. Then, we generate the elements of the constellation $\Xc_k$ as
\begin{align} \label{eq:precoding}
\Xm_k^{(i)} = \sqrt{P_kT}\frac{\Um_k \Cm_k^{(i)}}{\|\Um_k \Cm_k^{(i)}\|_{\rm F}}, \quad i \in [|\Xc_k|],
\end{align}
where $\pmb{U}_k \in \CC^{T\times (T - (K-1)M)}$ is a full-rank linear precoder associated to user $k$. %We refer to this kind of mapping as {\em normalized linear encoder mapping}. 
%where the normalization is to 
Therefore, each symbol $\xv_k^{(i)}$ of user $k$ belongs to the column space $\Uc_k$ of $\Um_k$. In this way, we embed the users' signatures into the set of user-specific precoders. %\footnote{This geometric separation is a discrimination factor of the signals transmitted from different users, which was exploited in \cite{HoangAsilomar2018multipleAccess} to design efficient multi-user detection. We recall that in the current paper, we focus on the ML detection.} % of dimension $T-K+1$. 
For example, when $T=3$, $K=2$, $M = 1$, and $|\Xc_1| = |\Xc_2| = 4$, a geometric interpretation for the precoders $\Um_1 = [\ev_1 \ \ev_3]$ and $\Um_2 = [\ev_2 \ \ev_3]$ is provided in Fig.~\ref{fig:geometric}. (We use $\ev_k$ to denote the $k$-th column of $\Id_T$.)} %\footnote{$\ev_k$ denotes the $k$-th canonical basis vector in $\CC^T$ with $1$ at position $k$ an 0 elsewhere.} 
%This geometric separation is the discrimination factor of the signals transmitted from different users. 

\begin{figure}[t!]
\centering
%\vspace{-.6cm}
\begin{tikzpicture}[scale=.9,font=\small]
\def\R{2.5} % sphere radius
\def\angEl{20} % elevation angle
\def\angAz{-140} % azimuth angle
\def\angPhi{60} % longitude of point P
\def\angBeta{30} % latitude of point P

%% working planes

\pgfmathsetmacro\H{\R*cos(\angEl)} % distance to north pole
\tikzset{xyplane/.style={cm={cos(\angAz),sin(\angAz)*sin(\angEl),-sin(\angAz),
cos(\angAz)*sin(\angEl),(0,0)}}}
\tikzset{acplane/.style={cm={sin(\angAz)*sin(\angEl),cos(\angAz),-sin(\angAz),
cos(\angAz)*sin(\angEl),(0,0)}}}                   		
\LongitudePlane[xzplane]{\angEl}{\angAz}
\LongitudePlane[pzplane]{\angEl}{\angPhi}
\LatitudePlane[equator]{\angEl}{0}

%% draw xyplane and sphere

%	\draw[xyplane] (-1.3*\R,-.8*\R) rectangle (2.2*\R,2.5*\R);
%	\draw[acplane] (-2*\R,-1.3*\R) rectangle (1.6*\R,2.5*\R);
\fill[ball color=white] (0,0) circle (\R); % 3D lighting effect
%\draw (0,0) circle (\R);

%% characteristic points

\coordinate (O) at (0,0);
\coordinate[mark coordinate] (I) at (.64*\R,-.26*\R);
\coordinate[mark coordinate,color=blue] (c11) at (0.3*\H,.78*\H);
\coordinate[mark coordinate,color=blue] (c12) at (0.64*\H,.11*\H);
\coordinate[mark coordinate,color=blue] (c13) at (.64*\H,-.62*\H);
\coordinate[mark coordinate,color=blue] (c14) at (0.3*\H,-1.02*\H);
%	\coordinate[mark coordinate,color=blue] (c11n) at (-0.3*\H,-.78*\H);
%	\coordinate[mark coordinate,color=blue] (c12n) at (-0.64*\H,-.11*\H);
%	\coordinate[mark coordinate,color=blue] (c13n) at (-.64*\H,.62*\H);
%	\coordinate[mark coordinate,color=blue] (c14n) at (-0.3*\H,1.02*\H);

\coordinate[mark coordinate,color=red] (c21) at (-.4*\R,-.31*\R);
\coordinate[mark coordinate,color=red] (c22) at (.3*\R,-.325*\R);
\coordinate[mark coordinate,color=red] (c23) at (.9*\R,-.15*\R);
\coordinate[mark coordinate,color=red] (c24) at (.9*\R,.15*\R);
%	\coordinate[mark coordinate,color=red] (c21n) at (.4*\R,.31*\R);
%	\coordinate[mark coordinate,color=red] (c22n) at (-.3*\R,.325*\R);
%	\coordinate[mark coordinate,color=red] (c23n) at (-.9*\R,.15*\R);
%	\coordinate[mark coordinate,color=red] (c24n) at (-.9*\R,-.15*\R);

%% draw meridians and latitude circles

\DrawLatitudeCircle[\R]{0} % equator
%\DrawLatitudeCircle[\R]{\angBeta}
%\DrawLongitudeCircle[\R]{\angAz} % xzplane
\DrawLongitudeCircle[\R]{\angAz+90} % yzplane
%\DrawLongitudeCircle[\R]{\angPhi} % pzplane

%% draw xyz coordinate system

\draw[xyplane,<->] (1.6*\R,0) node[below] {$x_2$} -- (0,0) -- (0,2*\R)
node[below] {$x_3$};
\draw[xyplane,-] (\R,-.5*\R) -- (2.2*\R,-.5*\R) node[right = 18pt] {$\mathcal{U}_2$}-- (2.2*\R,2.8*\R) -- (\R,2.8*\R);
\draw[-] (-.8*\R,1.2*\R) -- (-.8*\R,1.6*\R) node[below=12pt,right] {$\mathcal{U}_1$}-- (1.3*\R,.9*\R) -- (1.3*\R,.4*\R);
\draw[->] (0,0) -- (0,1.2*\R) node[left] {$x_1$};

%% draw lines and put labels

%	\node[above=8pt] at (4.6,1.8) {$\mathcal{U}_1$};
%	\node[above=8pt] at (5.8,-1.3) {$\mathcal{U}_2$};
\node[left=7pt,below] at (c11) {\color{blue}$\xv_1^{(1)}$};
\node[left] at (c12) {\color{blue}$\xv_1^{(2)}$};
\node[left] at (c13) {\color{blue}$\xv_1^{(3)}$};
\node[left=5pt,above] at (c14) {\color{blue}$\xv_1^{(4)}$};
\node[above=7pt,right] at (c21) {\color{red}$\xv_2^{(1)}$};
\node[above=8pt,left] at (c22) {\color{red}$\xv_2^{(2)}$};
\node[right=5pt] at (c23) {\color{red}$\xv_2^{(3)}$};
\node[right=6pt] at (c24) {\color{red}$\xv_2^{(4)}$};

\node[below=8pt,right] at (I) {$I$};
\end{tikzpicture}
%\vspace{-.2cm}
\caption{{A geometric interpretation in the real domain of the precoding-based constellations for the precoders $\Um_1 = [\ev_1 \ \ev_3]$ and $\Um_2 = [\ev_2 \ \ev_3]$, $T=3$, $K = 2$, $M=1$, and $|\Xc_1| = |\Xc_2| = 4$. The symbols of user $1$ and user $2$\textemdash represented by their projections on the unit sphere\textemdash belong to the column spaces $\Uc_1$ and $\Uc_2$ of $\Um_1$ and $\Um_2$, respectively. The axis $x_t$, $t\in \{1,2,3\}$, corresponds to the $t$-th component of a symbol~\cite{HoangAsilomar2018multipleAccess}.}
}
\label{fig:geometric}
%\vspace{-.2cm}
\end{figure}
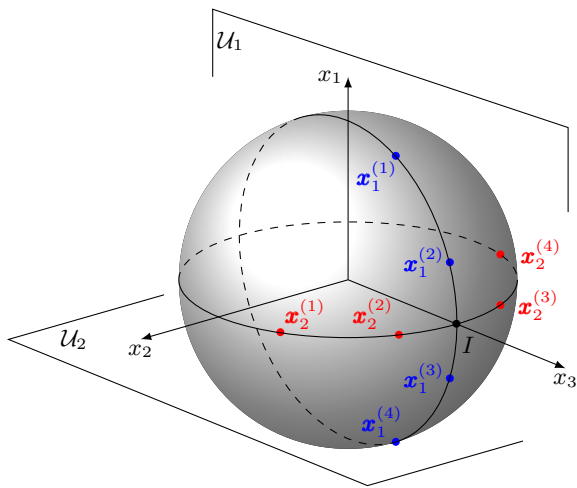

{We now design the precoders $\{\Um_k\}$. % following the insights that we have developed for the joint constellation design criteria. Specifically, leaning on Proposition~\ref{prop:sufficient}, we aim to design the precoders such that $\big\|{\Xm'_k}^\H \Xm_k \big\|_{\rm F}^2$, $\Xm \ne \Xm' \in \Xc_k, k \in [K]$, and $\|{\Xm}_k^\H\Xm_l \|_{\rm F}^2$, $\Xm_k\in\Xcal_k, \Xm_l\in\Xcal_l,k\ne l \in [K]$ are small in order to guarantee a high value of the metric $\min_{k\in[K]} d_k(\Xc_k)$. (Note that by construction~\eqref{eq:precoding}, the symbols $\{\Xm_k\}$ satisfy $\trace[\Xm_k^\H \Xm_k] = PT$, thus Proposition~\ref{prop:sufficient} applies.) 
To this end, we consider the QR factorization 
$
\Um_k = \Qm_k \Rm_k$, $k\in [K],
$
where the truncated unitary matrix $\Qm_k \in \CC^{T\times(T-(K-1)M)}$ controls the \textit{subspace} which the symbols $\Xm_k^{(i)}$ lie in, and the upper triangular matrix $\Rm_k \in \CC^{(T-(K-1)M)\times(T-(K-1)M)}$ controls the \textit{orientation} of the symbols in this subspace.

%For notational simplicity, let us assume for now that the precoders $\{\Um_k\}$ are truncated unitary matrices. 
We first design $\{\Qm_k\}$. Leaning on Proposition~\ref{prop:sufficient}, we aim to design $\{\Qm_k\}$ such that $\|{\Xm}_k^\H\Xm_l \|_{\rm F}^2$, $\Xm_k\in\Xcal_k, \Xm_l\in\Xcal_l,k\ne l \in [K]$ are small in order to guarantee a high value of the metric $\min_{k\in[K]} d_k(\Xc_k)$. (Note that by construction~\eqref{eq:precoding}, the symbols $\{\Xm_k\}$ satisfy $\trace[\Xm_k^\H \Xm_k] = PT$, thus Proposition~\ref{prop:sufficient} applies.) Let us consider two users $k$ and $l$ and assume that the column spaces of their precoders $\Um_k$ and $\Um_l$ share $r$ dimensions. Hence, we write $\Qm_k = [\Qm_0 ~\Vm_k]$ and $\Qm_l = [\Qm_0~\Vm_l]$ where $\Qm_0 \in \CC^{T \times r}$, $\Vm_k \in \CC^{T \times (T - (K-1)M - r)}$, and $\Vm_l \in \CC^{T \times (T - (K-1)M - r)}$ are a truncated unitary matrices. We impose that $\Vm_k^\H \Vm_l = \mathbf{0}$. 
A symbol of user~$k$ can be expressed as $\Xm_k = \Qm_0 \Dm_k + \Vm_k \Em_k$ with $\Dm_k \in \CC^{r \times M}$ and $\Em_k \in \CC^{(T - (K-1)M - r) \times M}$. Similarly, $\Xm_l = \Qm_0 \Dm_l + \Vm_l \Em_l$ with $\Dm_l \in \CC^{r \times M}$ and $\Em_l \in \CC^{(T - (K-1)M - r) \times M}$. Therefore, 
\begin{align}
\|\Xm^\H_k \Xm_l\|^2_{\rm F} = \|\Dm^\H_k \Dm_l\|^2_{\rm F} + \|\Em^\H_k \Vm^\H_k \Vm_l \Em_l\|^2_{\rm F} = \|\Dm^\H_k \Dm_l\|^2_{\rm F}.
\end{align}
That is, $\|\Xm^\H_k \Xm_l\|^2_{\rm F}$ depends only on the projections $\Dm_k$ and $\Dm_l$ of $\Xm_l$ and $\Xm_l$, respectively, on the shared subspace $\Span[\Qm_0]$. Therefore, to minimize $\|\Xm^\H_k \Xm_l\|^2_{\rm F}$, we minimize the dimension $r$ of $\Span[\Qm_0]$, or equivalently, maximize the number of orthogonal dimensions $T-r$ of $\Span[\Qm_k]$ and $\Span[\Qm_l]$. %More generally, to minimize $\displaystyle\max_{\Xm_k\in\Xcal_k, \Xm_l\in\Xcal_l,k\ne l \in [K]} \|{\Xm}_k^\H\Xm_l \|_{\rm F}^2$, we minimize the number of shared dimensions among the column spaces of all precoders $\{\Um_k\}$. 
With $T \ge KM$, we can design $\{\Qm_k\}$ such that their first $M$ columns are mutually orthogonal, so $T - r = 2M$ for any pair $\Qm_k$ and $\Qm_l$:\footnote{{In this way, the users' signals are orthogonal in the first $KM$ channel uses. This is in the same spirit as the pilot-based scheme where orthogonal pilots are sent in the first $KM$ channel users.}}
\begin{equation} \label{eq:precoder_Q_type1}
\Qm_k = [\ev_{(k-1)M+1} \dots \ev_{kM} \ \ev_{KM+1} \ \dots \ \ev_T], ~~ k\in [K]. %\\
%	\Vm_k &= [\ev_1 \ \dots \ \ev_{(k-1)M} \ \ev_{kM+1} \ \dots \ \ev_{KM}],
\end{equation}
If $T\ge K(K-1)M$, we can further increase $T-r$ to $2(K-1)M$ for any pair $\Qm_k$ and $\Qm_l$ with the design:
\begin{equation} \label{eq:precoder_Q_type2}
\Qm_k = [\ev_1 \ \dots \ \ev_{(k-1)(K-1)M} \ \ev_{k(K-1)M+1} \ \dots \ \ev_T],  ~~ k\in [K].%\\
%	\Vm_k &= [\ev_{(k-1)(K-1)M+1} \ \dots \ \ev_{k(K-1)M}],
\end{equation}
In this way, the orthogonal complements of $\{\Qm_k\}$ are mutually orthogonal. 

Next, we design $\{\Rm_k\}$. We let $\Rm_k = \diag(\eta_{k,1},\eta_{k,2}, \dots,\eta_{k,T - (K - 1)M})$, where $\eta_{k,i}$ indicates the weight of a symbol in the dimension of $\Uc_k$ represented by column $i$ of ${\Qm_k}$. These factors control the orientation of the symbols in $\Uc_k$. A particular choice is to set higher weights $\eta_{k,i}$ for the dimensions in the \emph{mutually exclusive} parts of $\Uc_1, \dots, \Uc_K$, and lower weights for the dimensions in the \emph{intersection} of $\Uc_1, \dots, \Uc_K$. 
In Fig.~\ref{fig:geometric}, this can be interpreted as putting the points representing the symbols further away from point $I$ representing the intersection of $\Uc_1$ and $\Uc_2$. 
% We present next an example of precoders in this spirit.
Let the weights within the mutually exclusive parts of $\Uc_1, \dots, \Uc_K$ (corresponding to the first $M$ columns of $\Qm_k$ in \eqref{eq:precoder_Q_type1} and the first $(K-1)(K-1)M$ columns of $\Qm_k$ in \eqref{eq:precoder_Q_type2}) be equally $\eta_1$ and the weights within the intersection (corresponding to the remaining columns of $\Qm_k$) be equally $\eta_2$ ($\eta_2 < \eta_1$). By determining $\eta_1$ and $\eta_2$ such that the joint symbols have equal weights in average in every dimension of $\Span[{[\Um_1 \ \dots \ \Um_K]}]$, we obtain $(\eta_1,\eta_2) = (\sqrt{K},1 )$ for $\Qm_k$ in \eqref{eq:precoder_Q_type1}, and $(\eta_1,\eta_2) = \left(\sqrt{\frac{K}{K-1}},1 \right)$ for $\Qm_k$ in \eqref{eq:precoder_Q_type2}.
% We further assume that $\sum_{i =1}^{T-(K-1)M} \eta_{k,i}^2 = \|\Um_k\|^2_{\rm F} = 1$, $\forall k \in [K]$. It follows that $\Xm_k\of{i} = \sqrt{P_k T} \Um_k \Cm_k\of{i}$ and thus, after some manipulations,
%\begin{align}
%	\|\Xm_k\ofH{i} \Xm_k\of{j}\|^2 = P_k T \sum_{\ell = 1}^{T-(K-1)M} \eta_{k,\ell}^2 
%\end{align}
}

{We summarize the proposed precoders as follows.
% 				\begin{example} Two precoder designs are provided below: 
\begin{itemize}
\item (Type-I precoder) Assuming that $T \ge KM$, let
\begin{align} \label{eq:precoder_type1}
\Um_k &= \big[\eta_1[\ev_{(k-1)M+1} \dots \ev_{kM}] \ \eta_2[\ev_{KM+1} \ \dots \ \ev_T] \big]
\end{align}
for $k\in[K]$, 
where $(\eta_1,\eta_2) = (\sqrt{K},1 )$.
\item (Type-II precoder) Assuming that $T  \ge  K(K  - 1)M$, let 
\begin{align} \label{eq:precoder_type2}
\Um_k = \big[\eta_1 [&\ev_1 \ \dots \ \ev_{(k-1)(K-1)M} \ \ev_{k(K-1)M+1} \ \\ 
&\quad \dots \ev_{K(K-1)M}] \ \eta_2 [\ev_{K(K-1)M+1}\ \dots \ \ev_T] \big]
\end{align}
for $k\in[K]$, 
where $(\eta_1,\eta_2) = \left(\sqrt{\frac{K}{K-1}},1 \right)$.
\end{itemize}}

\section{Proof of Proposition~\ref{prop:theta}} \label{proof:prop:theta}
We shall make use of the following lemma.

\begin{lemma} \label{lem:monotone}
Consider three {\em distinct} $T$-dimensional unit vectors $\av$, $\bv$, $\cv$, $\rho >0$, a variable $\theta \ge 0$, and two functions $\delta_1(\theta) \defeq \rho \av^\H (\Id_T + \rho \bv \bv^\H + \rho \theta \cv \cv^\H)^{-1} \av$ and $\delta_2(\theta) \defeq \rho \theta \av^\H (\Id_T + \rho \bv \bv^\H + \rho \theta \cv \cv^\H)^{-1} \av$. Then, $\delta_1(\theta)$ is monotonically decreasing in $\theta$ while $\delta_2(\theta)$ is strictly increasing  in $\theta$.
\end{lemma}

%\vspace{-.5cm}
\begin{proof}
After some simple manipulations, we obtain 
\begin{align}
\frac{\partial \delta_1}{\partial \theta} &= -\frac{\rho^2 |\rho \av^\H\bv \bv^\H \cv - (1+\rho) \av^\H \cv|^2}{(1+\rho + \rho(1 + \rho (1 - |\bv^\H \cv|^2))\theta)^2}, \\
\frac{\partial \delta_2}{\partial \theta} &= \frac{1}{\rho + \rho^2 (1 - |\bv^\H\cv|^2)}\notag \\
&\quad \cdot \bigg[\rho^2(1-|\av^\H\cv|^2) \notag \\
&\qquad ~~+ \rho^3\big[(1-|\av^\H\bv|^2)(1-|\bv^\H\cv|^2) - |\av^\H\bv\bv^\H\cv - \av^\H \cv|^2\big] \bigg. \notag \\
&\qquad~~ \bigg. + \frac{(1+\rho)|\rho \av^\H\bv \bv^\H \cv - (1+\rho) \av^\H \cv|^2}{(1+\rho + \rho(1 + \rho(1-|\bv^\H \cv|^2))\theta)^2}  \bigg].
\end{align}
It is obvious that $\frac{\partial \delta_1}{\partial \theta} \le 0, \forall \theta \ge 0$. Let $\{\bar{\bv}_i\}_{i=1}^{T-1}$ form an orthogonal complement of $\bv$, i.e., $\bv\bv^\H + \sum_{i=1}^{T-1} \bar{\bv}_i \bar{\bv}_i^\H = \Id_T$, we have that
\begin{align}
&(1-|\av^\H\bv|^2)(1-|\bv^\H\cv|^2) - |\av^\H\bv\bv^\H\cv - \av^\H \cv|^2 \notag \\
&= \av^\H(\Id_T - \bv\bv^\H)\av \cv^\H(\Id_T - \bv\bv^\H)\cv - |\av^\H(\Id_T - \bv\bv^\H)\cv|^2 \\
&= \Bigg(\sum_{i=1}^{T-1} |\av^\H \bar{\bv}_i|^2\Bigg) \Bigg(\sum_{i=1}^{T-1} |\cv^\H \bar{\bv}_i|^2\Bigg) - \Bigg|\sum_{i=1}^{T-1}\av^\H \bar{\bv}_i \bar{\bv}_i^\H \cv \ \Bigg|^2 \\
&\ge 0,
\end{align}
where the last inequality is due to the Cauchy–Schwarz inequality. This and $|\av^\H\cv|^2 < 1$ (since $\av \ne \cv$) imply that $\frac{\partial \delta_2}{\partial \theta} > 0, \forall \theta \ge 0$.
\end{proof}

In the remainder of the proof, the symbols $\xv_{1}$, $\xv'_1$, $\hat{\xv}_1$, $\xv_2$, $\hat{\xv}_2$, and $\hat{\xv}'_2$ implicitly satisfy $\{\xv_{1}, \xv'_1, \hat{\xv}_1\} \subset \bar{\Xc}_1$, $\{\xv_2, \hat{\xv}_2, \hat{\xv}'_2\} \subset \bar{\Xc}_2$, $\xv_{1} \ne \xv'_1$ and $\hat{\xv}_2 \ne \hat{\xv}'_2$. For notational simplicity, we write $\delta_1(\theta,\xv_{1}, \xv'_1, \xv_2)$ as $\delta_1(\theta)$ and $\delta_2(\theta,\hat{\xv}_1, \hat{\xv}_2, \hat{\xv}'_2)$ as $\delta_2(\theta)$.
\begin{enumerate}
\item From Lemma~\ref{lem:monotone}, we have that $\delta_1(\theta)$ is monotonically decreasing in $\theta$ for any $\xv_{1}, \xv'_1, \xv_2$, so $d_1(\Xc^\theta) = \displaystyle\min_{\xv_{1}, \xv'_1, \xv_2} \delta_1(\theta)$  (see \eqref{eq:d12}) is also monotonically decreasing in $\theta$. Also from Lemma~\ref{lem:monotone}, $\delta_2(\theta)$ is strictly increasing in $\theta$ for any $\hat{\xv}_1, \hat{\xv}_2, \hat{\xv}'_2$, and so is $d_2(\Xc^\theta) = \displaystyle\min_{\hat{\xv}_{2}, \hat{\xv}'_2, \hat{\xv}_1} \delta_2(\theta)$ (see \eqref{eq:d21}).
Furthermore, $\delta_1(0) = P_1T - \frac{P_1^2T^2|\xv_1^\H \xv_1'|^2}{1 + P_1T} > \frac{P_1T}{1+ P_1T} > 0 = \delta_2(0)$ for any $\xv_{1}$, $\xv'_1$, $\hat{\xv}_1$, $\xv_2$, $\hat{\xv}_2$, $\hat{\xv}'_2$, so $d_1(\Xc^\theta) > d_2(\Xc^\theta)$ at $\theta = 0$. Therefore, there exists a {\em unique} $\theta^* > 0$ such that $d_1(\Xc^{\theta^*}) = d_2(\Xc^{\theta^*})$, and thus $\theta^*$ maximizes $\min\{d_1(\Xc^\theta),d_2(\Xc^\theta)\}$. 

Let $\tilde{\theta} = \arg \max_\theta d_{\rm min}(\Xc^{\theta})$. Then $d_{\rm min}(\Xc^{\theta^*}) \le d_{\rm min}(\Xc^{\tilde{\theta}})$. Also, we have that 
\begin{align}
d_{\rm min}(\Xc^{\tilde{\theta}}) &\le \min\{d_1(\Xc^{\tilde{\theta}}),d_2(\Xc^{\tilde{\theta}})\} + 1 \label{eq:tmp2362} \\
&\le \min\{d_1(\Xc^{{\theta^*}}),d_2(\Xc^{{\theta^*}})\} + 1 \label{eq:tmp2363} \\
&\le d_{\rm min}(\Xc^{{\theta^*}}) + 1, \label{eq:tmp2364}
\end{align}
where \eqref{eq:tmp2362} and \eqref{eq:tmp2364} follow from \eqref{eq:tmp736}, and \eqref{eq:tmp2363} holds because $\theta^*$ maximizes the term $\min\{d_1(\Xc^\theta),d_2(\Xc^\theta)\}$. Therefore, $d_{\rm min}(\Xc^{\theta^*}) \le \max_\theta d_{\rm min}(\Xc^{\theta}) \le d_{\rm min}(\Xc^{\theta^*}) + 1$, implying that $\theta^*$ is approximately the solution to $\max_\theta d_{\rm min}(\Xc^{\theta})$.

\item  Since $\delta_1(\theta)$ is decreasing in $\theta$ and $\delta_2(\theta)$ is increasing in $\theta$ while $\delta_1(0) > \delta(0)$, for any 6-tuple $\{\xv_{1}, \xv'_1, \hat{\xv}_1, \xv_2, \hat{\xv}_2, \hat{\xv}'_2\}$, there exists a unique $\hat{\theta} > 0$ such that $\delta_1(\hat{\theta}) = \delta_2(\hat{\theta})$. Note that $\hat{\theta}$ is a function of $\{\xv_{1}, \xv'_1, \hat{\xv}_1, \xv_2, \hat{\xv}_2, \hat{\xv}'_2\}$. The condition $\delta_1(\hat{\theta}) = \delta_2(\hat{\theta})$ can be written as a cubic equation $a \hat{\theta}^3 + b\hat{\theta}^2 + c\hat{\theta} + d = 0$ where $a,b,c,d$ are given in \eqref{eq:tmp1344}. Note that $a > 0$. Then, \eqref{eq:tmp1344-a} follows by solving this equation for a positive root.

Recall that we denote the set of values of $\hat{\theta}$ for all possible values of $\{\xv_{1}, \xv'_1, \hat{\xv}_1, \xv_2, \hat{\xv}_2, \hat{\xv}'_2\}$ by $\Thetam$. Then $\Thetam$ is also the set of $\theta$ such that $\delta_1(\theta) = \delta_2(\theta)$ for some $\{\xv_{1}, \xv'_1, \hat{\xv}_1, \xv_2, \hat{\xv}_2, \hat{\xv}'_2\}$. 

Since $d_1(\Xc^{\theta^*}) = d_2(\Xc^{\theta^*}) \eqdef \delta(\theta^*)$, it is straightforward that $\theta^* \in \Thetam$. Let $\breve{\theta} \defeq \arg\displaystyle\min_{\hat{\theta} \in \Thetam} \delta(\hat{\theta})$ and $\breve{\delta}_k(\theta)$ be the function $\delta_k(\theta)$ with $\xv_{1}$, $\xv'_1$, $\hat{\xv}_1$, $\xv_2$, $\hat{\xv}_2$, $\hat{\xv}'_2$ satisfying $\breve{\delta}_1(\breve{\theta}) = \breve{\delta}_2(\breve{\theta}) = \delta(\breve{\theta})$. We have $\breve{\delta}_1(\theta^*) \ge d_1(\Xc^{\theta^*}) = \delta(\theta^*) \ge \delta(\breve{\theta}) = \breve{\delta}_1(\breve{\theta})$ where the first equality follows from the $\min$ in $d_1(\Xc^{\theta^*})$ and the second inequality holds because $\theta^* \in \Thetam$ and due to the definition of $\breve{\theta}$. As a consequence, $\theta^* \le \breve{\theta}$ because $\breve{\delta}_1(\theta)$ is decreasing in $\theta$. Similarly, we have that  $\breve{\delta}_2(\theta^*) \ge d_2(\Xc^{\theta^*}) =  \delta(\theta^*) \ge \delta(\breve{\theta}) = \breve{\delta}_2(\breve{\theta})$, so $\theta^*\ge \breve{\theta}$ because $\breve{\delta}_1(\theta)$ is increasing in $\theta$. We conclude that $\theta^* = \breve{\theta}$. 
%Finally, \eqref{eq:tmp1333} can be written as a cubic equation $a \theta^3 + b\theta^2 + c\theta + d = 0$ (note that $a > 0$). Then, \eqref{eq:tmp1344} follows by solving this equation for a positive root.
\end{enumerate}

\section{The Riemannian Gradient of $g(\Xc)$} \label{app:gradient}
The $n$-th constellation symbol of user~$k$ can be written as $\sqrt{\rho_k}\Sm_{k,n}$ where $\rho_k \defeq \frac{P_k T}{M}$. Here, the matrix $\Sm_{k,n} \in \CC^{T \times M}$ has unit-norm columns and represents a point in the Grassmann manifold $G(\CC^T,M)$. 
The joint constellation $\Xc$ can be equivalently represented by a collection of $K2^{B}$ of those matrices denoted by $\Sc = \{\Sm_{k,n}\}_{k \in [K], n\in [2^B]}$. Therefore, we interchangeably write $g(\Xc)$ as $g(\Sc)$. To optimize $\Xc$ for a fixed set of powers $\{P_k\}$, we optimize $\Sc$ following~\eqref{eq:tmp1102} by gradient descent on the Grassmann manifold. To this end, we need to compute the Riemannian gradient  $\nabla_{\rm R} g({\Sc})$.
According to~\cite[Sec.~3.6]{AbsMahSep2008optManifolds}, 
the Riemannian gradient can be computed by projection as $$\nabla_{\rm R} g({\Sc}) = \bigg\{ (\Id_T - \Sm_{k,n}\Sm_{k,n}^\H) \frac{\partial g(\Sc)}{\partial \Sm_{k,n}}\bigg\}_{k \in [K], n\in [2^B]},$$
%\begin{align}
%\nabla_R g(\Cm) = (\Id_T - \Cm\Cm^\H) \nabla_E g(\Cm), \label{eq:Riemannian_gradient}
%\end{align}
where $\frac{\partial g({\Sc})}{\partial \Sm_{k,n}}$ is the Euclidean derivative of $g(\Sc)$ with respect to $\Sm_{k,n}$ given by
\begin{align}
\frac{\partial g({\Sc})}{\partial \Sm_{k,n}^{(n)}} &= - \Bigg(\sum_{\Xm \ne {\Xm'} \in \Xc}  \exp \bigg(-\frac{f(\Xm, {\Xm'})}{\epsilon}\bigg) \Bigg)^{-1} \notag \\
&\quad\cdot \sum_{\substack{\Xm = \big[\sqrt{\rho_1}\Sm_{1,i_1} \ \dots \ \sqrt{\rho_K}\Sm_{K,i_K} \big] \\ \ne {\Xm'} = \big[\sqrt{\rho_1}\Sm_{1,j_1} \dots \sqrt{\rho_K}\Sm_{K,j_K}\big], \\  \{(1,i_1),\dots,(K,i_K),(1,j_1),\dots,(K,j_K)\}\ni (k,n)}} \notag \\
&\qquad \qquad \cdot \exp \bigg(-\frac{f(\Xm, {\Xm'})}{\epsilon}\bigg) \frac{\partial f(\Xm,{\Xm'})}{\partial \Sm_{k,n}}.
\end{align}

We present next the expression of the derivative $\frac{\partial f(\Xm,{\Xm'})}{\partial \Sm_{k,n}}$. For the Max-$J_{1/2,\min}$, Max-$e_{\min}$, Max-$d_{\min}$, and Min-$m_1$ criteria, $f(\Xm,{\Xm'})$ is given by $J_{1/2}(\Xm,\Xm')$, $\frac{1}{N} \meanLLR$, $d(\Xm\to\Xm')$, and $-\trace[\Xm\Xm^\H \Xm' {\Xm'}^\H]$, respectively. For $J_{1/2}(\Xm,\Xm')$ and $-\trace[\Xm\Xm^\H \Xm' {\Xm'}^\H]$, which are symmetric in $\Xm$ and $\Xm'$, we consider w.l.o.g. $(k,n) = (1,i_1)$. %We can write $\Xm \Xm' = [\rho_1 \Sm_{1,i_1}  \Sm_{1,i_1} + \bar{\Xm} \bar{\Xm}']$ where $\bar{Xm} \defeq \big[\sqrt{\rho_2}\Sm_{2,i_1} \ \dots \ \sqrt{\rho_K}\Sm_{K,i_K} \big]$. 
After some manipulations, we have that
\begin{align}
\frac{\partial J_{1/2}(\Xm,\Xm')}{\partial \Sm_{1,i_1}} &=  \frac{\rho_1}{2} (\Psim + \Psim^\H)\Sm_{1,i_1},
\end{align}
where 
\begin{align}
\Psim &\defeq \Omegam^{-1} (\Id_T + \Xm'{\Xm'}^\H)^{-1} \notag \\
&\quad - (\Id_T + \Xm{\Xm}^\H)^{-1}  (\Id_T + \Xm'{\Xm'}^\H) \Omegam^{-1} (\Id_T + \Xm{\Xm}^\H)^{-1}
\end{align} 
with 
\begin{align}
\Omegam &\defeq 2\Id_T + (\Id_T + \Xm'{\Xm'}^\H)^{-1}(\Id_T + \Xm{\Xm}^\H)  \notag \\
&\quad + (\Id_T + \Xm{\Xm}^\H)^{-1} (\Id_T + \Xm'{\Xm'}^\H).
\end{align} Furthermore,
\begin{align}
\frac{\partial (-\trace[\Xm\Xm^\H \Xm' {\Xm'}^\H])}{\partial \Sm_{1,i_1}} &= - 2 \rho_1 \Xm'{\Xm'}^\H \Sm_{1,i_1}.
\end{align}
On the other hand, for $\frac{1}{N} \meanLLR$ and $d(\Xm\to\Xm')$, which are asymmetric in $\Xm$ and $\Xm'$, we consider w.l.o.g. $(k,n) \in \{(1,i_1), (1,j_1)\}$. After some manipulations, we have that
\begin{align}
&\frac{\partial \frac{1}{N} \meanLLR}{\partial \Sm_{1,i_1}} \notag \\
&= 2 \rho_1 \big[(\Id_T + \Xm'{\Xm'}^\H)^{-1} - (\Id_T + \Xm{\Xm}^\H)^{-1}\big] \Sm_{1,i_1}, \\
&\frac{\partial \frac{1}{N} \meanLLR}{\partial \Sm_{1,j_1}} \notag \\
&= 2 \rho_1 (\Id_T + \Xm'{\Xm'}^\H)^{-1} \notag \\
&\quad \cdot [\Id_T - (\Id_T + \Xm{\Xm}^\H) (\Id_T + \Xm'{\Xm'}^\H)^{-1} ] \Sm_{1,j_1},
\end{align}
and that
\begin{align}
\frac{\partial d(\Xm\to\Xm')}{\partial \Sm_{1,i_1}} &= 2 \rho_1 (\Id_T + \Xm'{\Xm'}^\H)^{-1} \Sm_{1,i_1}, \\
\frac{\partial d(\Xm\to\Xm')}{\partial \Sm_{1,j_1}} &= - 2 \rho_1 (\Id_T + \Xm'{\Xm'}^\H)^{-1} \Xm{\Xm}^\H \notag \\
&\quad \cdot  (\Id_T + \Xm'{\Xm'}^\H)^{-1} \Sm_{1,j_1}.
\end{align}

For the Min-$m_2(\Xc)$ criterion, $g(\Xc)$ is replaced by $m_2(\Xc) = \ln \sum_{\Xm\ne \Xm'\in \Xc} \det^{-N}\Big(\Id_T - \omega \Xm\Xm^\H\Xm'{\Xm'}^\H\Big)$ with $\omega = \frac{M_{\rm tot}^2}{\|\Xm\|_{\rm F}^2 \|\Xm'\|_{\rm F}^2} = \Big(\frac{M_{\rm tot}}{T \sum_{k=1}^K P_k}\Big)^2$. As for $g(\Xc)$, we write interchangeably $m_2(\Xc)$ as $m_2(\Sc)$ and compute the Riemannian gradient of $m_2(\Sc)$ as $$\nabla_{\rm R} m_2({\Sc}) = \bigg\{ (\Id_T - \Sm_{k,n}\Sm_{k,n}^\H) \frac{\partial m_2(\Sc)}{\partial \Sm_{k,n}}\bigg\}_{k \in [K], n\in [2^B]}.$$
%\begin{align}
%\nabla_R g(\Cm) = (\Id_T - \Cm\Cm^\H) \nabla_E g(\Cm), \label{eq:Riemannian_gradient}
%\end{align}
Here,  the Euclidean derivative $\frac{\partial m_2({\Sc})}{\partial \Sm_{k,n}}$ is given by
\begin{align}
&\frac{\partial m_2({\Sc})}{\partial \Sm_{k,n}^{(n)}} \notag \\ 
&= - N \Bigg(\sum_{\Xm \ne {\Xm'} \in \Xc} \det^{-N}\Big(\Id_T - \omega \Xm\Xm^\H\Xm'{\Xm'}^\H\Big) \Bigg)^{-1} \notag \\
&\quad\cdot \sum_{\substack{\Xm = \big[\sqrt{\rho_1}\Sm_{1,i_1} \ \dots \ \sqrt{\rho_K}\Sm_{K,i_K} \big] \\ \ne {\Xm'} = \big[\sqrt{\rho_1}\Sm_{1,j_1} \dots \sqrt{\rho_K}\Sm_{K,j_K}\big], \\  \{(1,i_1),\dots,(K,i_K),(1,j_1),\dots,(K,j_K)\}\ni (k,n)}} \notag \\
&\qquad \cdot 
\det^{-N-1}\Big(\Id_T - \omega \Xm\Xm^\H\Xm'{\Xm'}^\H\Big) \notag \\
&\qquad \cdot \frac{\partial \det(\Id_T - \omega \Xm\Xm^\H\Xm'{\Xm'}^\H)}{\partial \Sm_{k,n}}.
\end{align} 
We obtain after some manipulations that
\begin{align}
&\frac{\partial \det(\Id_T - \omega \Xm\Xm^\H\Xm'{\Xm'}^\H)}{\partial \Sm_{k,n}} \notag \\
&= - \omega \det\big(\Id_T - \omega \Xm\Xm^\H\Xm'{\Xm'}^\H\big) \notag \\
&\quad \cdot \big[\Xm\Xm^\H  \Xm'{\Xm'}^\H + \Xm'{\Xm'}^\H \Xm\Xm^\H\big] \Sm_{k,n},
\end{align}
for all $(k,n) \in \{(1,i_1),\dots,(K,i_K),(1,j_1),\dots,(K,j_K)\}$.

\bibliographystyle{IEEEtran}
\bibliography{IEEEabrv,./biblio}

\begin{IEEEbiographynophoto}{Khac-Hoang Ngo}
	(Member, IEEE) received the B.E. degree (Hons.) in electronics and telecommunications from University of Engineering and Technology, Vietnam National University, Hanoi, Vietnam, in 2014; and the M.Sc. degree (Hons.) and Ph.D. degree in wireless communications from CentraleSupélec, Paris-Saclay University, France, in 2016 and 2020, respectively. His Ph.D. thesis was also realized at Paris Research Center, Huawei Technologies France. Since September 2020, he has been a postdoctoral researcher at Chalmers University of Technology, Sweden. He is also an adjunct lecturer at University of Engineering and Technology, Vietnam National University Hanoi, Vietnam. His research interests include wireless communications and information theory, with an emphasis on massive random access, edge computing, MIMO, noncoherent communications, coded caching, and network coding. He received the Marie Skłodowska-Curie Actions (MSCA) Individual Fellowship and the ``Signal, Image \& Vision Ph.D. Thesis Prize'' by Club EEA, GRETSI and GdR-ISIS, France in 2021.
\end{IEEEbiographynophoto}

\begin{IEEEbiographynophoto}{Sheng Yang}
	(Member, IEEE) received the B.E. degree in electrical engineering from Jiaotong University, Shanghai, China, in 2001, and both the engineer degree and the M.Sc. degree in electrical	engineering from Telecom ParisTech, Paris, France, in 2004. In 2007, he obtained the Ph.D. degree from Universite Pierre et Marie Curie (Paris VI). From October 2007 to November 2008, he was with Motorola Research Center in Gif-sur-Yvette, France, as a Senior Staff Research Engineer. Since December 2008, he has joined CentraleSupelec, Paris-Saclay University, where he is currently a Full Professor. From April 2015, he also
	holds an Honorary Associate Professorship in the Department of Electrical and Electronic Engineering of the University of Hong Kong (HKU). He received the 2015 IEEE ComSoc Young Researcher Award for the Europe, Middle East, and Africa Region (EMEA). He was an Associate Editor of the IEEE TRANSACTIONS ON WIRELESS COMMUNICATIONS from 2015 to 2020. He is currently an Associate Editor of the IEEE TRANSACTIONS ON INFORMATION THEORY.
\end{IEEEbiographynophoto}

\begin{IEEEbiographynophoto}{Maxime Guillaud}
	(Senior Member, IEEE) received the M.Sc. degree in electrical engineering from ENSEA, Cergy, France, in 2000, and the Ph.D.	degree in electrical engineering and communications from Telecom Paris Tech, Paris, France, in 2005. From 2000 to 2001, he was a Research Engineer with Lucent Bell Laboratories (currently Nokia), Holmdel, NJ, USA. From 2006 to 2010, he was	a Senior Researcher with FTW, Vienna, Austria.
	From 2010 to 2014, he was a Researcher with	the Vienna University of Technology, Vienna. Since 2014, he has been a Researcher with the Huawei Technologies France, where he heads the Signal and Information Processing Team. He worked on numerous aspects of the physical layer of radio access networks, including transceiver algorithms, channel modeling, machine learning, and modulation design for non-coherent and multiple access communications. He introduced the principle of relative calibration for the exploitation of channel reciprocity. He has authored over 80 research articles and holds 18 patents. He is an Associate Editor of the IEEE TRANSACTIONS ON WIRELESS COMMUNICATIONS.
\end{IEEEbiographynophoto}

\begin{IEEEbiographynophoto}{Alexis Decurninge}
	(Member, IEEE) received the Ph.D. degree in statistics from Université Pierre et Marie Curie, Paris, France, in 2015. His Ph.D. thesis on statistical methods for radar signal processing was made in collaboration with Thales Air Systems. Since 2015, he has been a Research Engineer with the Mathematical and Algorithmic Sciences Laboratory, Huawei Technologies France, Paris. His research interests focus on statistical signal processing, Riemannian geometry, and wireless communications.
\end{IEEEbiographynophoto}

\end{document}

%% file: preamble_global.tex
\usepackage{amsmath,amssymb}
\usepackage{xifthen}
\usepackage{mathtools}
\usepackage{enumerate}
\usepackage{microtype}
\usepackage{xspace}
\usepackage{bm}
\usepackage[T1]{fontenc}
\usepackage{fancyhdr}
\usepackage{lastpage}
\usepackage{bbm}

\newcommand{\mbs}[1]{\pmb{#1}}
\newcommand{\vect}[1]{{\lowercase{\mbs{#1}}}}
\newcommand{\mat}[1]{{\uppercase{\mbs{#1}}}}

\newcommand{\cond}{\,\vert\,}
\renewcommand{\Re}[1][]{\ifthenelse{\isempty{#1}}{\operatorname{Re}}{\operatorname{Re}\left(#1\right)}}
\renewcommand{\Im}[1][]{\ifthenelse{\isempty{#1}}{\operatorname{Im}}{\operatorname{Im}\left(#1\right)}}

\newcommand{\SNR}{\mathsf{snr}}

\newcommand{\av}{\vect{a}}
\newcommand{\bv}{\vect{b}}
\newcommand{\cv}{\vect{c}}

\newcommand{\ev}{\vect{e}}

\newcommand{\rv}{\vect{r}}

\newcommand{\vv}{\vect{v}}

\newcommand{\xv}{\vect{x}}
\newcommand{\yv}{\vect{y}}

\newcommand{\muv}{\vect{\mu}}

\newcommand{\Sigmam}{\pmb{\Sigma}}
\newcommand{\Gammam}{\pmb{\Gamma}}
\newcommand{\Lambdam}{\pmb{\Lambda}}

\newcommand{\Psim}{\pmb{\Psi}}
\newcommand{\Thetam}{\pmb{\Theta}}

\newcommand{\Omegam}{\pmb{\Omega}}

\newcommand{\Am}{\mat{a}}
\newcommand{\Bm}{\mat{b}}
\newcommand{\Cm}{\mat{c}}
\newcommand{\Dm}{\mat{d}}
\newcommand{\Em}{\mat{e}}

\newcommand{\Gm}{\mat{g}}

\newcommand{\Mm}{\mat{M}}

\newcommand{\Qm}{\mat{q}}
\newcommand{\Rm}{\mat{r}}
\newcommand{\Sm}{\mat{s}}

\newcommand{\Um}{\mat{u}}
\newcommand{\Vm}{\mat{V}}
\newcommand{\Wm}{\mat{w}}
\newcommand{\Xm}{\mat{x}}
\newcommand{\Ym}{\mat{y}}

\newcommand{\Bc}{{\mathcal B}}
\newcommand{\Cc}{{\mathcal C}}

\newcommand{\Kc}{{\mathcal K}}
\newcommand{\Lc}{{\mathcal L}}

\newcommand{\Nc}{{\mathcal N}}

\newcommand{\Pc}{{\mathcal P}}

\newcommand{\Sc}{{\mathcal S}}

\newcommand{\Uc}{{\mathcal U}}

\newcommand{\Xc}{{\mathcal X}}
\newcommand{\Yc}{{\mathcal Y}}

\newcommand{\CC}{\mathbb{C}}

\newcommand{\Id}{\mat{\mathrm{I}}}

\newcommand{\CN}[1][]{\ifthenelse{\isempty{#1}}{\mathcal{N}_{\mathbb{C}}}{\mathcal{N}_{\mathbb{C}}\left(#1\right)}}

\renewcommand{\P}[1][]{\ifthenelse{\isempty{#1}}{\mathbb{P}}{\mathbb{P}\left(#1\right)}}
\newcommand{\E}[1][]{\ifthenelse{\isempty{#1}}{\mathbb{E}}{\mathbb{E}\left[#1\right]}}
\newcommand{\Var}[1][]{\ifthenelse{\isempty{#1}}{\mathsf{Var}}{\mathsf{Var}\left[#1\right]}}
\newcommand{\I}[1][]{\ifthenelse{\isempty{#1}}{\mathbb{I}}{\mathbb{I}\left\{#1\right\}}}
\renewcommand{\det}[1][]{\ifthenelse{\isempty{#1}}{\mathrm{det}}{\mathrm{det}\left(#1\right)}}
\newcommand{\trace}[1][]{\ifthenelse{\isempty{#1}}{\mathrm{tr}}{\mathrm{tr}\left(#1\right)}}
\newcommand{\rank}[1][]{\ifthenelse{\isempty{#1}}{\mathrm{rank}}{\mathrm{rank}\left(#1\right)}}
\newcommand{\diag}[1][]{\ifthenelse{\isempty{#1}}{\mathrm{diag}}{\mathrm{diag}\left(#1\right)}}
\newcommand{\Cov}[1][]{\ifthenelse{\isempty{#1}}{\mathsf{Cov}}{\mathsf{Cov}\left(#1\right)}}

\newcommand{\defeq}{\triangleq}
\newcommand{\eqdef}{\triangleq}

\newtheorem{proposition}{Proposition}%[section]
\newtheorem{remark}{Remark}%[section]
\newtheorem{corollary}{Corollary}%[section]
\newtheorem{lemma}{Lemma}%[section]
\newcounter{enumi_saved}
\usepackage{answers}
\Newassociation{solution}{Solution}{solutionfile}

\AtBeginDocument{\Opensolutionfile{solutionfile}[\jobname]}
\AtEndDocument{\Closesolutionfile{solutionfile}\clearpage
}

\IfFileExists{MinionPro.sty}{
}{
}

%% file: fig/T4_1B6_2B2_P30dB_minTrace_metricsVSpower_new.tex
\begin{tikzpicture}[scale=.85,style={mark size=3pt,line width=3pt}]

\begin{axis}[%
width=2.8in,
height=2.2in,
at={(0.758in,0.481in)},
scale only axis,
unbounded coords=jump,
xmin=0,
xmax=3,
xlabel style={font=\color{white!15!black}},
xlabel={$\theta = P_2/P_1$},
ymin=0,
ymax=80,
ylabel style={font=\color{white!15!black}},
ylabel={Metrics},
%ylabel near ticks,
axis background/.style={fill=white},
xmajorgrids,
ymajorgrids,
legend style={at={(1,.01)}, anchor=south east, legend cell align=left, align=left,  fill=white, fill opacity=0.8, text opacity=1, draw=none, %draw=white!15!black, fill=white!98!black
}
]
\addplot [color=blue, line width=1pt]
  table[row sep=crcr]{%
0	inf\\
0.05	4.15202779509282\\
0.1	5.40806381138211\\
0.15	6.66727720667879\\
0.2	7.92733995807096\\
0.25	9.18775190681731\\
0.3	10.4483411620563\\
0.35	11.7090327476655\\
0.4	12.9697887378409\\
0.45	14.2305878875095\\
0.5	15.491417370058\\
0.55	16.7522689833597\\
0.6	18.0131372379639\\
0.65	19.2740183212898\\
0.7	20.5349095027604\\
0.75	21.7958087753932\\
0.8	23.0567146310045\\
0.85	24.3176259143155\\
0.9	25.578541725435\\
0.95	26.8394613529845\\
1	28.1003842271922\\
1.05	29.3613098863281\\
1.1	30.6222379522537\\
1.15	31.8831681123248\\
1.2	33.1441001058\\
1.25	34.4050337134995\\
1.3	35.6659687498413\\
1.35	36.926905056641\\
1.4	38.1878424982363\\
1.45	39.4487809576157\\
1.5	40.7097203333205\\
1.55	41.9706605369436\\
1.6	43.2316014910963\\
1.65	44.4925431277434\\
1.7	45.7534853868319\\
1.75	47.0144282151532\\
1.8	48.2753715653955\\
1.85	49.5363153953487\\
1.9	50.797259667235\\
1.95	52.0582043471415\\
2	53.319149404539\\
2.05	54.5800948118696\\
2.1	55.8410405441943\\
2.15	57.1019865788889\\
2.2	58.3629328953813\\
2.25	59.6238794749244\\
2.3	60.8848263003979\\
2.35	62.1457733561356\\
2.4	63.4067206277744\\
2.45	64.6676681021212\\
2.5	65.9286157670364\\
2.55	66.1981283233784\\
2.6	66.1961124351086\\
2.65	66.1941725286792\\
2.7	66.192304387747\\
2.75	66.1905041022641\\
2.8	66.1887680411613\\
2.85	66.1870928279028\\
2.9	66.1854753185671\\
2.95	66.1839125821548\\
3	66.1824018828635\\
};
\addlegendentry{$e_{\min} (\Xc^\theta)$}

\addplot [color=red, dashed, line width=1.2pt]
  table[row sep=crcr]{%
0	0.997506234413917\\
0.05	3.02757234929029\\
0.1	4.32342370342215\\
0.15	5.59672360533424\\
0.2	6.86399300272918\\
0.25	8.12878268703578\\
0.3	9.39231317667682\\
0.35	10.6551168957997\\
0.4	11.9174631827634\\
0.45	13.1795029212021\\
0.5	14.4413272096439\\
0.55	15.7029943037174\\
0.6	16.9645431931344\\
0.65	18.2260009578648\\
0.7	19.487386992928\\
0.75	20.7487155539101\\
0.8	22.0099973535508\\
0.85	23.2712405979625\\
0.9	24.5324516792573\\
0.95	25.7936356505394\\
1	27.0547965590729\\
1.05	28.3159376846886\\
1.1	29.5770617134505\\
1.15	30.8381708662034\\
1.2	32.0992669951121\\
1.25	33.3603516571173\\
1.3	34.6214261705033\\
1.35	35.8824916589381\\
1.4	37.1435490861087\\
1.45	38.4045992832144\\
1.5	39.6656429709777\\
1.55	40.9266807774092\\
1.6	42.1877132522505\\
1.65	43.4487408787995\\
1.7	44.7097640836537\\
1.75	45.9707832447891\\
1.8	47.2317986982945\\
1.85	48.4928107440184\\
1.9	49.7538196503247\\
1.95	51.0148256581198\\
2	52.2758289842767\\
2.05	53.5368298245588\\
2.1	54.7978283561277\\
2.15	56.058824739701\\
2.2	57.3198191214173\\
2.25	58.5808116344524\\
2.3	59.8418024004253\\
2.35	61.1027915306257\\
2.4	62.3637791270873\\
2.45	63.6247652835305\\
2.5	64.8783438362352\\
2.55	64.8751258706284\\
2.6	64.8720315501399\\
2.65	64.8690538828023\\
2.7	64.8661863940763\\
2.75	64.8634230798569\\
2.8	64.8607583645086\\
2.85	64.8581870633142\\
2.9	64.8557043488053\\
2.95	64.8533057205132\\
3	64.8509869777425\\
};
\addlegendentry{$d_{\min}(\Xc^\theta)$}

\addplot [color=OliveGreen, dashdotted, line width=1.2pt]
table[row sep=crcr]{%
0	103.38285767408\\
0.05	66.883945001063\\
0.1	65.4319160722002\\
0.15	64.9274807926502\\
0.2	64.6719002109615\\
0.25	64.517458456029\\
0.3	64.4140376324953\\
0.35	64.3399391853356\\
0.4	64.2842410222217\\
0.45	64.2408463330582\\
0.5	64.2060839202112\\
0.55	64.1776110389017\\
0.6	64.1538623573744\\
0.65	64.1337521900193\\
0.7	64.1165038540396\\
0.75	64.10154703836\\
0.8	64.0884535296086\\
0.85	64.0768955511433\\
0.9	64.0666179474916\\
0.95	64.0574191277323\\
1	64.049137710453\\
1.05	64.0416429746366\\
1.1	64.0348279084697\\
1.15	64.0286040672525\\
1.2	64.0228977139427\\
1.25	64.0176468839705\\
1.3	64.0127991259985\\
1.35	64.0083097437321\\
1.4	64.0041404137585\\
1.45	63.9981400507946\\
1.5	63.9893956321326\\
1.55	63.9805286538778\\
1.6	63.972214968884\\
1.65	63.9644043563521\\
1.7	63.9570524948522\\
1.75	63.9501201207979\\
1.8	63.9435723269748\\
1.85	63.9373779746616\\
1.9	63.931509198449\\
1.95	63.9259409871391\\
2	63.9206508274289\\
2.05	63.9156183996723\\
2.1	63.9108253170513\\
2.15	63.9062549010977\\
2.2	63.9018919877903\\
2.25	63.897722759476\\
2.3	63.8937345986894\\
2.35	63.8899159606133\\
2.4	63.8862562614646\\
2.45	63.8827457805332\\
2.5	63.8793755739632\\
2.55	63.8761373986682\\
2.6	63.8730236450161\\
2.65	63.870027277128\\
2.7	63.8671417798028\\
2.75	63.8643611112281\\
2.8	63.8616796607517\\
2.85	63.8590922110953\\
2.9	63.8565939044726\\
2.95	63.8541802121507\\
3	63.8518469070527\\
};
\addlegendentry{$d_1(\Xc^{\theta})$}

\addplot [color=violet, dotted, line width=1.2pt]
table[row sep=crcr]{%
0	0\\
0.05	2.03211862918655\\
0.1	3.3280631176626\\
0.15	4.60139597144089\\
0.2	5.86868222725041\\
0.25	7.13348215242531\\
0.3	8.39701952255437\\
0.35	9.65982818265161\\
0.4	10.9221781898523\\
0.45	12.1842208304497\\
0.5	13.4460474460845\\
0.55	14.7077164478503\\
0.6	15.96926692951\\
0.65	17.2307260433011\\
0.7	18.4921132360018\\
0.75	19.7534428012418\\
0.8	21.0147254803496\\
0.85	22.2759695013371\\
0.9	23.5371812733748\\
0.95	24.7983658630519\\
1	26.0595273284336\\
1.05	27.3206689581033\\
1.1	28.5817934452932\\
1.15	29.8429030167749\\
1.2	31.103999529656\\
1.25	32.3650845450323\\
1.3	33.6261593847059\\
1.35	34.8872251753436\\
1.4	36.1482828832044\\
1.45	37.4093333417055\\
1.5	38.6703772734929\\
1.55	39.9314153082531\\
1.6	41.1924479971949\\
1.65	42.4534758249058\\
1.7	43.7144992191219\\
1.75	44.9755185588278\\
1.8	46.2365341810093\\
1.85	47.497546386315\\
1.9	48.758555443825\\
1.95	50.0195615950886\\
2	51.2805650575576\\
2.05	52.5415660275179\\
2.1	53.8025646826037\\
2.15	55.0635611839617\\
2.2	56.3245556781203\\
2.25	57.5855482986107\\
2.3	58.8465391673768\\
2.35	60.1075283960049\\
2.4	61.3685160868011\\
2.45	62.6295023337361\\
2.5	63.890487223277\\
2.55	65.1514708351222\\
2.6	66.4124532428501\\
2.65	67.6734345144966\\
2.7	68.9344147130661\\
2.75	70.1953938969892\\
2.8	71.4563721205299\\
2.85	72.7173494341504\\
2.9	73.9783258848386\\
2.95	75.2393015164018\\
3	76.5002763697313\\
};
\addlegendentry{$d_2(\Xc^{\theta})$}

\end{axis}
\vspace{-.2cm}

\end{tikzpicture}%

%% file: fig/T4_1B6_2B2_P30dB_precode_metricsVSpower_new.tex

\begin{tikzpicture}[scale=.85,style={mark size=3pt,line width=3pt}]

\begin{axis}[%
width=2.8in,
height=2.2in,
at={(0in,0in)},
scale only axis,
unbounded coords=jump,
xmin=0,
xmax=1.5,
xlabel style={font=\color{white!15!black}},
xlabel={$\theta = P_2/P_1$},
ymin=0,
ymax=18,
ylabel style={font=\color{white!15!black}},
ylabel={Metrics},
%ylabel near ticks,
axis background/.style={fill=white},
xmajorgrids,
ymajorgrids,
legend style={at={(1,.01)}, anchor=south east, legend cell align=left, align=left,  fill=white, fill opacity=0.8, text opacity=1, draw=none, %draw=white!15!black, fill=white!98!black
}
]
\addplot [color=blue, line width=1pt]
  table[row sep=crcr]{%
0	inf\\
0.03	3.64501512198085\\
0.06	3.98985866075074\\
0.09	4.30660465153878\\
0.12	4.61521059512203\\
0.15	4.92035325553019\\
0.18	5.22370280881043\\
0.21	5.52600427349475\\
0.24	5.82764015829891\\
0.27	6.12882703105709\\
0.3	6.42969669842936\\
0.33	6.73033397710003\\
0.36	7.0307959201901\\
0.39	7.33112231783209\\
0.42	7.63134176717842\\
0.45	7.93147534950889\\
0.48	8.23153894712266\\
0.51	8.53154475294328\\
0.54	8.83150228299261\\
0.57	9.1314190727905\\
0.6	9.43130116710574\\
0.63	9.73115347122465\\
0.66	10.0309800073673\\
0.69	10.3307841048439\\
0.72	10.6305685431039\\
0.75	10.9303356607509\\
0.78	11.2300874396123\\
0.81	11.5298255702757\\
0.84	11.8295515036929\\
0.87	12.1292664921852\\
0.9	12.1456799104602\\
0.93	12.1391909888884\\
0.96	12.1331065111993\\
0.99	12.1273898076844\\
1.02	12.1220085115667\\
1.05	12.1169339458399\\
1.08	12.1121406120515\\
1.11	12.1076057617887\\
1.14	12.1033090356692\\
1.17	12.099232157744\\
1.2	12.0953586756431\\
1.23	12.0916737386624\\
1.26	12.0881639074896\\
1.29	12.0848169904248\\
1.32	12.0816219018887\\
1.35	12.0785685397609\\
1.38	12.0756476786876\\
1.41	12.072850876981\\
1.44	12.0701703951393\\
1.47	12.0675991243217\\
1.5	12.0651305233952\\
};
\addlegendentry{$e_{\min} (\Xc^\theta)$}

\addplot [color=red, dashed, line width=1.2pt]
  table[row sep=crcr]{%
0	0.997506234413926\\
0.03	1.56423606980544\\
0.06	1.87452630837152\\
0.09	2.1778609957232\\
0.12	2.47934994995769\\
0.15	2.7800823906478\\
0.18	3.0804313742072\\
0.21	3.38055929899564\\
0.24	3.68054820433897\\
0.27	3.98044400337008\\
0.3	4.2802743963002\\
0.33	4.58005708654773\\
0.36	4.87980391723759\\
0.39	5.17952311075115\\
0.42	5.47922055416984\\
0.45	5.77890057339152\\
0.48	6.07856641846369\\
0.51	6.3782205784607\\
0.54	6.6778649918877\\
0.57	6.9775011909318\\
0.6	7.27713040261925\\
0.63	7.57675362118387\\
0.66	7.87637166077457\\
0.69	8.17598519446197\\
0.72	8.47559478352804\\
0.75	8.77520089974891\\
0.78	9.07480394255265\\
0.81	9.37440425237523\\
0.84	9.67400212116434\\
0.87	9.97359780071416\\
0.9	10.2731915093402\\
0.93	10.5727834372636\\
0.96	10.8723737509902\\
0.99	11.0368624105855\\
1.02	11.0315522635552\\
1.05	11.02654479142\\
1.08	11.0218148328637\\
1.11	11.017339940249\\
1.14	11.0131000233812\\
1.17	11.0090770479521\\
1.2	11.0052547791212\\
1.23	11.0016185625368\\
1.26	10.998155136578\\
1.29	10.9948524707401\\
1.32	10.9916996260144\\
1.35	10.9886866338479\\
1.38	10.9858043908615\\
1.41	10.9830445669811\\
1.44	10.9803995250338\\
1.47	10.97786225017\\
1.5	10.9754262877411\\
};
\addlegendentry{$d_{\min}(\Xc^\theta)$}

\addplot [color=OliveGreen, dashdotted, line width=1.2pt]
table[row sep=crcr]{%
0	73.3606254863172\\
0.03	15.4411555586749\\
0.06	12.7663585744548\\
0.09	11.8238738478479\\
0.12	11.3424734091149\\
0.15	11.0502981143655\\
0.18	10.854104874167\\
0.21	10.7132699368456\\
0.24	10.6072601204733\\
0.27	10.5245796124843\\
0.3	10.4582907573228\\
0.33	10.4039586333311\\
0.36	10.3586158490946\\
0.39	10.3202019091165\\
0.42	10.2872413484093\\
0.45	10.2586498612969\\
0.48	10.2336127336897\\
0.51	10.2115059603483\\
0.54	10.1918435282897\\
0.57	10.1742412706837\\
0.6	10.1583915170106\\
0.63	10.1440449557642\\
0.66	10.130997423956\\
0.69	10.1190801296815\\
0.72	10.1081523101031\\
0.75	10.0980956453269\\
0.78	10.0888099570318\\
0.81	10.0802098598474\\
0.84	10.0722221280441\\
0.87	10.0647836053939\\
0.9	10.0578395318364\\
0.93	10.0513421930962\\
0.96	10.0452498228028\\
0.99	10.0395257036927\\
1.02	10.0341374270137\\
1.05	10.0290562785683\\
1.08	10.0242567268353\\
1.11	10.019715993899\\
1.14	10.0154136939721\\
1.17	10.0113315274028\\
1.2	10.0074530204828\\
1.23	10.0037633032472\\
1.26	10.0002489189528\\
1.29	9.99689766008719\\
1.32	9.9936984266928\\
1.35	9.99064110354673\\
1.38	9.98771645332921\\
1.41	9.98491602340353\\
1.44	9.98223206422911\\
1.47	9.9796574577447\\
1.5	9.97718565433284\\
};
\addlegendentry{$d_1(\Xc^{\theta})$}

\addplot [color=violet, dotted, line width=1.2pt]
table[row sep=crcr]{%
0	0\\
0.03	0.566738282885083\\
0.06	0.877028860541448\\
0.09	1.1803636670548\\
0.12	1.48185268209039\\
0.15	1.78258515966135\\
0.18	2.08293416797697\\
0.21	2.38306211053194\\
0.24	2.68305102924611\\
0.27	2.98294683870412\\
0.3	3.2827772399931\\
0.33	3.58255993709114\\
0.36	3.88230677349767\\
0.39	4.18202597185403\\
0.42	4.4817234194278\\
0.45	4.78140344225362\\
0.48	5.08106929048173\\
0.51	5.38072345326523\\
0.54	5.68036786917056\\
0.57	5.98000407043324\\
0.6	6.27963328411834\\
0.63	6.57925650449113\\
0.66	6.8788745457262\\
0.69	7.17848808091554\\
0.72	7.47809767135879\\
0.75	7.77770378884706\\
0.78	8.077306832821\\
0.81	8.3769071437274\\
0.84	8.67650501352309\\
0.87	8.97610069401031\\
0.9	9.27569440351142\\
0.93	9.57528633225361\\
0.96	9.87487664674796\\
0.99	10.1744654933731\\
1.02	10.4740530013291\\
1.05	10.7736392850856\\
1.08	11.0732244464229\\
1.11	11.3728085761422\\
1.14	11.6723917555087\\
1.17	11.9719740574698\\
1.2	12.2715555476933\\
1.23	12.5711362854543\\
1.26	12.870716324393\\
1.29	13.1702957131696\\
1.32	13.4698744960294\\
1.35	13.7694527132921\\
1.38	14.0690304017773\\
1.41	14.3686075951767\\
1.44	14.6681843243798\\
1.47	14.9677606177579\\
1.5	15.267336501417\\
};
\addlegendentry{$d_2(\Xc^{\theta})$}

\end{axis}
\vspace{-.2cm}
\end{tikzpicture}%

%% file: fig/T4_1B6_2B2_P30dB_partition_metricsVSpower_new.tex

\begin{tikzpicture}[scale=.85,style={mark size=3pt,line width=3pt}]

\begin{axis}[%
width=2.8in,
height=2.2in,
at={(0.758in,0.481in)},
scale only axis,
unbounded coords=jump,
xmin=0,
xmax=1,
xlabel style={font=\color{white!15!black}},
xlabel={$\theta=P_2/P_1$},
ymin=0,
ymax=10,
ylabel style={font=\color{white!15!black}},
ylabel={Metrics},
%ylabel near ticks,
axis background/.style={fill=white},
xmajorgrids,
ymajorgrids,
legend style={at={(1,.01)}, anchor=south east, legend cell align=left, align=left,  fill=white, fill opacity=0.8, text opacity=1, draw=none, %draw=white!15!black, fill=white!98!black
}
]
\addplot [color=blue, line width=1pt]
  table[row sep=crcr]{%
0	inf\\
0.02	3.16123557925052\\
0.04	3.32347369232122\\
0.06	3.48451402818753\\
0.08	3.64518165235089\\
0.1	3.80568618526132\\
0.12	3.9661049696123\\
0.14	4.12647314024778\\
0.16	4.28680894976603\\
0.18	4.44712281826182\\
0.2	4.60742112672126\\
0.22	4.76770800087144\\
0.24	4.92798622650664\\
0.26	5.08825775254373\\
0.28	5.24852398329372\\
0.3	5.40878595625439\\
0.32	5.56904445452596\\
0.34	5.72930008027598\\
0.36	5.88955330413755\\
0.38	6.04980449926198\\
0.4	6.21005396531444\\
0.42	6.37030194571397\\
0.44	6.53054864023814\\
0.46	6.69079421438317\\
0.48	6.85103880641501\\
0.5	7.01128253274767\\
0.52	7.17152549209649\\
0.54	7.33176776871821\\
0.56	7.49200943496442\\
0.58	7.65225055331239\\
0.6	7.63331182337438\\
0.62	7.58271020148712\\
0.64	7.53524793456723\\
0.66	7.49064166136827\\
0.68	7.44864112653771\\
0.7	7.40902448300679\\
0.72	7.37159437219583\\
0.74	7.33617463581772\\
0.76	7.30260754365165\\
0.78	7.27075144523872\\
0.8	7.24047877174823\\
0.82	7.21167432857374\\
0.84	7.18423383047507\\
0.86	7.15806264000054\\
0.88	7.13307467701839\\
0.9	7.1091914728851\\
0.92	7.08634134734901\\
0.94	7.06445869000354\\
0.96	7.04348333111832\\
0.98	7.02335998913884\\
1	7.00403778417535\\
};
\addlegendentry{$e_{\min} (\Xc^\theta)$}

\addplot [color=red, dashed, line width=1.2pt]
  table[row sep=crcr]{%
0	0.997506234413926\\
0.02	1.80963712096005\\
0.04	2.01563127953862\\
0.06	2.19258027439534\\
0.08	2.36147559799604\\
0.1	2.52700780695697\\
0.12	2.69081718075853\\
0.14	2.85362648714778\\
0.16	3.01580380435314\\
0.18	3.17755631950906\\
0.2	3.33900957761691\\
0.22	3.50024408911415\\
0.24	3.66131386110565\\
0.26	3.82225647391725\\
0.28	3.98309888520481\\
0.3	4.14386093555684\\
0.32	4.30455755178552\\
0.34	4.46520017943196\\
0.36	4.62579774174178\\
0.38	4.7863572981935\\
0.4	4.94688450694937\\
0.42	5.10738395612012\\
0.44	5.26785940530085\\
0.46	5.4283139645082\\
0.48	5.58875022866367\\
0.5	5.74917037999349\\
0.52	5.90957626693537\\
0.54	6.06996946560732\\
0.56	6.23035132817524\\
0.58	6.39072302126527\\
0.6	6.55108555673278\\
0.62	6.54415403243885\\
0.64	6.49692024251321\\
0.66	6.45252869790124\\
0.68	6.4107303479524\\
0.7	6.37130441354515\\
0.72	6.33405448616197\\
0.74	6.29880525552172\\
0.76	6.26539975069883\\
0.78	6.23369700312322\\
0.8	6.20357005806568\\
0.82	6.17490427545293\\
0.84	6.1475958720616\\
0.86	6.12155066601324\\
0.88	6.09668299155552\\
0.9	6.07291475778328\\
0.92	6.05017462950485\\
0.94	6.02839731215408\\
0.96	6.00752292565007\\
0.98	5.98749645455689\\
1	5.96826726391417\\
};
\addlegendentry{$d_{\rm min}(\Xc^\theta)$}

\addplot [color=OliveGreen, dashdotted, line width=1.2pt]
table[row sep=crcr]{%
	0	205.226057553434\\
	0.02	30.0169240694069\\
	0.04	18.6636078454026\\
	0.06	14.5705899178059\\
	0.08	12.4608602582233\\
	0.1	11.1740372209172\\
	0.12	10.3072293952775\\
	0.14	9.68365047066871\\
	0.16	9.21352055282319\\
	0.18	8.84640450876378\\
	0.2	8.45758760941565\\
	0.22	8.12495295894879\\
	0.24	7.84697176117705\\
	0.26	7.61119684780455\\
	0.28	7.37898497222778\\
	0.3	7.15863225187132\\
	0.32	6.96542165503605\\
	0.34	6.79462924004395\\
	0.36	6.64256719030544\\
	0.38	6.50631450842299\\
	0.4	6.38352755606984\\
	0.42	6.27230412135119\\
	0.44	6.17108416366073\\
	0.46	6.07857618714758\\
	0.48	5.99370184086263\\
	0.5	5.91555369002995\\
	0.52	5.84336264439942\\
	0.54	5.77647256177212\\
	0.56	5.71432024795938\\
	0.58	5.65641956114533\\
	0.6	5.60234867046232\\
	0.62	5.5517397619427\\
	0.64	5.50427066041639\\
	0.66	5.45965796384024\\
	0.68	5.41765138083094\\
	0.7	5.3780290323788\\
	0.72	5.34059353148878\\
	0.74	5.3051686945104\\
	0.76	5.27159676851266\\
	0.78	5.23973608264142\\
	0.8	5.20945904969814\\
	0.82	5.18065045849002\\
	0.84	5.1532060087615\\
	0.86	5.12703104943403\\
	0.88	5.10203948798164\\
	0.9	5.07815284446201\\
	0.92	5.05529942830226\\
	0.94	5.03341361964871\\
	0.96	5.01243524010689\\
	0.98	4.9923090001623\\
	1	4.97298401259887\\
};
\addlegendentry{$d_1(\Xc^{\theta})$}

\addplot [color=violet, dotted, line width=1.2pt]
table[row sep=crcr]{%
	0	0\\
	0.02	0.812581384351257\\
	0.04	1.01860716603363\\
	0.06	1.1955677121461\\
	0.08	1.36446902146098\\
	0.1	1.53000489203407\\
	0.12	1.69381673687134\\
	0.14	1.85662782319329\\
	0.16	2.01880648359116\\
	0.18	2.18056004837738\\
	0.2	2.34201414931199\\
	0.22	2.50324935246974\\
	0.24	2.66431970227748\\
	0.26	2.82526280503112\\
	0.28	2.9861056370158\\
	0.3	3.14686805253102\\
	0.32	3.30756498870422\\
	0.34	3.46820789898603\\
	0.36	3.6288057127885\\
	0.38	3.78936549446867\\
	0.4	3.94989290609896\\
	0.42	4.11039253896067\\
	0.44	4.27086815524698\\
	0.46	4.43132286712378\\
	0.48	4.5917592713058\\
	0.5	4.75217955152726\\
	0.52	4.91258555750332\\
	0.54	5.07297886644114\\
	0.56	5.23336083144099\\
	0.58	5.3937326199353\\
	0.6	5.55409524447865\\
	0.62	5.71444958760627\\
	0.64	5.87479642205471\\
	0.66	6.03513642732235\\
	0.68	6.19547020331955\\
	0.7	6.3557982816905\\
	0.72	6.51612113525422\\
	0.74	6.67643918592188\\
	0.76	6.83675281136903\\
	0.78	6.99706235068428\\
	0.8	7.15736810917489\\
	0.82	7.317670362472\\
	0.84	7.47796936004898\\
	0.86	7.63826532825509\\
	0.88	7.79855847293153\\
	0.9	7.95884898168405\\
	0.92	8.11913702585857\\
	0.94	8.27942276226475\\
	0.96	8.43970633468571\\
	0.98	8.59998787520354\\
	1	8.76026750536532\\
};
\addlegendentry{$d_2(\Xc^{\theta})$}

\end{axis}
\end{tikzpicture}%

%% file: fig/joint_SER_T5B4M2N4.tex
\begin{tikzpicture}[scale=.75,style={mark size=3pt,line width=3pt}]

\begin{axis}[%
width=4in,
height=3in,
at={(0.758in,0.481in)},
scale only axis,
xmin=8,
xmax=22,
xtick={8, 10, 12, 14, 16, 18,20,22},
xlabel style={font=\color{white!15!black}},
xlabel={SNR (dB)},
ymode=log,
ymin=1e-06,
ymax=1e-1,
yminorticks=true,
ylabel style={font=\color{white!15!black}},
ylabel={Joint Symbol Error Rate},
axis background/.style={fill=white},
xmajorgrids,
ymajorgrids,
yminorgrids,
legend style={at={(0.02,0.001)}, anchor=south west, legend cell align=left, align=left, fill=white, fill opacity=0.6, text opacity=1, draw=none, %draw=white!15!black, 
	nodes={scale=0.95}}
]

\addplot [color=red, line width=1pt, mark=diamond, mark size=4pt, mark options={solid, red}]
table[row sep=crcr]{%
	8.000000000000000   0.047270000000000 \\
	10.000000000000000   0.015240000000000 \\
	12.000000000000000   0.004600000000000 \\
	14.000000000000000   0.000960000000000 \\
	16.000000000000000   0.000210000000000 \\
	18.000000000000000   4.000e-05 \\
	20.000000000000000   7.000e-06 \\
	22 1.187500000000000e-06 \\
};
\addlegendentry{ \textbf{Max-$J_{1/2,\min}$}}

\addplot [color=blue, dashdotted, line width=1pt, mark=+, mark size=4pt, mark options={solid, blue}]
table[row sep=crcr]{%
	8.000000000000000   0.123850000000000 \\
	10.000000000000000   0.042370000000000 \\
	12.000000000000000   0.010870000000000 \\
	14.000000000000000   0.002640000000000 \\
	16.000000000000000   0.000450000000000 \\
	18.000000000000000   8.900e-05 \\
	20.000000000000000   1.800e-05 \\
	22 2.831250000000000e-06 \\
};
\addlegendentry{ \textbf{Max-$d_{\min}$}}

\addplot [color=blue, dashdotted, line width=1pt, mark=pentagon, mark size=4pt, mark options={solid, blue}]
table[row sep=crcr]{%
	8.000000000000000   0.123620000000000 \\
	10.000000000000000   0.043300000000000 \\
	12.000000000000000   0.011020000000000 \\
	14.000000000000000   0.002870000000000 \\
	16.000000000000000   0.000460000000000 \\
	18.000000000000000   1.120e-04 \\
	20.000000000000000   2.100e-05 \\
	22 3.393750000000000e-06 \\
};
\addlegendentry{\textbf{Alternating optimization of $d_{\min}$}}

\addplot [color=violet, dashed, line width=1pt, mark=triangle, mark size=4pt, mark options={solid, violet}]
table[row sep=crcr]{%
	8.000000000000000   0.158050000000000 \\
	10.000000000000000   0.081850000000000 \\
	12.000000000000000   0.041110000000000 \\
	14.000000000000000   0.022820000000000 \\
	16.000000000000000   0.012050000000000 \\
	18.000000000000000   6.690e-03 \\
	20.000000000000000   3.195e-03  \\
	22 1.338e-03 \\
};
\addlegendentry{\textbf{Type-II Precoding}}

\addplot [color=violet, dashed, line width=1pt, mark=triangle, mark size=4pt, mark options={solid, violet,rotate=180}]
table[row sep=crcr]{%
	8.000000000000000   0.053830000000000 \\
	10.000000000000000   0.020860000000000 \\
	12.000000000000000   0.006960000000000 \\
	14.000000000000000   0.002390000000000 \\
	16.000000000000000   0.000980000000000 \\
	18.000000000000000   2.750e-04  \\
	20.000000000000000   9.300e-05   \\
	22 3.100e-05\\
};
\addlegendentry{\textbf{Partitioning}}

\addplot [color=pimentgreen, dashdotted, line width=1pt, mark=x, mark size=4pt, mark options={solid, pimentgreen}]
table[row sep=crcr]{%
	8.000000000000000   0.123620000000000 \\
	10.000000000000000   0.043300000000000 \\
	12.000000000000000   0.011020000000000 \\
	14.000000000000000   0.002870000000000 \\
	16.000000000000000   0.000460000000000 \\
	18.000000000000000   9.600e-05 \\
	20.000000000000000   1.600e-05 \\
	22 2.468750000000000e-06 \\
};
\addlegendentry{ {Max-$e_{\min}$}}

\addplot [color=pimentgreen, dashdotted, line width=1pt, mark=o, mark size=3pt, mark options={solid, pimentgreen}]
table[row sep=crcr]{%
	8.000000000000000   0.269310000000000 \\
	10.000000000000000   0.147690000000000 \\
	12.000000000000000   0.073680000000000 \\
	14.000000000000000   0.035080000000000 \\
	16.000000000000000   0.016740000000000 \\
	18.000000000000000   8.174e-03 \\
	20.000000000000000   3.917e-03 \\
	22 1.876e-03 \\
};
\addlegendentry{ Min-$m_1$}

\addplot [color=pimentgreen, dashdotted, line width=1pt, mark=square, mark size=3pt, mark options={solid, pimentgreen}]
table[row sep=crcr]{%
	8.000000000000000   0.070780000000000 \\
	10.000000000000000   0.029530000000000 \\
	12.000000000000000   0.011210000000000 \\
	14.000000000000000   0.004250000000000 \\
	16.000000000000000   0.001880000000000 \\
	18.000000000000000   6.050e-04 \\
	20.000000000000000   2.550e-04 \\
	22 8.600e-05 \\
};
\addlegendentry{ Min-$m_2$}

\addplot [color=black, dotted, line width=1pt, mark=*, mark size=2pt, mark options={solid, black}]
table[row sep=crcr]{%
	8.000000000000000   0.162840000000000 \\
	10.000000000000000   0.061780000000000 \\
	12.000000000000000   0.018250000000000 \\
	14.000000000000000   0.004190000000000 \\
	16.000000000000000   0.000800000000000 \\
	18.000000000000000   1.740e-04  \\
	20.000000000000000   3.000e-05	 \\
	22 5.062500000000000e-06 \\
};
\addlegendentry{ Pilot-based scheme}

%\draw [black] (axis cs:1,1) ellipse [x radius=0.1825, y radius=0.809524];
-\end{axis}
\vspace{-.2cm}
%\begin{axis}[%
%width=4.2in,
%height=2.95in,
%at={(0.758in,0.481in)},
%scale only axis,
%xmin=0,
%xmax=1,
%ymin=0,
%ymax=1,
%axis line style={draw=none},
%ticks=none,
%axis x line*=bottom,
%axis y line*=left
%]
%\draw [black] (axis cs:0.40,0.745) ellipse [x radius=3, y radius=11];
%\draw [-latex,thick] (axis cs:.5,.9) node [right] {$R_1 = R_2 = B/T, P_1 = P_2 = P$} -- (axis cs:.415,.84);
%\draw [-latex,thick] (axis cs:.17,.625) node [below,text width=3.2cm] {$R_k \!=\! 2B/T, R_{3-k}\!=\! 0$, $P_k \!=\! 2P, P_{3-k} \!=\! 0$} -- (axis cs:.19,.65);
%\end{axis}
\end{tikzpicture}%

%% file: fig/metric_beta_T5B4K2M2.tex
\begin{tikzpicture}[scale=.7,style={mark size=3pt,line width=3pt},spy using outlines={mark size = 2pt, rectangle,lens={scale=3.2}, height=2cm,width=3cm, connect spies}]
\begin{axis}[%
width=4.3in,
height=3in,
at={(0.758in,0.481in)},
scale only axis,
xmin=4,
xmax=32,
xtick={ 4,  8, 12, 16, 20, 24, 28, 32},
xlabel style={font=\color{white!15!black}},
xlabel={SNR (dB)},
%ymode=log,
ymin=1,
ymax=15,
yminorticks=true,
ytick={1,3,6,9,12,15},
ylabel style={font=\color{white!15!black}},
ylabel={\large $b_{\min}(\Xc)$},
axis background/.style={fill=white},
xmajorgrids,
ymajorgrids,
yminorgrids,
legend style={at={(0.02,1)}, anchor=north west, legend cell align=left, align=left, fill=white, fill opacity=0.6, text opacity=1, draw=none %, nodes={scale=1.1}%draw=white!15!black, 
	}
]

\addplot [color=red,  line width=1pt, mark=diamond, mark size=4pt, mark options={solid, red}]
table[row sep=crcr]{%
	4 3.3071 \\
	8 4.5815 \\
	12 6.1692 \\
	16 7.7704 \\
	20 9.4986 \\
	24 11.2921 \\
	28 13.1142 \\
	32 14.9483 \\
};
\addlegendentry{ \textbf{Max-$J_{1/2,\min}$}}

\addplot [color=blue, dashdotted, line width=1pt, mark=+, mark size=4pt, mark options={solid, blue}]
table[row sep=crcr]{%
	4 2.4915 \\
	8 3.7542 \\
	12 5.2758 \\
	16 6.9649 \\
	20 8.7407 \\
	24 10.5554 \\
	28 12.3857 \\
	32 14.2228 \\
};
\addlegendentry{ \textbf{Max-$d_{\min}$}}

\addplot [color=blue, dashdotted, line width=1pt, mark=pentagon, mark size=4pt, mark options={solid, blue}]
table[row sep=crcr]{%
	4 2.2648 \\
	8 3.4353 \\
	12 4.8843 \\
	16 6.531 \\
	20 8.2863 \\
	24 10.0922 \\
	28 11.9196 \\
	32 13.7557 \\
};
\addlegendentry{\textbf{Alternating optimization of $d_{\min}$}}

\addplot [color=violet, dashed, line width=1pt, mark=triangle, mark size=4pt, mark options={solid, violet}]
table[row sep=crcr]{%
	4 1.4286 \\
	8 1.6932 \\
	12 1.9863 \\
	16 2.4515 \\
	20 3.2228 \\
	24 4.3598 \\
	28 5.8133 \\
	32 7.4707 \\
};
\addlegendentry{\textbf{Type-II Precoding}}

\addplot [color=violet, dashed, line width=1pt, mark=triangle, mark size=4pt, mark options={solid, violet,rotate=180}]
table[row sep=crcr]{%
	4 3.1743 \\
	8 3.8186 \\
	12 4.3219 \\
	16 4.8373 \\
	20 5.5911 \\
	24 6.6916 \\
	28 8.1161 \\
	32 9.7573 \\
};
\addlegendentry{\textbf{Partitioning}}

\addplot [color=pimentgreen, dashdotted, line width=1pt, mark=x, mark size=4pt, mark options={solid, pimentgreen}]
table[row sep=crcr]{%
	4 2.4915 \\
	8 3.7542 \\
	12 5.2758 \\
	16 6.9649 \\
	20 8.7407 \\
	24 10.5554 \\
	28 12.3857 \\
	32 14.2228 \\
};
\addlegendentry{ {Max-$e_{\min}$}}

\addplot [color=pimentgreen, dashdotted, line width=1pt, mark=o, mark size=3pt, mark options={solid, pimentgreen}]
table[row sep=crcr]{%
	4 1.2134 \\
	8 1.4178 \\
	12 1.6429 \\
	16 1.9989 \\
	20 2.5823 \\
	24 3.4339 \\
	28 4.6268 \\
	32 6.1116 \\
};
\addlegendentry{ Min-$m_1$}

\addplot [color=pimentgreen, dashdotted, line width=1pt, mark=square, mark size=3pt, mark options={solid, pimentgreen}]
table[row sep=crcr]{%
	4 2.3536 \\
	8 3.1537 \\
	12 3.6702 \\
	16 4.2627 \\
	20 5.0855 \\
	24 6.228 \\
	28 7.6716 \\
	32 9.3201 \\
};
\addlegendentry{ Min-$m_2$}

\addplot [color=black, dotted, line width=1pt, mark=*, mark size=2pt, mark options={solid, black}]
table[row sep=crcr]{%
	4 2.047 \\
	8 3.1486 \\
	12 4.5439 \\
	16 6.1579 \\
	20 7.8969 \\
	24 9.6956 \\
	28 11.52 \\
	32 13.355 \\
};
\addlegendentry{ Pilot-based scheme}

\end{axis}
%\spy [black] on (11.2,7.45)
%in node [left] at (19,3);
\end{tikzpicture}%

%% file: fig/metric_Js_T5B4K2M2.tex
\begin{tikzpicture}[scale=.7,style={mark size=3pt,line width=3pt},spy using outlines={mark size = 2pt, rectangle,lens={scale=3.2}, height=2cm,width=3cm, connect spies}]
\begin{axis}[%
width=4.5in,
height=3in,
at={(0.758in,0.481in)},
scale only axis,
xmin=4,
xmax=32,
xtick={ 4,  8, 12, 16, 20, 24, 28, 32},
xlabel style={font=\color{white!15!black}},
xlabel={SNR (dB)},
%ymode=log,
ymin=0,
ymax=6,
yminorticks=true,
ylabel style={font=\color{white!15!black}},
ylabel={\large $J_{1/2, \min}$},
axis background/.style={fill=white},
xmajorgrids,
ymajorgrids,
yminorgrids,
legend style={at={(0.02,1)}, anchor=north west, legend cell align=left, align=left, fill=white, fill opacity=0.6, text opacity=1, draw=none %, nodes={scale=1.1}%draw=white!15!black, 
	}
]

\addplot [color=red, line width=1pt, mark=diamond, mark size=4pt, mark options={solid, red}]
table[row sep=crcr]{%
	4 0.43527 \\
	8 0.84572 \\
	12 1.4599 \\
	16 2.229 \\
	20 3.0894 \\
	24 3.9518 \\
	28 4.8712 \\
	32 5.8075 \\
};
\addlegendentry{ \textbf{Max-$J_{1/2,\min}$}}

\addplot [color=blue, dashdotted, line width=1pt, mark=+, mark size=4pt, mark options={solid, blue}]
table[row sep=crcr]{%
	4 0.3652 \\
	8 0.77561 \\
	12 1.3897 \\
	16 2.1567 \\
	20 3.0092 \\
	24 3.9016 \\
	28 4.8107 \\
	32 5.7268 \\
};
\addlegendentry{ \textbf{Max-$d_{\min}$}}

\addplot [color=blue, dashdotted, line width=1pt, mark=pentagon, mark size=4pt, mark options={solid, blue}]
table[row sep=crcr]{%
	4 0.3048 \\
	8 0.66153 \\
	12 1.2226 \\
	16 1.9541 \\
	20 2.7884 \\
	24 3.6726 \\
	28 4.5786 \\
	32 5.4936 \\
};
\addlegendentry{\textbf{Alternating optimization of $d_{\min}$}}

\addplot [color=violet, dashed, line width=1pt, mark=triangle, mark size=4pt, mark options={solid, violet}]
table[row sep=crcr]{%
	4 0.11666 \\
	8 0.16408 \\
	12 0.22462 \\
	16 0.33754 \\
	20 0.56605 \\
	24 0.97766 \\
	28 1.591 \\
	32 2.357 \\
};
\addlegendentry{\textbf{Type-II Precoding}}

\addplot [color=violet, dashed, line width=1pt, mark=triangle, mark size=4pt, mark options={solid, violet,rotate=180}]
table[row sep=crcr]{%
	4 0.33501 \\
	8 0.4955 \\
	12 0.67014 \\
	16 0.81537 \\
	20 1.0488 \\
	24 1.4532 \\
	28 2.0567 \\
	32 2.8159 \\
};
\addlegendentry{\textbf{Partitioning}}

\addplot [color=pimentgreen, dashdotted, line width=1pt, mark=x, mark size=4pt, mark options={solid, pimentgreen}]
table[row sep=crcr]{%
	4 0.3652 \\
	8 0.77561 \\
	12 1.3897 \\
	16 2.1567 \\
	20 3.0092 \\
	24 3.9016 \\
	28 4.8107 \\
	32 5.7268 \\
};
\addlegendentry{ {Max-$e_{\min}$}}

\addplot [color=pimentgreen, dashdotted, line width=1pt, mark=o, mark size=3pt, mark options={solid, pimentgreen}]
table[row sep=crcr]{%
	4 0.082473 \\
	8 0.10578 \\
	12 0.1326 \\
	16 0.18272 \\
	20 0.2929 \\
	24 0.52377 \\
	28 0.94057 \\
	32 1.559 \\
};
\addlegendentry{ Min-$m_1$}

\addplot [color=pimentgreen, dashdotted, line width=1pt, mark=square, mark size=3pt, mark options={solid, pimentgreen}]
table[row sep=crcr]{%
	4 0.21088 \\
	8 0.37594 \\
	12 0.49676 \\
	16 0.67296 \\
	20 0.90646 \\
	24 1.3047 \\
	28 1.9025 \\
	32 2.658 \\
};
\addlegendentry{ Min-$m_2$}

\addplot [color=black, dotted, line width=1pt, mark=*, mark size=2pt, mark options={solid, black}]
table[row sep=crcr]{%
	4 0.2512 \\
	8 0.56454 \\
	12 1.0819 \\
	16 1.7826 \\
	20 2.6004 \\
	24 3.4772 \\
	28 4.38 \\
	32 5.2937 \\
};
\addlegendentry{ Pilot-based scheme}

\end{axis}
%\spy [black] on (11.2,7.45)
%in node [left] at (19,3);
\end{tikzpicture}%

%% file: fig/metric_emin_dmin_T5B4K2M2.tex
\begin{tikzpicture}[scale=.7,style={mark size=3pt,line width=3pt},spy using outlines={mark size = 2pt, rectangle,lens={scale=3.2}, height=1.8cm,width=3cm, connect spies}]
\begin{axis}[%
width=4.5in,
height=3in,
at={(0.758in,0.481in)},
scale only axis,
xmin=4,
xmax=32,
xtick={ 4,  8, 12, 16, 20, 24, 28, 32},
xlabel style={font=\color{white!15!black}},
xlabel={SNR (dB)},
ymode=log,
ymin=1,
ymax=3000,
yminorticks=true,
ylabel style={font=\color{white!15!black}},
ylabel={\large $e_{\min}(\Xc)$ (lines) and $d_{\min}(\Xc)$ (markers)},
axis background/.style={fill=white},
xmajorgrids,
ymajorgrids,
yminorgrids,
legend style={at={(0.02,1)}, anchor=north west, legend cell align=left, align=left, fill=white, fill opacity=0.6, text opacity=1, draw=none %, nodes={scale=1.1}%draw=white!15!black, 
	}
]

\addplot [color=red, line width=1pt, mark=diamond, mark size=4pt, mark options={solid, red}]
table[row sep=crcr]{%
	0 0 \\
	2 .1 \\
};
\addlegendentry{ \textbf{Max-$J_{1/2,\min}$}}

\addplot [color=blue, dashdotted, line width=1pt, mark=+, mark size=4pt, mark options={solid, blue}]
table[row sep=crcr]{%
	0 0 \\
	2 .1 \\
};
\addlegendentry{ \textbf{Max-$d_{\min}$}}

\addplot [color=blue, dashdotted, line width=1pt, mark=pentagon, mark size=4pt, mark options={solid, blue}]
table[row sep=crcr]{%
	0 0 \\
	2 .1 \\
};
\addlegendentry{\textbf{Alternating optimization of $d_{\min}$}}

\addplot [color=violet, dashed, line width=1pt, mark=triangle, mark size=4pt, mark options={solid, violet}]
table[row sep=crcr]{%
	0 0 \\
	2 .1 \\
};
\addlegendentry{\textbf{Type-II Precoding}}

\addplot [color=violet, dashed, line width=1pt, mark=triangle, mark size=4pt, mark options={solid, violet,rotate=180}]
table[row sep=crcr]{%
	0 0 \\
	2 .1 \\
};
\addlegendentry{\textbf{Partitioning}}

\addplot [color=pimentgreen, dashdotted, line width=1pt, mark=x, mark size=4pt, mark options={solid, pimentgreen}]
table[row sep=crcr]{%
	0 0 \\
	2 .1 \\
};
\addlegendentry{ {Max-$e_{\min}$}}

\addplot [color=pimentgreen, dashdotted, line width=1pt, mark=o, mark size=3pt, mark options={solid, pimentgreen}]
table[row sep=crcr]{%
	0 0 \\
	2 .1 \\
};
\addlegendentry{ Min-$m_1$}

\addplot [color=pimentgreen, dashdotted, line width=1pt, mark=square, mark size=3pt, mark options={solid, pimentgreen}]
table[row sep=crcr]{%
	0 0 \\
	2 .1 \\
};
\addlegendentry{ Min-$m_2$}

\addplot [color=black, dotted, line width=1pt, mark=*, mark size=2pt, mark options={solid, black}]
table[row sep=crcr]{%
	0 0 \\
	2 .1 \\
};
\addlegendentry{ Pilot-based scheme}

%--------------------

\addplot [color=red, line width=1pt]
table[row sep=crcr]{%
	4 6.5133 \\
	8 7.5748 \\
	12 9.1456 \\
	16 11.848 \\
	20 17.377 \\
	24 30.1502 \\
	28 61.4267 \\
	32 139.5224 \\
};
%\addlegendentry{ \textbf{Max-$J_{1/2,\min}$}}

\addplot [color=pimentgreen, dashdotted, line width=1pt]
table[row sep=crcr]{%
	4 6.9632 \\
	8 9.7876 \\
	12 17.3551 \\
	16 35.8706 \\
	20 82.9834 \\
	24 199.9199 \\
	28 492.9513 \\
	32 1229.8917 \\
};
%\addlegendentry{ \textbf{Max-$e_{\min}$}}

\addplot [color=blue, dashdotted, line width=1pt]
table[row sep=crcr]{%
	4 6.7632 \\
	8 9.6876 \\
	12 17.0551 \\
	16 35.5706 \\
	20 82.0834 \\
	24 198.9199 \\
	28 492.2513 \\
	32 1228.8917 \\
};
%\addlegendentry{ \textbf{Max-$d_{\min}$}}

\addplot [color=blue, dashdotted, line width=1pt]
table[row sep=crcr]{%
	4 6.4254 \\
	8 8.751 \\
	12 14.5845 \\
	16 29.2311 \\
	20 66.0182 \\
	24 158.4215 \\
	28 390.5277 \\
	32 973.5517 \\
};
%\addlegendentry{\textbf{Alternating optimization of $d_{\min}$}}

\addplot [color=violet, dashed, line width=1pt]
table[row sep=crcr]{%
	4 5.4362 \\
	8 5.6141 \\
	12 5.8583 \\
	16 6.3632 \\
	20 7.5767 \\
	24 10.6013 \\
	28 18.1906 \\
	32 37.2505 \\
};
%\addlegendentry{\textbf{Type-II Precoding}}

\addplot [color=violet, dashed, line width=1pt]
table[row sep=crcr]{%
	4 6.261 \\
	8 6.9204 \\
	12 7.91 \\
	16 8.9682 \\
	20 10.9159 \\
	24 14.0241 \\
	28 21.688 \\
	32 40.8805 \\
};
%\addlegendentry{\textbf{Partitioning}}

\addplot [color=pimentgreen, dashdotted,line width=1pt]
table[row sep=crcr]{%
	4 5.3436 \\
	8 5.4456 \\
	12 5.5652 \\
	16 5.7966 \\
	20 6.3475 \\
	24 7.7187 \\
	28 11.1579 \\
	32 19.7948 \\
};
%\addlegendentry{ Min-$m_1$}

\addplot [color=pimentgreen, dashdotted,line width=1pt]
table[row sep=crcr]{%
	4 5.8481 \\
	8 6.5402 \\
	12 6.9923 \\
	16 7.6428 \\
	20 9.0274 \\
	24 12.3945 \\
	28 20.8056 \\
	32 41.9145 \\
};
%\addlegendentry{ Min-$m_2$}

\addplot [color=black, dotted, line width=1pt]
table[row sep=crcr]{%
	4 6.1423 \\
	8 8.0346 \\
	12 12.8014 \\
	16 24.7811 \\
	20 54.8753 \\
	24 130.4694 \\
	28 320.3537 \\
	32 797.3216 \\
};
%\addlegendentry{ Pilot-based scheme}

%------------------------------

\addplot [color=red, dashed, line width=1pt, mark=diamond, mark size=4pt, mark options={solid, red}, only marks]
table[row sep=crcr]{%
	4 3.808 \\
	8 4.5737 \\
	12 5.5802 \\
	16 7.6392 \\
	20 12.6015 \\
	24 24.9785 \\
	28 56.0325 \\
	32 134.0225 \\
};
%\addlegendentry{ \textbf{Max-$J_{1/2,\min}$}}

\addplot [color=pimentgreen, dashdotted, line width=1pt, mark=x, mark size=4pt, mark options={solid, pimentgreen}, only marks]
table[row sep=crcr]{%
	4 5.3228 \\
	8 8.5992 \\
	12 16.2337 \\
	16 34.7211 \\
	20 81.3635 \\
	24 198.9119 \\
	28 492.2481 \\
	32 1228.8904 \\
};
%\addlegendentry{ \textbf{Max-$e_{\min}$}}

\addplot [color=blue, dashdotted, line width=1pt, mark=+, mark size=4pt, mark options={solid, blue}, only marks]
table[row sep=crcr]{%
	4 5.1228 \\
	8 8.3992 \\
	12 15.9337 \\
	16 34.5211 \\
	20 81.0635 \\
	24 197.9119 \\
	28 491.2481 \\
	32 1227.8904 \\
};
%\addlegendentry{ \textbf{Max-$d_{\min}$}}

\addplot [color=blue, dashdotted, line width=1pt, mark=pentagon, mark size=4pt, mark options={solid, blue}, only marks]
table[row sep=crcr]{%
	4 4.8217 \\
	8 7.4835 \\
	12 13.4729 \\
	16 28.1858 \\
	20 65 \\
	24 157.4143 \\
	28 389.5248 \\
	32 972.5506 \\
};
%\addlegendentry{\textbf{Alternating optimization of $d_{\min}$}}

\addplot [color=violet, dashed, line width=1pt, mark=triangle, mark size=4pt, mark options={solid, violet}, only marks]
table[row sep=crcr]{%
	4 3.5723 \\
	8 4.0185 \\
	12 4.3901 \\
	16 4.9813 \\
	20 6.2906 \\
	24 9.3416 \\
	28 16.9347 \\
	32 35.9961 \\
};
%\addlegendentry{\textbf{Type-II Precoding}}

\addplot [color=violet, dashed, line width=1pt, mark=triangle, mark size=4pt, mark options={solid, violet,rotate=180}, only marks]
table[row sep=crcr]{%
	4 4.0482 \\
	8 4.864 \\
	12 6.0678 \\
	16 7.5022 \\
	20 9.8663 \\
	24 13.2769 \\
	28 20.9522 \\
	32 40.1493 \\
};
%\addlegendentry{\textbf{Partitioning}}

\addplot [color=pimentgreen, dashdotted, line width=1.5pt, mark=o, mark size=3pt, mark options={solid, pimentgreen}, only marks]
table[row sep=crcr]{%
	4 3.7665 \\
	8 4.1926 \\
	12 4.4602 \\
	16 4.754 \\
	20 5.3304 \\
	24 6.7118 \\
	28 10.1551 \\
	32 18.7937 \\
};
%\addlegendentry{ Min-$m_1$}

\addplot [color=pimentgreen, dashdotted, line width=1.5pt, mark=square, mark size=3pt, mark options={solid, pimentgreen}, only marks]
table[row sep=crcr]{%
	4 3.59 \\
	8 4.2808 \\
	12 4.9581 \\
	16 5.6518 \\
	20 7.0406 \\
	24 10.4088 \\
	28 18.8203 \\
	32 39.9294 \\
};
%\addlegendentry{ Min-$m_2$}

\addplot [color=black, dotted, line width=1pt, mark=*, mark size=2pt, mark options={solid, black}, only marks]
table[row sep=crcr]{%
	4 4.5537 \\
	8 6.7762 \\
	12 11.6941 \\
	16 23.7377 \\
	20 53.8579 \\
	24 129.4625 \\
	28 319.351 \\
	32 796.3205 \\
};
%\addlegendentry{ Pilot-based scheme}

\end{axis}
%\spy [black] on (11.2,6.8)
%in node [left] at (19,3);
\end{tikzpicture}%

%% file: fig/metric_minChordal_T5B4K2M2.tex
\begin{tikzpicture}[scale=.72]
\begin{axis}[
ybar = .2,
ymin=2,
ymax=6.5,
width=5in,
height=2.5in,
enlarge y limits = 0.12,
bar width=20pt,
ylabel={$m_1(\Xc)$},
symbolic x coords={{Max-$J_{1/2,\min}$},{Max-$d_{\min}$},{Alternating Opt.},{Precoding},{Partitioning},{Max-$e_{\min}$},{Min-$m_1$},{Min-$m_2$},{Pilot-Based}},
xticklabels={{\textbf{Max-$J_{1/2,\min}$}},{\textbf{Max-$d_{\min}$}},{\textbf{Alternating Opt.}},{\textbf{Type-II~Precoding}},{\textbf{Partitioning}},{{Max-$e_{\min}$}},{Min-$m_1$},{Min-$m_2$},{Pilot-Based}},
ytick={2,3,4,5,6},
xtick=data,
%x tick label style={rotate=45, anchor=east, align=right,text width=3.5cm},
nodes near coords,
every node near coord/.append style={/pgf/number format/precision=3
},
xticklabel style={inner sep=0pt,anchor=north east,rotate=25, align=right}
%nodes near coords align={horizontal},
%legend style={ anchor=north east, legend cell align=left, align=left, draw=white!15!black, %nodes={scale=0.98}%at={(0.02,0.6)},
%}
]

\addplot[fill=blue!30, draw=blue, nodes near coords style={color=blue}] coordinates {
	({Max-$J_{1/2,\min}$},5.454)
	({Max-$d_{\min}$},4.1809) %+- (0004.840391362472,0)
	({Alternating Opt.},4.0832) %+- (0002.635243995103,0)
	({Precoding},4.5547) %+- (0037.611235654106,0)
	({Partitioning},4.4433) %+- (0005.113151754871,0)
		({Max-$e_{\min}$},4.18) %+- (0004.943069922073,0)
	({Min-$m_1$},3.8512) %+- (0003.671752066412,0)
	({Min-$m_2$},5.9683)
	({Pilot-Based},4.1600)
};
%\addlegendentry{Correlated fading}
%
%\addplot[fill=red!30, draw=red, nodes near coords style={color=red}] coordinates {
%	(1261.015342000001,Exact Marginalization) %+- (0014.167520097206,0)
%	(0086.315282000000,EP) %+- (0007.888196338832,0)
%	(0073.2118220,{EPAK, $t_0=0$}) %+- (0006.491571935403,0)
%	(0073.711169000,{EPAK, $t_0=2$}) %+- (0005.849913196009,0)
%	(0005.59252900,MMSE-SIA) %+- (0002.669959467617,0)
%	(0026.24637900,POCIS) %+- (0002.140511977281,0)
%};
%\addlegendentry{Uncorrelated fading}

\end{axis}
\end{tikzpicture}

%% file: fig/metric_sumDet_T5B4K2M2.tex
\begin{tikzpicture}[scale=.72]
\begin{axis}[
ybar = .2,
ymin=26,
ymax=300,
width=5in,
height=2.5in,
enlarge y limits = 0.12,
bar width=20pt,
ylabel={$m_2(\Xc)$},
symbolic x coords={{Max-$J_{1/2,\min}$},{Max-$d_{\min}$},{Alternating Opt.},{Precoding},{Partitioning},{Max-$e_{\min}$},{Min-$m_1$},{Min-$m_2$},{Pilot-Based}},
xticklabels={{\textbf{Max-$J_{1/2,\min}$}},{\textbf{Max-$d_{\min}$}},{\textbf{Alternating Opt.}},{\textbf{Type-II~Precoding}},{\textbf{Partitioning}},{{Max-$e_{\min}$}},{Min-$m_1$},{Min-$m_2$},{Pilot-Based}},
ytick={25,100,200,300},
xtick=data,
%x tick label style={rotate=45, anchor=east, align=right,text width=3.5cm},
nodes near coords,
every node near coord/.append style={/pgf/number format/precision=3
},
xticklabel style={inner sep=0pt,anchor=north east,rotate=25, align=right}
%nodes near coords align={horizontal},
%legend style={ anchor=north east, legend cell align=left, align=left, draw=white!15!black, %nodes={scale=0.98}%at={(0.02,0.6)},
%}
]

\addplot[fill=blue!30, draw=blue, nodes near coords style={color=blue}] coordinates {
	({Max-$J_{1/2,\min}$},43.73)
	({Max-$d_{\min}$},286.2) %+- (0004.840391362472,0)
	({Alternating Opt.},286.78) %+- (0002.635243995103,0)
	({Precoding},62.68) %+- (0037.611235654106,0)
	({Partitioning},51.03) %+- (0005.113151754871,0)
	({Max-$e_{\min}$},304.52) %+- (0004.943069922073,0)
	({Min-$m_1$},61.13) %+- (0003.671752066412,0)
	({Min-$m_2$},25.05) % 35.587061936517763)
	({Pilot-Based},26.94)
};
%\addlegendentry{Correlated fading}
%
%\addplot[fill=red!30, draw=red, nodes near coords style={color=red}] coordinates {
%	(1261.015342000001,Exact Marginalization) %+- (0014.167520097206,0)
%	(0086.315282000000,EP) %+- (0007.888196338832,0)
%	(0073.2118220,{EPAK, $t_0=0$}) %+- (0006.491571935403,0)
%	(0073.711169000,{EPAK, $t_0=2$}) %+- (0005.849913196009,0)
%	(0005.59252900,MMSE-SIA) %+- (0002.669959467617,0)
%	(0026.24637900,POCIS) %+- (0002.140511977281,0)
%};
%\addlegendentry{Uncorrelated fading}

\end{axis}
\end{tikzpicture}

%% file: fig/joint_SER_T7B3K3M2N6.tex
\begin{tikzpicture}[scale=.75,style={mark size=3pt,line width=3pt}]

\begin{axis}[%
width=4in,
height=3in,
at={(0.758in,0.481in)},
scale only axis,
xmin=6,
xmax=16,
xtick={6, 8, 10, 12, 14, 16},
xlabel style={font=\color{white!15!black}},
xlabel={SNR (dB)},
ymode=log,
ymin=1e-07,
ymax=1e-1,
yminorticks=true,
ylabel style={font=\color{white!15!black}},
ylabel={Joint Symbol Error Rate},
axis background/.style={fill=white},
xmajorgrids,
ymajorgrids,
yminorgrids,
legend style={at={(0.02,0.001)}, anchor=south west, legend cell align=left, align=left, fill=white, fill opacity=0.6, text opacity=1, draw=none, %draw=white!15!black, 
	nodes={scale=1}}
]

\addplot [color=red, line width=1pt, mark=diamond, mark size=4pt, mark options={solid, red}]
table[row sep=crcr]{%
	4.000000000000000   0.114000000000000 \\
	6.000000000000000   0.030600000000000 \\
	8.000000000000000   0.005000000000000 \\
	10.000000000000000   0.000630000000000 \\
	12.000000000000000   0.000071000000000 \\
	14.000000000000000   0.000005000000000 \\
	16 3e-7 \\
};
\addlegendentry{ \textbf{Max-$J_{1/2,\min}$}}

\addplot [color=blue, dashdotted, line width=1pt, mark=+, mark size=4pt, mark options={solid, blue}]
table[row sep=crcr]{%
	4.000000000000000   0.111800000000000 \\
	6.000000000000000   0.029800000000000 \\
	8.000000000000000   0.004800000000000 \\
	10.000000000000000   0.000760000000000 \\
	12.000000000000000   0.000080000000000 \\
	14.000000000000000   0.000006000000000 \\
	16 4.5e-7 \\
};
\addlegendentry{ \textbf{Max-$d_{\min}$}}

\addplot [color=blue, dashdotted, line width=1pt, mark=pentagon, mark size=4pt, mark options={solid, blue}]
table[row sep=crcr]{%
	4.000000000000000   0.202200000000000 \\
	6.000000000000000   0.065100000000000 \\
	8.000000000000000   0.012300000000000 \\
	10.000000000000000   0.002040000000000 \\
	12.000000000000000   0.000236000000000 \\
	14.000000000000000   0.000021000000000 \\
	16 9.4e-7 \\
};
\addlegendentry{\textbf{Alternating optimization of $d_{\min}$}}

\addplot [color=violet, dashed, line width=1pt, mark=triangle, mark size=4pt, mark options={solid, violet}]
table[row sep=crcr]{%
	4.000000000000000   0.133100000000000 \\
	6.000000000000000   0.056400000000000 \\
	8.000000000000000   0.024400000000000 \\
	10.000000000000000   0.011200000000000 \\
	12.000000000000000   0.004896000000000 \\
	14.000000000000000   0.001806000000000 \\
	16 4.910e-04 \\
};
\addlegendentry{\textbf{Type-I Precoding}}

\addplot [color=violet, dashed, line width=1pt, mark=triangle, mark size=4pt, mark options={solid, violet,rotate=180}]
table[row sep=crcr]{%
	4.000000000000000   0.023400000000000 \\
	6.000000000000000   0.005200000000000 \\
	8.000000000000000   0.001100000000000 \\
	10.000000000000000   0.000350000000000 \\
	12.000000000000000   0.000114000000000 \\
	14.000000000000000   0.000047500000000 \\
	16 2.360e-05 \\
};
\addlegendentry{\textbf{Partitioning}}

\addplot [color=pimentgreen, dashdotted, line width=1pt, mark=x, mark size=4pt, mark options={solid, pimentgreen}]
table[row sep=crcr]{%
	4.000000000000000   0.109300000000000 \\
	6.000000000000000   0.029400000000000 \\
	8.000000000000000   0.004500000000000 \\
	10.000000000000000   0.000710000000000 \\
	12.000000000000000   0.000080000000000 \\
	14.000000000000000   0.000007000000000 \\
	16 3.9e-7 \\
};
\addlegendentry{ {Max-$e_{\min}$}}

\addplot [color=pimentgreen, dashdotted, line width=1pt, mark=o, mark size=3pt, mark options={solid, pimentgreen}]
table[row sep=crcr]{%
	4.000000000000000   0.361000000000000 \\
	6.000000000000000   0.158000000000000 \\
	8.000000000000000   0.046400000000000 \\
	10.000000000000000   0.008590000000000 \\
	12.000000000000000   0.001166000000000 \\
	14.000000000000000   0.000104000000000 \\
	16 8.000e-06 \\
};
\addlegendentry{ Min-$m_1$}

\addplot [color=pimentgreen, dashdotted, line width=1pt, mark=square, mark size=3pt, mark options={solid, pimentgreen}]
table[row sep=crcr]{%
	4.000000000000000   0.021800000000000 \\
	6.000000000000000   0.006300000000000 \\
	8.000000000000000   0.001300000000000 \\
	10.000000000000000   0.000320000000000 \\
	12.000000000000000   0.000087000000000 \\
	14.000000000000000   0.000019000000000 \\
	16 4.2e-6 \\
};
\addlegendentry{ Min-$m_2$}

\addplot [color=black, dotted, line width=1pt, mark=*, mark size=2pt, mark options={solid, black}]
table[row sep=crcr]{%
	4.000000000000000   0.291400000000000 \\
	6.000000000000000   0.127100000000000 \\
	8.000000000000000   0.033900000000000 \\
	10.000000000000000   0.008360000000000 \\
	12.000000000000000   0.001321000000000 \\
	14.000000000000000   0.000189000000000 \\
	16  1.800e-05 \\
};
\addlegendentry{ Pilot-based scheme}

%\draw [black] (axis cs:1,1) ellipse [x radius=0.1825, y radius=0.809524];
-\end{axis}
\vspace{-.2cm}
%\begin{axis}[%
%width=4.2in,
%height=2.95in,
%at={(0.758in,0.481in)},
%scale only axis,
%xmin=0,
%xmax=1,
%ymin=0,
%ymax=1,
%axis line style={draw=none},
%ticks=none,
%axis x line*=bottom,
%axis y line*=left
%]
%\draw [black] (axis cs:0.40,0.745) ellipse [x radius=3, y radius=11];
%\draw [-latex,thick] (axis cs:.5,.9) node [right] {$R_1 = R_2 = B/T, P_1 = P_2 = P$} -- (axis cs:.415,.84);
%\draw [-latex,thick] (axis cs:.17,.625) node [below,text width=3.2cm] {$R_k \!=\! 2B/T, R_{3-k}\!=\! 0$, $P_k \!=\! 2P, P_{3-k} \!=\! 0$} -- (axis cs:.19,.65);
%\end{axis}
\end{tikzpicture}%

%% file: fig/metric_beta_T7B3K3M2.tex
\begin{tikzpicture}[scale=.7,style={mark size=3pt,line width=3pt},spy using outlines={mark size = 2pt, rectangle,lens={scale=3.2}, height=2cm,width=3cm, connect spies}]
\begin{axis}[%
width=4.5in,
height=3in,
at={(0.758in,0.481in)},
scale only axis,
xmin=4,
xmax=32,
xtick={ 4,  8, 12, 16, 20, 24, 28, 32},
xlabel style={font=\color{white!15!black}},
xlabel={SNR (dB)},
%ymode=log,
ymin=1,
ymax=16,
yminorticks=true,
ytick={1,4,8,12,16},
ylabel style={font=\color{white!15!black}},
ylabel={\large $b_{\min}(\Xc)$},
axis background/.style={fill=white},
xmajorgrids,
ymajorgrids,
yminorgrids,
legend style={at={(0.02,1)}, anchor=north west, legend cell align=left, align=left, fill=white, fill opacity=0.6, text opacity=1, draw=none %, nodes={scale=1.1}%draw=white!15!black, 
	}
]

\addplot [color=red,  line width=1pt, mark=diamond, mark size=4pt, mark options={solid, red}]
table[row sep=crcr]{%
	4 3.0844 \\
	8 4.4457 \\
	12 6.0343 \\
	16 7.7599 \\
	20 9.5524 \\
	24 11.3742 \\
	28 13.2081 \\
	32 15.047 \\
};
\addlegendentry{ \textbf{Max-$J_{1/2,\min}$}}

\addplot [color=blue, dashdotted, line width=1pt, mark=+, mark size=4pt, mark options={solid, blue}]
table[row sep=crcr]{%
	4 3.0842 \\
	8 4.4456 \\
	12 6.0343 \\
	16 7.7599 \\
	20 9.5525 \\
	24 11.3743 \\
	28 13.2082 \\
	32 15.047 \\
};
\addlegendentry{ \textbf{Max-$d_{\min}$}}

\addplot [color=blue, dashdotted, line width=1pt, mark=pentagon, mark size=4pt, mark options={solid, blue}]
table[row sep=crcr]{%
	4 2.7636 \\
	8 4.0863 \\
	12 5.6507 \\
	16 7.3636 \\
	20 9.1504 \\
	24 10.9697 \\
	28 12.8027 \\
	32 14.6411 \\
};
\addlegendentry{\textbf{Alternating optimization of $d_{\min}$}}

\addplot [color=violet, dashed, line width=1pt, mark=triangle, mark size=4pt, mark options={solid, violet}]
table[row sep=crcr]{%
	4 1.9708 \\
	8 2.46 \\
	12 3.11 \\
	16 4.0666 \\
	20 5.3691 \\
	24 6.9364 \\
	28 8.6548 \\
	32 10.4448 \\
};
\addlegendentry{\textbf{Type-I Precoding}}

\addplot [color=violet, dashed, line width=1pt, mark=triangle, mark size=4pt, mark options={solid, violet,rotate=180}]
table[row sep=crcr]{%
	4 4.018 \\
	8 5.0787 \\
	12 6.1016 \\
	16 7.2862 \\
	20 8.3184 \\
	24 9.6435 \\
	28 11.2171 \\
	32 12.9373 \\
};
\addlegendentry{\textbf{Partitioning}}

\addplot [color=pimentgreen, dashdotted, line width=1pt, mark=x, mark size=4pt, mark options={solid, pimentgreen}]
table[row sep=crcr]{%
	4 3.0842 \\
	8 4.4456 \\
	12 6.0343 \\
	16 7.7599 \\
	20 9.5525 \\
	24 11.3743 \\
	28 13.2082 \\
	32 15.047 \\
};
\addlegendentry{ {Max-$e_{\min}$}}

\addplot [color=pimentgreen, dashdotted, line width=1pt, mark=o, mark size=3pt, mark options={solid, pimentgreen}]
table[row sep=crcr]{%
	4 2.3686 \\
	8 3.5802 \\
	12 5.064 \\
	16 6.7318 \\
	20 8.4975 \\
	24 10.3078 \\
	28 12.1371 \\
	32 13.974 \\
};
\addlegendentry{ Min-$m_1$}

\addplot [color=pimentgreen, dashdotted, line width=1pt, mark=square, mark size=3pt, mark options={solid, pimentgreen}]
table[row sep=crcr]{%
	4 3.9563 \\
	8 4.6828 \\
	12 5.3714 \\
	16 6.2447 \\
	20 7.4411 \\
	24 8.9344 \\
	28 10.6143 \\
	32 12.3871 \\
};
\addlegendentry{ Min-$m_2$}

\addplot [color=black, dotted, line width=1pt, mark=*, mark size=2pt, mark options={solid, black}]
table[row sep=crcr]{%
	4 2.0761 \\
	8 3.1654 \\
	12 4.5523 \\
	16 6.1617 \\
	20 7.8985 \\
	24 9.6963 \\
	28 11.5203 \\
	32 13.3551 \\
};
\addlegendentry{ Pilot-based scheme}

\end{axis}
%\spy [black] on (11.2,7.45)
%in node [left] at (19,3);
\end{tikzpicture}%

%% file: fig/metric_Js_T7B3K3M2.tex
\begin{tikzpicture}[scale=.7,style={mark size=3pt,line width=3pt},spy using outlines={mark size = 2pt, rectangle,lens={scale=3.2}, height=2cm,width=3cm, connect spies}]
\begin{axis}[%
width=4.5in,
height=3in,
at={(0.758in,0.481in)},
scale only axis,
xmin=4,
xmax=32,
xtick={ 4,  8, 12, 16, 20, 24, 28, 32},
xlabel style={font=\color{white!15!black}},
xlabel={SNR (dB)},
%ymode=log,
ymin=0,
ymax=7,
yminorticks=true,
ylabel style={font=\color{white!15!black}},
ylabel={\large $J_{1/2, \min}$},
axis background/.style={fill=white},
xmajorgrids,
ymajorgrids,
yminorgrids,
legend style={at={(0.02,1)}, anchor=north west, legend cell align=left, align=left, fill=white, fill opacity=0.6, text opacity=1, draw=none %, nodes={scale=1.1}%draw=white!15!black, 
	}
]

\addplot [color=red,  line width=1pt, mark=diamond, mark size=4pt, mark options={solid, red}]
table[row sep=crcr]{%
	4 0.54358 \\
	8 1.0422 \\
	12 1.7264 \\
	16 2.5345 \\
	20 3.4067 \\
	24 4.3076 \\
	28 5.2204 \\
	32 6.1382 \\
};
\addlegendentry{ \textbf{Max-$J_{1/2,\min}$}}

\addplot [color=blue, dashdotted, line width=1pt, mark=+, mark size=4pt, mark options={solid, blue}]
table[row sep=crcr]{%
	4 0.54353 \\
	8 1.0422 \\
	12 1.7264 \\
	16 2.5345 \\
	20 3.4067 \\
	24 4.3076 \\
	28 5.2205 \\
	32 6.1383 \\
};
\addlegendentry{ \textbf{Max-$d_{\min}$}}

\addplot [color=blue, dashdotted, line width=1pt, mark=pentagon, mark size=4pt, mark options={solid, blue}]
table[row sep=crcr]{%
	4 0.4436 \\
	8 0.9006 \\
	12 1.5542 \\
	16 2.3452 \\
	20 3.2094 \\
	24 4.1069 \\
	28 5.0183 \\
	32 5.9356 \\
};
\addlegendentry{\textbf{Alternating optimization of $d_{\min}$}}

\addplot [color=violet, dashed, line width=1pt, mark=triangle, mark size=4pt, mark options={solid, violet}]
table[row sep=crcr]{%
	4 0.18348 \\
	8 0.2595 \\
	12 0.38189 \\
	16 0.61768 \\
	20 1.036 \\
	24 1.6547 \\
	28 2.4241 \\
	32 3.2777 \\
};
\addlegendentry{\textbf{Type-I Precoding}}

\addplot [color=violet, dashed, line width=1pt, mark=triangle, mark size=4pt, mark options={solid, violet,rotate=180}]
table[row sep=crcr]{%
	4 0.3632 \\
	8 0.46291 \\
	12 0.55027 \\
	16 0.67704 \\
	20 0.91415 \\
	24 1.3326 \\
	28 1.9509 \\
	32 2.72 \\
};
\addlegendentry{\textbf{Partitioning}}

\addplot [color=pimentgreen, dashdotted, line width=1pt, mark=x, mark size=4pt, mark options={solid, pimentgreen}]
table[row sep=crcr]{%
	4 0.54353 \\
	8 1.0422 \\
	12 1.7264 \\
	16 2.5345 \\
	20 3.4067 \\
	24 4.3076 \\
	28 5.2205 \\
	32 6.1383 \\
};
\addlegendentry{ \textbf{Max-$e_{\min}$}}

\addplot [color=pimentgreen, dashdotted, line width=1pt, mark=o, mark size=3pt, mark options={solid, pimentgreen}]
table[row sep=crcr]{%
	4 0.33189 \\
	8 0.71257 \\
	12 1.2987 \\
	16 2.0475 \\
	20 2.8908 \\
	24 3.7791 \\
	28 4.6869 \\
	32 5.6025 \\
};
\addlegendentry{ Min-$m_1$}

\addplot [color=pimentgreen, dashdotted, line width=1pt, mark=square, mark size=3pt, mark options={solid, pimentgreen}]
table[row sep=crcr]{%
	4 0.44555 \\
	8 0.64595 \\
	12 0.91488 \\
	16 1.2131 \\
	20 1.5689 \\
	24 2.0341 \\
	28 2.6088 \\
	32 3.3476 \\
};
\addlegendentry{ Min-$m_2$}

\addplot [color=black, dotted, line width=1pt, mark=*, mark size=2pt, mark options={solid, black}]
table[row sep=crcr]{%
	4 0.25808 \\
	8 0.57007 \\
	12 1.0854 \\
	16 1.7843 \\
	20 2.6011 \\
	24 3.4775 \\
	28 4.3801 \\
	32 5.2938 \\
};
\addlegendentry{ Pilot-based scheme}

\end{axis}
%\spy [black] on (11.2,7.45)
%in node [left] at (19,3);
\end{tikzpicture}%

%% file: fig/metric_emin_dmin_T7B3K3M2.tex
\begin{tikzpicture}[scale=.7,style={mark size=3pt,line width=3pt},spy using outlines={mark size = 2pt, rectangle,lens={scale=3.2}, height=1.8cm,width=3cm, connect spies}]
\begin{axis}[%
width=4.5in,
height=3in,
at={(0.758in,0.481in)},
scale only axis,
xmin=4,
xmax=32,
xtick={ 4,  8, 12, 16, 20, 24, 28, 32},
xlabel style={font=\color{white!15!black}},
xlabel={SNR (dB)},
ymode=log,
ymin=1,
ymax=3000,
yminorticks=true,
ylabel style={font=\color{white!15!black}},
ylabel={\large $e_{\min}(\Xc)$ (lines) and $d_{\min}(\Xc)$ (markers)},
axis background/.style={fill=white},
xmajorgrids,
ymajorgrids,
yminorgrids,
legend style={at={(0.02,1)}, anchor=north west, legend cell align=left, align=left, fill=white, fill opacity=0.6, text opacity=1, draw=none %, nodes={scale=1.1}%draw=white!15!black, 
	}
]

\addplot [color=red,  line width=1pt, mark=diamond, mark size=4pt, mark options={solid, red}]
table[row sep=crcr]{%
	0 0 \\
	2 .1 \\
};
\addlegendentry{ \textbf{Max-$J_{1/2,\min}$}}

\addplot [color=blue, dashdotted, line width=1pt, mark=+, mark size=4pt, mark options={solid, blue}]
table[row sep=crcr]{%
	0 0 \\
	2 .1 \\
};
\addlegendentry{ \textbf{Max-$d_{\min}$}}

\addplot [color=blue, dashdotted, line width=1pt, mark=pentagon, mark size=4pt, mark options={solid, blue}]
table[row sep=crcr]{%
	0 0 \\
	2 .1 \\
};
\addlegendentry{\textbf{Alternating optimization of $d_{\min}$}}

\addplot [color=violet, dashed, line width=1pt, mark=triangle, mark size=4pt, mark options={solid, violet}]
table[row sep=crcr]{%
	0 0 \\
	2 .1 \\
};
\addlegendentry{\textbf{Type-I Precoding}}

\addplot [color=violet, dashed, line width=1pt, mark=triangle, mark size=4pt, mark options={solid, violet,rotate=180}]
table[row sep=crcr]{%
	0 0 \\
	2 .1 \\
};
\addlegendentry{\textbf{Partitioning}}

\addplot [color=pimentgreen, dashdotted, line width=1pt, mark=x, mark size=4pt, mark options={solid, pimentgreen}]
table[row sep=crcr]{%
	0 0 \\
	2 .1 \\
};
\addlegendentry{ {Max-$e_{\min}$}}

\addplot [color=pimentgreen, dashdotted, line width=1pt, mark=o, mark size=3pt, mark options={solid, pimentgreen}]
table[row sep=crcr]{%
	0 0 \\
	2 .1 \\
};
\addlegendentry{ Min-$m_1$}

\addplot [color=pimentgreen, dashdotted, line width=1pt, mark=square, mark size=3pt, mark options={solid, pimentgreen}]
table[row sep=crcr]{%
	0 0 \\
	2 .1 \\
};
\addlegendentry{ Min-$m_2$}

\addplot [color=black, dotted, line width=1pt, mark=*, mark size=2pt, mark options={solid, black}]
table[row sep=crcr]{%
	0 0 \\
	2 .1 \\
};
\addlegendentry{ Pilot-based scheme}

%--------------------

\addplot [color=red,  line width=1pt]
table[row sep=crcr]{%
	4 9.8886 \\
	8 14.3414 \\
	12 25.4813 \\
	16 53.4395 \\
	20 123.6568 \\
	24 300.0302 \\
	28 743.0584 \\
	32 1855.8942 \\
};
%\addlegendentry{ \textbf{Max-$J_{1/2,\min}$}}

\addplot [color=pimentgreen, dashdotted, line width=1pt]
table[row sep=crcr]{%
	4 9.8883 \\
	8 14.3413 \\
	12 25.482 \\
	16 53.4425 \\
	20 123.6654 \\
	24 300.053 \\
	28 743.1168 \\
	32 1856.0419 \\
};
%\addlegendentry{ \textbf{Max-$e_{\min}$}}

\addplot [color=blue, dashdotted, line width=1pt]
table[row sep=crcr]{%
	4 9.8883 \\
	8 14.3413 \\
	12 25.482 \\
	16 53.4425 \\
	20 123.6654 \\
	24 300.053 \\
	28 743.1166 \\
	32 1856.0416 \\
};
%\addlegendentry{ \textbf{Max-$d_{\min}$}}

\addplot [color=blue, dashdotted, line width=1pt]
table[row sep=crcr]{%
	4 9.2333 \\
	8 12.8443 \\
	12 21.9259 \\
	16 44.7421 \\
	20 102.0555 \\
	24 246.0211 \\
	28 607.6464 \\
	32 1516.0082 \\
};
%\addlegendentry{\textbf{Alternating optimization of $d_{\min}$}}

\addplot [color=violet, dashed, line width=1pt]
table[row sep=crcr]{%
	4 7.7139 \\
	8 8.0385 \\
	12 8.6323 \\
	16 10.0078 \\
	20 13.3768 \\
	24 21.8191 \\
	28 43.017 \\
	32 96.2604 \\
};
%\addlegendentry{\textbf{Type-II Precoding}}

\addplot [color=violet, dashed, line width=1pt]
table[row sep=crcr]{%
	4 8.5763 \\
	8 9.0662 \\
	12 9.5607 \\
	16 10.3703 \\
	20 12.0525 \\
	24 16.1805 \\
	28 26.5096 \\
	32 52.4387 \\
};
%\addlegendentry{\textbf{Partitioning}}

\addplot [color=pimentgreen, dashdotted,line width=1pt]
table[row sep=crcr]{%
	4 8.5744 \\
	8 11.1569 \\
	12 17.6581 \\
	16 33.9945 \\
	20 75.0321 \\
	24 178.1151 \\
	28 437.0482 \\
	32 1087.459 \\
};
%\addlegendentry{ Min-$m_1$}

\addplot [color=pimentgreen, dashdotted,line width=1pt]
table[row sep=crcr]{%
	4 8.9938 \\
	8 9.8457 \\
	12 10.6868 \\
	16 11.7967 \\
	20 13.9708 \\
	24 18.9015 \\
	28 26.1181 \\
	32 44.1425 \\
};
%\addlegendentry{ Min-$m_2$}

\addplot [color=black, dotted, line width=1pt]
table[row sep=crcr]{%
	4 8.1778 \\
	8 10.0736 \\
	12 14.842 \\
	16 26.8224 \\
	20 56.9168 \\
	24 132.511 \\
	28 322.3954 \\
	32 799.3633 \\
};
%\addlegendentry{ Pilot-based scheme}

%------------------------------

\addplot [color=red, dashed, line width=1pt, mark=diamond, mark size=4pt, mark options={solid, red}, only marks]
table[row sep=crcr]{%
	4 8.1574 \\
	8 13.023 \\
	12 24.3495 \\
	16 52.3862 \\
	20 122.6354 \\
	24 299.0217 \\
	28 742.055 \\
	32 1854.8928 \\
};
%\addlegendentry{ \textbf{Max-$J_{1/2,\min}$}}

\addplot [color=pimentgreen, dashdotted, line width=1pt, mark=x, mark size=4pt, mark options={solid, pimentgreen}, only marks]
table[row sep=crcr]{%
	4 8.1574 \\
	8 13.0231 \\
	12 24.3504 \\
	16 52.3892 \\
	20 122.6441 \\
	24 299.0445 \\
	28 742.1134 \\
	32 1855.0405 \\
};
%\addlegendentry{ \textbf{Max-$e_{\min}$}}

\addplot [color=blue, dashdotted, line width=1pt, mark=+, mark size=4pt, mark options={solid, blue}, only marks]
table[row sep=crcr]{%
	4 8.1574 \\
	8 13.0231 \\
	12 24.3503 \\
	16 52.3892 \\
	20 122.6441 \\
	24 299.0444 \\
	28 742.1133 \\
	32 1855.0403 \\
};
%\addlegendentry{ \textbf{Max-$d_{\min}$}}

\addplot [color=blue, dashdotted, line width=1pt, mark=pentagon, mark size=4pt, mark options={solid, blue}, only marks]
table[row sep=crcr]{%
	4 7.5951 \\
	8 11.5722 \\
	12 20.8144 \\
	16 43.6972 \\
	20 101.0376 \\
	24 245.0139 \\
	28 606.6435 \\
	32 1515.0071 \\
};
%\addlegendentry{\textbf{Alternating optimization of $d_{\min}$}}

\addplot [color=violet, dashed, line width=1pt, mark=triangle, mark size=4pt, mark options={solid, violet}, only marks]
table[row sep=crcr]{%
	4 4.888 \\
	8 5.6148 \\
	12 6.6056 \\
	16 8.4486 \\
	20 12.081 \\
	24 20.5803 \\
	28 41.8023 \\
	32 95.0555 \\
};
%\addlegendentry{\textbf{Type-II Precoding}}

\addplot [color=violet, dashed, line width=1pt, mark=triangle, mark size=4pt, mark options={solid, violet,rotate=180}, only marks]
table[row sep=crcr]{%
	4 6.4241 \\
	8 7.3694 \\
	12 8.0849 \\
	16 9.0043 \\
	20 10.9152 \\
	24 15.216 \\
	28 25.5545 \\
	32 51.4873 \\
};
%\addlegendentry{\textbf{Partitioning}}

\addplot [color=pimentgreen, dashdotted, line width=1.5pt, mark=o, mark size=3pt, mark options={solid, pimentgreen}, only marks]
table[row sep=crcr]{%
	4 6.9532 \\
	8 9.8929 \\
	12 16.5501 \\
	16 32.951 \\
	20 74.0148 \\
	24 177.1082 \\
	28 436.0455 \\
	32 1086.4579 \\
};
%\addlegendentry{ Min-$m_1$}

\addplot [color=pimentgreen, dashdotted, line width=1.5pt, mark=square, mark size=3pt, mark options={solid, pimentgreen}, only marks]
table[row sep=crcr]{%
	4 6.5289 \\
	8 7.5966 \\
	12 8.4377 \\
	16 9.4893 \\
	20 11.6204 \\
	24 16.7524 \\
	28 24.3835 \\
	32 42.4091 \\
};
%\addlegendentry{ Min-$m_2$}

\addplot [color=black, dotted, line width=1pt, mark=*, mark size=2pt, mark options={solid, black}, only marks]
table[row sep=crcr]{%
	4 6.5298 \\
	8 8.797 \\
	12 13.7286 \\
	16 25.7767 \\
	20 55.8985 \\
	24 131.5038 \\
	28 321.3925 \\
	32 798.3621 \\
};
%\addlegendentry{ Pilot-based scheme}

\end{axis}
%\spy [black] on (11.2,6.8)
%in node [left] at (19,3);
\end{tikzpicture}%

%% file: fig/metric_minChordal_T7B3K3M2.tex
\begin{tikzpicture}[scale=.7]
\begin{axis}[
ybar = .2,
ymin=5,
ymax=8,
width=5in,
height=2.5in,
enlarge y limits = 0.12,
bar width=20pt,
ylabel={$m_1(\Xc)$},
symbolic x coords={{Max-$J_{1/2,\min}$},{Max-$d_{\min}$},{Alternating Opt.},{Precoding},{Partitioning},{Max-$e_{\min}$},{Min-$m_1$},{Min-$m_2$},{Pilot-Based}},
xticklabels={{\textbf{Max-$J_{1/2,\min}$}},{\textbf{Max-$d_{\min}$}},{\textbf{Alternating Opt.}},{\textbf{Type-I~Precoding}},{\textbf{Partitioning}},{Max-$e_{\min}$},{Min-$m_1$},{Min-$m_2$},{Pilot-Based}},
ytick={5,6,7,8},
xtick=data,
%x tick label style={rotate=45, anchor=east, align=right,text width=3.5cm},
nodes near coords,
every node near coord/.append style={/pgf/number format/precision=3
},
xticklabel style={inner sep=0pt,anchor=north east,rotate=25, align=right}
%nodes near coords align={horizontal},
%legend style={ anchor=north east, legend cell align=left, align=left, draw=white!15!black, %nodes={scale=0.98}%at={(0.02,0.6)},
%}
]

\addplot[fill=blue!30, draw=blue, nodes near coords style={color=blue}] coordinates {
	({Max-$J_{1/2,\min}$},6.837)
	({Max-$e_{\min}$},6.8341) %+- (0004.943069922073,0)
	({Max-$d_{\min}$},6.8341) %+- (0004.840391362472,0)
	({Alternating Opt.},5.9432) %+- (0002.635243995103,0)
	({Precoding},7.7315) %+- (0037.611235654106,0)
	({Partitioning},7.5652) %+- (0005.113151754871,0)
	({Min-$m_1$},5.8636) %+- (0003.671752066412,0)
	({Min-$m_2$},7.5802)
	({Pilot-Based},6.3673)
};
%\addlegendentry{Correlated fading}
%
%\addplot[fill=red!30, draw=red, nodes near coords style={color=red}] coordinates {
%	(1261.015342000001,Exact Marginalization) %+- (0014.167520097206,0)
%	(0086.315282000000,EP) %+- (0007.888196338832,0)
%	(0073.2118220,{EPAK, $t_0=0$}) %+- (0006.491571935403,0)
%	(0073.711169000,{EPAK, $t_0=2$}) %+- (0005.849913196009,0)
%	(0005.59252900,MMSE-SIA) %+- (0002.669959467617,0)
%	(0026.24637900,POCIS) %+- (0002.140511977281,0)
%};
%\addlegendentry{Uncorrelated fading}

\end{axis}
\end{tikzpicture}

%% file: fig/metric_sumDet_T7B3K3M2.tex
\begin{tikzpicture}[scale=.7]
\begin{axis}[
ybar = .2,
ymin=50,
ymax=470,
width=5in,
height=2.5in,
enlarge y limits = 0.12,
bar width=20pt,
ylabel={$m_2(\Xc)$},
symbolic x coords={{Max-$J_{1/2,\min}$},{Max-$d_{\min}$},{Alternating Opt.},{Precoding},{Partitioning},{Max-$e_{\min}$},{Min-$m_1$},{Min-$m_2$},{Pilot-Based}},
xticklabels={{\textbf{Max-$J_{1/2,\min}$}},{{Max-$e_{\min}$}},{\textbf{Max-$d_{\min}$}},{\textbf{Alternating Opt.}},{\textbf{Type-II~Precoding}},{\textbf{Partitioning}},{Min-$m_1$},{Min-$m_2$},{Pilot-Based}},
ytick={0,100,200, 300, 400},
xtick=data,
%x tick label style={rotate=45, anchor=east, align=right,text width=3.5cm},
nodes near coords,
every node near coord/.append style={/pgf/number format/precision=3
},
xticklabel style={inner sep=0pt,anchor=north east,rotate=25, align=right}
%nodes near coords align={horizontal},
%legend style={ anchor=north east, legend cell align=left, align=left, draw=white!15!black, %nodes={scale=0.98}%at={(0.02,0.6)},
%}
]

\addplot[fill=blue!30, draw=blue, nodes near coords style={color=blue}] coordinates {
	({Max-$J_{1/2,\min}$},258.0344)
	({Max-$e_{\min}$},462.4933) %+- (0004.943069922073,0)
	({Max-$d_{\min}$},420.6368) %+- (0004.840391362472,0)
	({Alternating Opt.},471.6821) %+- (0002.635243995103,0)
	({Precoding},115.6173) %+- (0037.611235654106,0)
	({Partitioning},92.2106) %+- (0005.113151754871,0)
	({Min-$m_1$},370.4266) %+- (0003.671752066412,0)
	({Min-$m_2$},56.37) % 35.587061936517763)
	({Pilot-Based},57.8164)
};
%\addlegendentry{Correlated fading}
%
%\addplot[fill=red!30, draw=red, nodes near coords style={color=red}] coordinates {
%	(1261.015342000001,Exact Marginalization) %+- (0014.167520097206,0)
%	(0086.315282000000,EP) %+- (0007.888196338832,0)
%	(0073.2118220,{EPAK, $t_0=0$}) %+- (0006.491571935403,0)
%	(0073.711169000,{EPAK, $t_0=2$}) %+- (0005.849913196009,0)
%	(0005.59252900,MMSE-SIA) %+- (0002.669959467617,0)
%	(0026.24637900,POCIS) %+- (0002.140511977281,0)
%};
%\addlegendentry{Uncorrelated fading}

\end{axis}
\end{tikzpicture}

%% file: fig/T4_1B6_2B2_P30dB_SER_fullvsOptPower.tex
\begin{tikzpicture}[scale=.75,style={mark size=3pt,line width=3pt}]

\begin{axis}[%
width=4in,
height=3in,
at={(0.758in,0.481in)},
scale only axis,
xmin=4,
xmax=20,
xtick = {4,6,8,10,12,14,16,18,20},
xlabel style={font=\color{white!15!black}},
xlabel={SNR (dB)},
ymode=log,
ymin=5e-06,
ymax=1,
yminorticks=true,
ylabel style={font=\color{white!15!black}},
ylabel={\large Joint Symbol Error Rate},
axis background/.style={fill=white},
xmajorgrids,
ymajorgrids,
yminorgrids,
legend style={at={(0.01,0.01)}, anchor=south west, legend cell align=left, align=left, fill=white, fill opacity=0.3, text opacity=1, draw=none, nodes={scale=1.1}%draw=white!15!black, fill=white!98!black
}
]

\addplot [color=blue,dashdotted, line width=1.2pt]
table[row sep=crcr]{%
	4	0.2274\\
	6	0.102\\
	8	0.0427\\
	10	0.0143\\
	12	0.00362\\
	14	0.00092\\
	16	0.00021\\
	18	3.7e-05\\
	20	7.5e-06\\
};
\addlegendentry{\textbf{Max-$d_{\min}$, full power}}

\addplot [color=blue, dashed, line width=1.2pt, draw=none, mark=+, mark size=4pt, mark options={solid, blue}, only marks]
table[row sep=crcr]{%
	4	0.2198\\
	6	0.102\\
	8	0.0427\\
	10	0.0143\\
	12	0.00362\\
	14	0.00092\\
	16	0.00021\\
	18	3.7e-05\\
	20	7.5e-06\\
};
\addlegendentry{\textbf{Max-$d_{\min}$, optimized power}}

\addplot [color=violet, dashed, line width=1.2pt]
  table[row sep=crcr]{%
4	0.296\\
6	0.1496\\
8	0.0814\\
10	0.0326\\
12	0.01282\\
14	0.0041\\
16	0.00124\\
18	0.000305\\
20	6.4e-05\\
};
\addlegendentry{\textbf{Type-II Precoding, full power}}

\addplot [color=violet,line width=1.2pt, draw=none, mark=triangle, mark size=4pt, mark options={solid, violet,rotate=180}, only marks]
  table[row sep=crcr]{%
4	0.3321\\
6	0.1734\\
8	0.0818\\
10	0.0328\\
12	0.01228\\
14	0.00399\\
16	0.00108\\
18	0.000281\\
20	6e-05\\
};
\addlegendentry{\textbf{Type-II Precoding, optimized power}}

\addplot [color=violet, dashed, line width=1.2pt]
  table[row sep=crcr]{%
4	0.2306\\
6	0.1168\\
8	0.0609\\
10	0.0279\\
12	0.01105\\
14	0.00435\\
16	0.00183\\
18	0.000643\\
20	0.000223\\
};
\addlegendentry{\textbf{Partitioning, full power}}

\addplot [color=violet, line width=1.2pt, draw=none, mark=triangle, mark size=4pt, mark options={solid, violet}, only marks]
  table[row sep=crcr]{%
4	0.2306\\
6	0.1168\\
8	0.0604\\
10	0.0253\\
12	0.01074\\
14	0.00424\\
16	0.00186\\
18	0.000606\\
20	0.000222\\
};
\addlegendentry{\textbf{Partitioning, optimized power}}

\addplot [color=black, dotted, line width=1.2pt, mark=*, mark size=2pt, mark options={solid, black}]
  table[row sep=crcr]{%
4	0.4268\\
6	0.2566\\
8	0.1345\\
10	0.058\\
12	0.01997\\
14	0.00585\\
16	0.00126\\
18	0.000268\\
20	6.05e-05\\
};
\addlegendentry{Pilot-based scheme}

%\addplot [color=black, dotted, line width=1.2pt, mark=+, mark size=4pt, mark options={solid, black}]
%  table[row sep=crcr]{%
%4	0.5433\\
%6	0.3868\\
%8	0.2486\\
%10	0.14\\
%12	0.06622\\
%14	0.0269\\
%16	0.00941\\
%18	0.003105\\
%20	0.000951\\
%};
%\addlegendentry{Pilot-based, linear MMSE detector}

\end{axis}
\end{tikzpicture}%

%% file: fig/theta_opt.tex
\begin{tikzpicture}[scale=.75,style={mark size=3pt,line width=3pt}]

\begin{axis}[%
width=4in,
height=3in,
at={(0.758in,0.481in)},
scale only axis,
xmin=4,
xmax=20,
xtick = {4,6,8,10,12,14,16,18,20},
xlabel style={font=\color{white!15!black}},
xlabel={SNR $P$ (dB)},
ymin=0.4,
ymax=1,
ylabel style={font=\color{white!15!black}},
ylabel={The optimal $\theta^*$ for user $2$},
axis background/.style={fill=white},
xmajorgrids,
ymajorgrids,
legend style={at={(.97,.03)}, anchor=south east, legend cell align=left, align=left, fill=white, fill opacity=0.6, text opacity=1, draw=none, nodes={scale=1.1} %draw=white!15!black, fill=white!98!black
}
]

\addplot[color=red, dashdotted, line width=2pt,]
table[row sep=crcr]{%
	4	0.44\\
	6	0.44\\
	8	0.48\\
	10	0.52\\
	12	0.6\\
	14	0.68\\
	16	0.76\\
	18	0.84\\
	20	0.88\\
};
\addlegendentry{\textbf{Type-II Precoding}}

\addplot[color=blue, dashed, line width=2pt,]
  table[row sep=crcr]{%
4	1\\
6	1\\
8	0.98\\
10	0.92\\
12	0.86\\
14	0.78\\
16	0.72\\
18	0.64\\
20	0.58\\
};
\addlegendentry{\textbf{Partitioning}}

\end{axis}
\end{tikzpicture}%